\RequirePackage{luatex85}% FINAL remove
\ifcsname directlua\endcsname\RequirePackage{pdftexcmds}\makeatletter\let\pdfstrcmp\pdf@strcmp\makeatother\fi
\documentclass{lmcs}
\pdfoutput=1

\usepackage[utf8]{inputenc}

\usepackage{lastpage}
\lmcsdoi{18}{1}{6}
\lmcsheading{}{\pageref{LastPage}}{}{}%
{Nov.~30,~2020}{Jan.~12,~2022}{}

\keywords{conditional bisimilarity, reactive systems, up-to context, graph transformation}

\AtBeginDocument{\hypersetup{breaklinks=true,pdfborder={1 1 1},linkbordercolor=[rgb]{1,0.8,0.8},citebordercolor=[rgb]{0.6,1,0.6}}}% final remove

\usepackage{mathtools}
\usepackage{stmaryrd} % needed for \bigsqcap

\usepackage[style=base]{caption}% style=base used to retain the default that is provided by the class (leaving it out would change it to \scshape)
\usepackage{subcaption}
\usepackage{amssymb}

\usepackage{hyperref}
\usepackage[nameinlink]{cleveref}
\newcommand{\renewtheorem}[1]{%
  \expandafter\let\csname #1\endcsname\relax%
  \expandafter\let\csname c@#1\endcsname\relax%
  \expandafter\let\csname end#1\endcsname\relax%
  \newtheorem{#1}%
}
\theoremstyle{plain}
\renewtheorem{thm}{Theorem}[section]
\renewtheorem{cor}[thm]{Corollary}
\renewtheorem{lem}[thm]{Lemma}
\renewtheorem{slem}[thm]{Sublemma}
\renewtheorem{prop}[thm]{Proposition}
\renewtheorem{asm}[thm]{Assumption}

\theoremstyle{definition}

\renewtheorem{rem}[thm]{Remark}
\renewtheorem{rems}[thm]{Remarks}
\renewtheorem{exa}[thm]{Example}
\renewtheorem{exas}[thm]{Examples}
\renewtheorem{defi}[thm]{Definition}
\renewtheorem{conv}[thm]{Convention}
\renewtheorem{conj}[thm]{Conjecture}
\renewtheorem{prob}[thm]{Problem}
\renewtheorem{oprob}[thm]{Open Problem}
\renewtheorem{oprobs}[thm]{Open Problems}
\renewtheorem{algo}[thm]{Algorithm}
\renewtheorem{obs}[thm]{Observation}
\renewtheorem{desc}[thm]{Description}
\renewtheorem{fact}[thm]{Fact}
\renewtheorem{qu}[thm]{Question}
\renewtheorem{oqu}[thm]{Open Question}
\renewtheorem{pty}[thm]{Property}
\renewtheorem{clm}[thm]{Claim}
\renewtheorem{nota}[thm]{Notation}
\renewtheorem{com}[thm]{Comment}
\renewtheorem{coms}[thm]{Comments}
\crefname{thm}{Theorem}{Theorems}
\crefname{defi}{Definition}{Definitions}
\crefname{exa}{Example}{Examples}
\crefname{lem}{Lemma}{Lemmas}

\crefname{subsection}{Subsection}{Subsections}

\theoremstyle{thmC}
\renewtheorem{lemC}[thm]{Lemma}
\crefname{lemC}{Lemma}{Lemmas}

\newenvironment{proofparts}{%
  \quad\par% don't put the list bullet in the middle of the line
  \begin{itemize}%
}{%
  \end{itemize}%
}
\newcommand{\proofPartNoNewline}[1]{\item \textbf{(#1):}}
\newcommand{\proofPart}[1]{\proofPartNoNewline{#1}\\}

\gdef\ExaEndSeen{0}
\newcommand\ExaEndHere{%
  \leavevmode\unskip\penalty9999 \hbox{}\nobreak\hfill\quad\hbox{$\lrcorner$}%
  \gdef\ExaEndSeen{1}%
}
\gdef\RemEndSeen{0}
\newcommand\RemEndHere{%
  \leavevmode\unskip\penalty9999 \hbox{}\nobreak\hfill\quad\hbox{$\lrcorner$}%
  \gdef\RemEndSeen{1}%
}
\AtBeginEnvironment{exa}{\gdef\ExaEndSeen{0}}
\AtBeginEnvironment{rem}{\gdef\RemEndSeen{0}}
\AtEndEnvironment{exa}{\if\ExaEndSeen0\ExaEndHere\fi}
\AtEndEnvironment{rem}{\if\RemEndSeen0\RemEndHere\fi}

\usepackage{mauflpic}

\def\defeq{:=}
\def\dprime{{\prime\prime}}
\def\notmodels{\mkern3mu\not\mkern-3mu\models} % shift the slash a bit to the right

\DeclareMathSymbol{;}{\mathord}{operators}{"3B} % chktex 18

\DeclareMathOperator{\dom}{dom}
\DeclareMathOperator{\codom}{codom}
\DeclareMathOperator{\Ro}{Ro}

\DeclareMathOperator{\Condrel}{Condrel}
\DeclareMathOperator{\condtrue}{true}
\DeclareMathOperator{\condfalse}{false}
\DeclareMathOperator{\id}{id}
\newcommand{\catname}[1]{\mathbf{#1}}
\def\ILC{I\mkern-1.5mu LC}

\def\obnull{0}
\def\Graphfin{Graph_{fin}}

\def\coincide{\discretionary{co}{incide}{co\kern.75pt incide}}

\def\mathFkern{\mkern-2mu}
\usepackage{xifthen}
\DeclareRobustCommand*{\xstep}[3]{\begingroup%
  \newtest{\leftF}[0]{\endswith{#1}{f}}% remove 2mu of space because of "f" before comma
  \newtest{\leftPrime}[0]{\endswith{#1}{\prime} \OR \endswith{#1}{'} }% remove a space becaue of an ending prime
  \newtest{\leftSimple}[0]{\equal{#1}{f} \OR \equal{#1}{\hat f} \OR \equal{#1}{\hat{f}} \OR \equal{#1}{d} }%
  \newtest{\rightSimple}[0]{\equal{#2}{\mathcal A} \OR \equal{#2}{\mathcal{A}} \OR \equal{#2}{\mathcal B} \OR \equal{#2}{\mathcal{B}} \OR \equal{#2}{\mathcal{B}_i} \OR \equal{#2}{\mathcal B_i} }%
  \xrightarrow{%
    #1%
    \ifthenelse{\leftF}{\mathFkern}{}%
    \ifthenelse{\leftPrime}{\mkern-4mu}{}%
    ,% chktex 26
    \ifthenelse{\leftSimple \AND \rightSimple}{\mkern2mu}{% very little space for something like f,A (2mu compensates exactly)
    \ifthenelse{\leftSimple \OR \rightSimple}{\mkern5mu}{% a bit more space if one side is simple
    \mkern8mu}}% much more space if both sides aren't simple
    #2%
  }_{#3}%
\endgroup}
\newcommand*{\cstep}[2]{\xstep{#1}{#2}{C}}
\newcommand*{\rstep}[2]{\xstep{#1}{#2}{R}}

\newcommand*{\ncxstep}[2]{\xrightarrow{#1}_{#2}}
\newcommand*{\nccstep}[1]{\ncxstep{#1}{C}}
\newcommand*{\ncrstep}[1]{\ncxstep{#1}{R}}

\def\simC{{\sim_C}}
\def\simR{{\sim_R}}
\newcommand*{\csimX}[1]{{\overset{\raisebox{-1.5pt}[2pt][0pt]{\scriptsize\ensuremath{\circ}}\mkern-2mu}{\sim}_{#1}}}
\def\csimC{\csimX{C}}
\def\csimR{\csimX{R}}
\def\csimtrue{\csimX{T}}
\def\csimenv{\csimX{E}}

\def\simidc{{\sim_{\id}}}

\newcommand*{\envstep}[1]{\stackrel{#1}{\leadsto}}

\usepackage{tikz}
\usetikzlibrary{positioning,shapes,arrows.meta,calc}

\newcommand\ensuretikz[1]{\tikzifinpicture{#1}{\tikz{#1}}}
\newcommand\gn{\ensuretikz{\node[gn] (n) at (0,0) {};}}
\tikzset{%
  gn/.style={circle,fill=black,inner sep=0pt, minimum size=6pt, prefix after command={\pgfextra{\tikzset{every label/.style={font=\footnotesize}}}}},
  roundbox/.style={rectangle, rounded corners=6pt},
  emptyroundbox/.style={rectangle, rounded corners=3pt, minimum width=0.5cm, minimum height=0.5cm},
  >=To,
  gedge/.style={->,>=latex},
  condtri/.style={draw,dart,dart tail angle=135,inner xsep=0.1em, inner ysep=0.2em,anchor=tip},
  fancydotted/.style={dash pattern=on 1.25pt off 1.75pt},
}
\tikzset{
  cospanintcommon/.style={draw=black!38,line width=0.5pt},
  cospanint/.style={->,cospanintcommon},
  cospanintmono/.style={cospanintcommon,>->},
  cospanarr/.style={->,line width=0.5pt}
}

\def\chanLabel{\mathit{ch}}
\def\relLabel{\mathit{rel}}
\def\unrelLabel{\mathit{unr}}
\def\noiseLabel{\mathit{noise}}

\def\faintborder{\draw[black!25, rounded corners=2pt, line width=0.2pt]
  ($(current bounding box.south west)+(-1pt,-1pt)$) rectangle
  ($(current bounding box.north east)+(1pt,0.5pt)$);}
\tikzset{
  mathchan/.style={baseline=-0.45ex,gn/.append style={minimum size=0.9ex},gedge/.append style={font=\footnotesize}}
}

\def\mathifnodessubscript{
  \tikz[mathchan,gn/.append style={minimum size=0.7ex}]{
    \node[gn] (l) at (0,0) {};
    \node[gn] (r) at (0.3,0) {};
    \faintborder
  }
}

\usepackage{keyval}
\makeatletter
\define@key{inlinechan}{chan}[\chanLabel]{\def\ic@chan{#1}}
\define@key{inlinechan}{reliable}[]{\def\ic@chan{\relLabel}}
\define@key{inlinechan}{unreliable}[]{\def\ic@chan{\unrelLabel}}
\define@key{inlinechan}{noise}[1]{\def\ic@noise{#1}}
\define@key{inlinechan}{msgleft}[1]{\def\ic@msgleft{#1}}
\define@key{inlinechan}{msgright}[1]{\def\ic@msgright{#1}}
\def\inlinechan#1{\begingroup%
  \setkeys{inlinechan}{#1}%
  \tikz[mathchan]{
    \ifundef{\ic@noise}{
      \ifundef{\ic@chan}{
        \pgfmathsetlengthmacro\nodedist{0.35cm}% neither channel nor noise
      }{
        \pgfmathsetlengthmacro\nodedist{0.75cm}% channel (but no noise)
    }}{
      \pgfmathsetlengthmacro\nodedist{0.97cm}% has noise
    }

    \node[gn] (l) at (0,0) {};
    \node[gn] (r) at (\nodedist,0) {};

    \ifundef{\ic@noise}{}{
      \draw[gedge] (l.-10) to[bend left=5] node[above,inner ysep=0pt]{\raisebox{1pt}{$\noiseLabel$}\kern4pt} (r.-170);
    }

    \ifundef{\ic@chan}{}{
      \draw[gedge] (l.-10) to node[above,inner ysep=0pt]{\raisebox{1pt}{$\ic@chan$}\kern4pt} (r.-170);
    }

    \ifundef{\ic@msgleft}{\def\ic@msgleft{0}}{}
    \ifundef{\ic@msgright}{\def\ic@msgright{0}}{}
    \ifcase\ic@msgleft% 0 (nothing)
    \or% 1
      \draw[gedge,overlay] (l) to[loop,out=140,in=-170,distance=0.3cm] node[left]{$m$\kern-2pt} (l);
      \node at (-14pt,2.5pt) {};
    \else% >= 2
      \PackageError{inlinechan}{Unsupported number of messages on left node: \ic@msgleft}{}
    \fi

    \ifcase\ic@msgright% 0 (nothing)
    \or% 1
      \draw[gedge,overlay] (r) to[loop,out=40,in=-10,distance=0.3cm] node[right]{\kern-1pt$m$} (r);
      \node at (\nodedist+14.5pt,2.5pt) {};
    \or% 2
      \draw[gedge,overlay] (r) to[loop,out=20,in=-15,distance=0.3cm,pos=0.7] node[right]{\kern-1pt$m$} (r);
      \draw[gedge,overlay] (r) to[loop,out=90,in=50,distance=0.3cm,pos=0.7] node[right]{\kern-1pt$m$} (r);
      \node at (\nodedist+14.5pt,5.5pt) {};
    \else% >= 3
      \PackageError{inlinechan}{Unsupported number of messages on right node: \ic@msgright}{}
    \fi

    \faintborder
  }%
\endgroup}
\define@key{inlinenlg}{left}[]{\def\in@left{#1}}
\define@key{inlinenlg}{right}[]{\def\in@right{#1}}
\define@key{inlinenlg}{a}[]{\def\in@left{a}}
\define@key{inlinenlg}{b}[]{\def\in@left{b}}
\define@key{inlinenlg}{e}[]{\def\in@left{e}}
\define@key{inlinenlg}{x}[]{\def\in@right{x}}
\define@key{inlinenlg}{c}[]{\def\in@right{c}}
\define@key{inlinenlg}{fw}[1]{\def\in@fullwidth{#1}}
\define@key{inlinenlg}{b1}[]{\def\in@left{b_1}}
\define@key{inlinenlg}{b2}[]{\def\in@left{b_2}}
\define@key{inlinenlg}{ea}[]{\def\in@left{e_a}}
\define@key{inlinenlg}{e1}[]{\def\in@left{e_1}}
\define@key{inlinenlg}{e2}[]{\def\in@left{e_2}}
\def\inlinenlg#1{\begingroup%
  \setkeys{inlinenlg}{#1}%
  \tikz[mathchan]{
    \node[gn] (n) at (0,0) {};
    \node at (0pt,-0.25pt) {};
    \node at (0pt,0.25pt) {};

    \ifundef{\in@left}{}{
      \draw[gedge,overlay] (n) to[loop,out=140,in=-170,distance=0.3cm] node[left]{$\in@left$\kern-2pt} (n);
      \node at (-14pt,2.5pt) {};
    }
    \ifundef{\in@right}{}{
      \draw[gedge,overlay] (n) to[loop,out=40,in=-10,distance=0.3cm] node[right]{\kern-1pt$\in@right$} (n);
      \node at (14.5pt,2.5pt) {};
    }
    \ifundef{\in@fullwidth}{\def\in@fullwidth{0}}{}
    \ifnum\in@fullwidth=1
      \node at (-14pt,2.5pt) {};
      \node at (14.5pt,2.5pt) {};
    \fi

    \faintborder
  }%
\endgroup}
\define@key{inlinerr}{a}[1]{\def\inrr@a{#1}}
\define@key{inlinerr}{b}[1]{\def\inrr@b{#1}}
\define@key{inlinerr}{c}[1]{\def\inrr@c{#1}}
\define@key{inlinerr}{fw}[1]{\def\inrr@fullwidth{#1}}
\def\inlinerr#1{\begingroup%
  \setkeys{inlinerr}{#1}%
  \tikz[mathchan]{
    \node[gn] (n) at (0,0) {};
    \node at (0pt,-0.25pt) {};
    \node at (0pt,0.25pt) {};

    \ifundef{\inrr@a}{}{
      \draw[gedge,overlay] (n) to[loop,out=90,in=140,distance=0.3cm,pos=0.75] node[left]{\kern-1pt$a$} (n);
      \node at (-13pt,4.5pt) {};
    }
    \ifundef{\inrr@b}{}{
      \draw[gedge,overlay] (n) to[loop,out=165,in=-160,distance=0.3cm,pos=0.5] node[left]{\kern-1pt$b$} (n);
      \node at (-13pt,0.25pt) {};
    }
    \ifundef{\inrr@c}{}{
      \draw[gedge,overlay] (n) to[loop,out=40,in=-10,distance=0.3cm] node[right]{\kern-1pt$c$} (n);
      \node at (13.5pt,2.5pt) {};
    }
    \ifundef{\inrr@fullwidth}{\def\inrr@fullwidth{0}}{}
    \ifnum\inrr@fullwidth=1
      \node at (-14pt,2.5pt) {};
      \node at (14.5pt,2.5pt) {};
    \fi

    \faintborder
  }%
\endgroup}
\makeatother

\newcommand{\loopOnArrow}[1]{\raisebox{0pt}[\dimexpr\height-1.5pt\relax][0pt]{$#1$}}

\usepackage[arrow,matrix,frame,curve,cmtip,color]{xy}

\begin{document}

\title{Conditional Bisimilarity for Reactive Systems}

\author[M.~Hülsbusch]{Mathias Hülsbusch\rsuper{a}}
\author[B.~König]{Barbara König\rsuper{a}}
\author[S.~Küpper]{Sebastian Küpper\rsuper{b}}
\author[L.~Stoltenow]{Lars Stoltenow\rsuper{a}}

\address{Universität Duisburg-Essen, Germany}
\email{\{barbara\_koenig,lars.stoltenow\}@uni-due.de}

\address{FernUniversität in Hagen, Germany}
\email{sebastian.kuepper@fernuni-hagen.de}

%\funding{Research partially supported by DFG Project BEMEGA.}% TODO
\iffalse%TODO
\begin{CCSXML}
<ccs2012>
  <concept>
    <concept_id>10003752.10003753.10003761</concept_id>
    <concept_desc>Theory of computation~Concurrency</concept_desc>
    <concept_significance>500</concept_significance>
  </concept>
  <concept>
    <concept_id>10003752.10010124.10010138</concept_id>
    <concept_desc>Theory of computation~Program reasoning</concept_desc>
    <concept_significance>500</concept_significance>
  </concept>
</ccs2012>
\end{CCSXML}

\ccsdesc[500]{Theory of computation~Concurrency}
\ccsdesc[500]{Theory of computation~Program reasoning}
% HACK: Remove or fix for FINAL version
\ACMCCS{%
  \textbf{[Theory of computation]}:
  Semantics and reasoning --- Program reasoning;
  \textbf{[Theory of computation]}:
  Models of computation --- Concurrency}
\fi

\begin{abstract}
  Reactive systems \`a la Leifer and Milner, an abstract categorical
  framework for rewriting, provide a suitable framework for deriving
  bisimulation congruences. This is done by synthesizing interactions
  with the environment in order to obtain a compositional semantics.

  We enrich the notion of reactive systems by conditions on two
  levels: first, as in earlier work, we consider rules enriched with
  application conditions and second, we investigate the notion of
  conditional bisimilarity. Conditional bisimilarity allows us to say
  that two system states are bisimilar provided that the environment
  satisfies a given condition.

  We present several equivalent definitions of conditional
  bisimilarity, including one that is useful for concrete proofs and
  that employs an up-to-context technique, and we compare with related
  behavioural equivalences. We consider examples based on DPO graph
  rewriting, an instantiation of reactive systems.
\end{abstract}

\maketitle

%%%%%%%%%%%%%%%%%%%%%%%%%%%%%%%%%%%%%%%%%%%%%%%%%%%%%%%%%%%%%%%%%%%%%%%

\section{Introduction}

Behavioural equivalences, such as bisimilarity, relate system states
with the same behaviour. Here, we are in particular interested in
conditional bisimilarity, which allows us to say that two states $a,b$
are bisimilar provided that the environment satisfies a condition
$\mathcal{C}$. Work on such conditional bisimulations appears somewhat
scattered in the literature (see for instance~\cite{l:context-dependent-bisim,hl:symbolic-bisimulations,f:bisimulations-boolean-vectors,BKKS17}). They
also play a role in the setting of featured transition systems for
modelling software product lines~\cite{DBLP:conf/icse/CordyCPSHL12},
where the behaviour of many products is specified in a single
transition system. In this setting it is possible to state that two
states are bisimilar for certain products, but not for others.

We believe that conditional notions of behavioural equivalence are
worthy of further study. In practice it may easily happen that two
sub-systems are only ever used in restricted environments and it is
too much to ask that they behave equivalently under all possible
contexts. Furthermore, instead of giving a simple yes/no-answer,
bisimulation checks can answer in a more fine-grained way, specifying
conditions which ensure bisimilarity.

We state our results in a very general setting: reactive systems \`a
la Leifer and Milner~\cite{LM00}, a categorical abstract framework for
rewriting, which provides a suitable framework for deriving
bisimulation congruences. In particular, this framework allows to
synthesize labelled transitions from plain reaction rules, such that
the resulting bisimilarity is automatically a congruence. Intuitively,
the label is the minimal context that has to be borrowed from the
environment in order to trigger a reduction. (Transitions labelled
with such a minimal context will be called
representative steps in the sequel. They are related to the idem pushout steps of~\cite{LM00}.)  Here, we rely on the notion of saturated bisimilarity
introduced in~\cite{bkm:saturated} and we consider reactive system
rules with application conditions, generalizing~\cite{DBC-CRS}.

Important instances of reactive systems are process calculi
with contextualization, bigraphs~\cite{jm:bigraphs} and
double-pushout graph rewriting~\cite{cmrehl:algebraic-approaches},
or in general rewriting in
adhesive categories~\cite{ls:adhesive-journal}. Hence we can use our
results to reason about process calculi as well as dynamically
evolving graphs and networks for various different types of graphs
(directed or undirected graphs, node- or edge-labelled graphs, hypergraphs, etc.).

Our contributions in this paper can be summarized as follows:

\begin{itemize}
\item We define the notion of conditional bisimilarity, in fact we
  provide three equivalent definitions: two notions are
  derived from saturated bisimilarity, where a context step (or a
  representative step) can be mimicked by several answering
  steps. Third, we compare with the notion of conditional environment
  congruence, which is based on the idea of annotating transitions
  with passive environments enabling a step.

\item Conditional bisimulation relations tend to
be very large --- often infinite in size.
To obtain possibly substantial reductions of proof obligations in a bisimulation proof,
we propose an up-to context technique (for up-to techniques and their
history see~\cite{ps:coinduction-enhancements-historical}).
In particular, it can replace an infinite conditional bisimulation relation by a possibly finite
bisimulation up-to context, which provides a witness for bisimilarity.
We also view our up-to technique in a general lattice-theoretical setting
and prove the compatibility property~\cite{ps:enhancements-coinductive}, which
not only implies soundness of the technique, but also allows it to be
composed with other (compatible) up-to techniques.

\item We use the notion of representative steps in order to
obtain finitely branching transition systems, further reducing proof
obligations.

\item We compare conditional bisimilarity with
related notions of behavioural equivalence.

\item To illustrate our concepts, we
work out a small case study in the context of double-pushout graph
rewriting, where we model message passing over reliable and unreliable
channels.
\end{itemize}

\noindent
The article is structured as follows: First, in
\Cref{sec:reasys-nocond} we introduce the fundamental
ideas for reactive systems without conditions, including all
preliminary definitions and techniques developed for reactive systems
relevant to our work. In \Cref{sec:reasys-cond}, we consider the
refinement to conditional reactive systems, before we turn towards our
main contribution in \Cref{sec:condbisim}, which is conditional
bisimulation and its up-to variant in \Cref{sec:uptocond}.  In
\Cref{sec:alternative-char} we give an alternative characterization of
conditional bisimilarity and compare to related notions~of behavioural
equivalence and we conclude in \Cref{sec:conclusion}.

% This paper is an extended version of~\cite{hkks:cond-bisim},
% where we have only shown soundness of our up-to technique.
% We took up on the suggestion of a reviewer and view up-to techniques in
% a general lattice-theoretical setting and prove the compatibility
% property~\cite{ps:enhancements-coinductive}, which not only implies soundness of the technique, but also
% allows it to be composed with other (compatible) up-to techniques. As a
% result, \Cref{sec:uptocond} has been completely
% rewritten, and there is an alternative characterization of the concepts
% of \Cref{sec:condbisim}.
% We also improved the general presentation of all concepts, extending
% and reworking many explanations throughout the paper.
% Additional examples have been included in the preliminaries
% (conditions, borrowed context diagrams, shift operation).

\section{Reactive Systems}%
\label{sec:reasys-nocond}

We denote the composition of arrows
$f \colon A \to B,\ g \colon B \to C$ by ${f;g} \colon A \to C$.
Usually written $g \circ f$, we chose $f;g$ to better match the
reading order of the diagram
$A \xrightarrow f B \xrightarrow g C = A \xrightarrow{f;g} C$.

\subsection{Reactive Systems without Conditions}\label{subsec-reasys-nocond}

We now define reactive systems, which were introduced in~\cite{LM00}
and extended in~\cite{DBC-CRS} with application conditions for rules.
We initially only look at reactive systems without conditions.
Conditions and the definition of reactive systems with conditions will
be introduced later, in \Cref{sec:reasys-cond}.

\begin{defi}[Reactive system rules, reaction]\label{def-reasys}%
  Let $\catname C$ be a category with a distinguished object $\obnull$
  %\footnote{Although in~\cite{LM00}, this object has been referred to as \enquote{initial object}, being initial is not required.}
  (not necessarily initial).  A \emph{rule} is a pair $(\ell,r)$ of
  arrows $\ell,r \colon {\obnull \to I}$ (called left-hand side and
  right-hand side).  A \emph{reactive system} is a set of rules.

  Let $\mathcal S$ be a reactive system and
  $a, a' \colon \obnull \to J$ be arrows. We say that $a$
  \emph{reduces to} $a'$ ($a \leadsto a'$) whenever there exists a
  rule $(\ell,r) \in \mathcal S$ with $\ell,r \colon \obnull \to I$
  and an arrow $c \colon I \to J$ (the \emph{reactive context}) such
  that $a = \ell;c$ and $a' = r;c$.
\end{defi}

% \rnewhl{Left-hand and right-hand sides are only contextualized from
%   the right, but never from the left, which would correspond to some
%   form of substitution. Hence it is sufficient to fix some object $0$
%   for the ``inner'' interface of an arrow.}{B: I added some
%   explanation for the distinguished object $0$. Maybe we have to
%   delete it due to space reasons.}

Using a notation closer to process calculi, we could write
$C[P]\leadsto C[P']$ whenever there is a reaction rule $P\to P'$ and a
context $C[\_]$. Fixing a distinguished object $0$ means that we
consider only ground reaction rules (as opposed to the open reactive
systems investigated in~\cite{kss:labels-from-reductions}).

An important instance are reactive systems where the arrows are
cospans in a base category $\catname D$ with pushouts~\cite{ss:reactive-cospans,Sobocinski}.  A \emph{cospan} is a pair of
arrows \mbox{$f_L \colon A \to C,\ f_R \colon B \to C$}.  A~cospan is
\emph{input linear} if its left arrow $f_L$ is mono.

% TODO: vielleicht wieder rein?
% Such cospans can be understood as equipping $C$ with input and output
% interfaces $A, B$ and corresponding arrows $f_L,f_R$ to relate the
% interfaces to $C$.

  \vskip\floatsep% We want these two figures on the same page (together with
  \noindent\hfill% the related text) that's why. They're not directly related,
  % that's why we do this instead of using subfigure 1a and 1b.
  \parbox[b]{0.45\textwidth}{%
    \captionsetup{type=figure}%
    \centering\begin{tikzpicture}[x=1.10cm,y=1.10cm]% PAGEBREAK-ADJUST

    \def\sqwo{1.2}
    \def\sqoo{1.4}
    \def\sqho{0.95}

    \node (a) at (-\sqwo-\sqoo,0) {$A$};
    \node (x) at (-\sqwo,0) {$X$};
    \node (b) at (0,\sqho) {$B$};
    \node (y) at (\sqwo,0) {$Y$};
    \node (c) at (\sqwo+\sqoo,0) {$C$};
    \node (z) at (0,-\sqho) {$Z$};

    \def\smf{\footnotesize}
    \draw[cospanint] (a) -- node[above]{\smf$f_L$} (x);
    \draw[cospanint] (b) -- node[above,pos=0.7]{\smf$f_R\ $} (x);
    \draw[cospanint] (b) -- node[above,pos=0.7]{\smf\kern5pt$g_L$} (y);
    \draw[cospanint] (c) -- node[above]{\smf$g_R$} (y);

    \draw[cospanint] (x) -- node[below,pos=0.4]{\smf$p_L$\kern8pt} (z);
    \draw[cospanint] (y) -- node[below,pos=0.4]{\smf\kern12pt$p_R$} (z);

    \draw[cospanarr] (a) to[bend left=30] node[above,pos=0.2]{$f$} (b);
    \draw[cospanarr] (b) to[bend left=30] node[above,pos=0.8]{$\ g$} (c);
    \draw[cospanarr] (a) to[bend right=50] node[below]{$f;g$} (c);

    \node at (0,0) {\textup{(PO)}};

    \end{tikzpicture}%
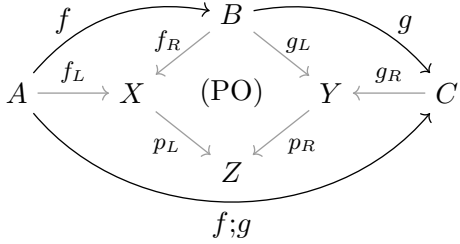
    \caption{Composition of cospans $f$ and~$g$ is done via pushouts}%
    \label{fig-cospan-compo}%
  }%
  \hfill%
  \parbox[b]{0.45\textwidth}{%
    \captionsetup{type=figure}%
    \centering\begin{tikzpicture}[x=1.10cm,y=1.10cm]% PAGEBREAK-ADJUST

    \def\sqw{1.25}
    \def\sqh{1.05}

    \node (d) at (0,0) {$\obnull$};
    \node (l) at (\sqw,0) {$L$};
    \node (i) at (2*\sqw,0) {$I$};
    \node (r) at (3*\sqw,0) {$R$};
    \node (r0) at (4*\sqw,0) {$\obnull$};
    \node (g) at (1*\sqw,-1*\sqh) {$G$};
    \node (c) at (2*\sqw,-1*\sqh) {$C$};
    \node (h) at (3*\sqw,-1*\sqh) {$H$};
    \node (k) at (2*\sqw,-2*\sqh) {$\obnull$};

    \draw[cospanintmono]
      (d) edge (l) edge (g)
      (l) edge (g)
      (i) edge (c)
      (r) edge (h)
      (r0) edge (r) edge (h);
    \draw[cospanint]
      (k) edge (g) edge (c) edge (h)
      (i) edge (l) edge (r)
      (c) edge (h) edge (g);

    \draw[cospanarr] (d.north east) to[bend left =10] node[above]{$\ell$} (i.north west);
    \draw[cospanarr] (d.south) to[bend right=10] node[left ]{$a\ $} (k.west);
    \draw[cospanarr] (i.south west) to[bend right=10,pos=0.2] node[left]{$c$} (k.north west);
    \draw[cospanarr] (r0.north west) to[bend right=10] node[above]{$r$} (i.north east);
    \draw[cospanarr] (r0.south) to[bend left=10] node[right]{$a^\prime$} (k.east);

    \end{tikzpicture}%
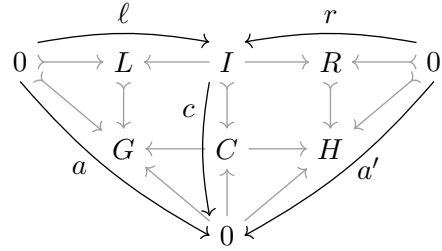
    \caption{Double-pushout graph transformation as reactive system steps}%
    \label{fig-dpo-reasys}
  }%
  \hfill\mbox{}%
  \vskip\floatsep

  Two cospans
  \mbox{$f \colon A \xrightarrow{f_L} X \xleftarrow{f_R} B,\ g \colon
    B \xrightarrow{g_L} Y \xleftarrow{g_R} C$} are composed by taking
  the pushout $(p_L, p_R)$ of $(f_R, g_L)$ as shown in
  \Cref{fig-cospan-compo}.  The result is the cospan
  \mbox{${f;g} \colon A \xrightarrow{f_L;p_L} Z \xleftarrow{g_R;p_R}
    C$}, where $Z$ is the pushout object of $f_R,\; g_L$.  For
  adhesive categories~\cite{ls:adhesive-journal} (see
  \mbox{Appendix~\ref{sec:adhesive}}), the composition of input linear
  cospans again yields an \mbox{input} linear cospan (by applying~\cite[Lemma 4.2]{ls:adhesive-journal} to the cospan composition
  diagram).  Given an \mbox{adhesive} category $\catname D$,
  $\ILC(\catname D)$ is the category where the objects are the objects
  of $\catname D$, the arrows $f\colon A\rightarrow C$ are
  \mbox{input} linear cospans $f\colon A \rightarrow B \leftarrow C$
  of $\catname D$ and composition is performed via pushouts as
  described above.  We see an arrow $f \colon A \to C$ of
  $\ILC(\catname D)$ as an object~$B$ of~$\catname D$ equipped with
  two interfaces $A,C$ and corresponding arrows $f_L,f_R$ to relate
  the interfaces to $B$, and composition glues the inner objects of
  two cospans via their common interface.  Input linearity is useful
  since we rely on adhesive categories where pushouts along monos are
  well-behaved. In particular, they always exist and form Van Kampen
  squares (see Appendix~\ref{sec:adhesive}),
  the latter being a requirement for borrowed context
  diagrams (\Cref{subsec-rsq}).
% maybe? ^* In such a category, the arrows can be interpreted as being the actual objects equipped with two interfaces, and composition of arrows relates two objects by their interfaces.

In this article, as a running example we consider the category $\catname\Graphfin$,
which has finite graphs (we use directed multigraphs with node and edge labels) as objects and
total graph morphisms (functions that map nodes and edges of one graph to another, with the edge map being consistent to the node map) as arrows.
% For easier readability, we depict graph morphisms implictly by numbering the nodes in the source of a map identically to their images. We do not denote the edge mappings, but the node mappings only allows for one possible (total) mapping of edges.
In $\catname\Graphfin$, monos are
% - source? (e.g. Treewidth)
exactly the injective graph morphisms.  We then use reactive systems
over $\ILC(\catname\Graphfin)$ (input-linear cospans of graphs),
i.e.\ we rewrite graphs with inter\-faces.  If the distinguished object
$0$ is the empty graph (the initial object of $\catname\Graphfin$),
such reactive systems \coincide~\cite{ss:reactive-cospans} with the
% ^Section 4.1
well-known \emph{double pushout (DPO) graph transformation}
approach~\cite{DPO,RevisitedDPO} when used with injective matches.  As
shown~in \Cref{fig-dpo-reasys}, a DPO rewrite step $G \Rightarrow H$
can be expressed as a reactive system reaction $a \leadsto a'$ where
the pushouts of the DPO step are obtained from cospan compositions
$\ell;c$ and $r;c$.

\subsection{Deriving Bisimulation Congruences}

The reduction relation $\leadsto$ generates an unlabelled transition
system, where the states are the reactive agents (in our example, graphs).
Note that bisimilarity on this transition system only checks whether any reaction is possible: for two bisimilar agents, it is not required that the same rule is used in their reactions, or even that the reaction is applied at the same position.

A disadvantage of bisimilarity on $\leadsto$ is that it usually is not a
congruence: it is easy to construct an example where neither $a$ nor
$b$ can perform a step since no complete left-hand side is
present (hence $a$ would be bisimilar to $b$). However, by adding a suitable context $f$, $a;f$~could contain a
full left-hand side and can reduce, whereas $b;f$ can not.

Therefore, to check whether two components can be exchanged, they have
to be combined with every possible context and bisimilarity has to be
shown for each.

% If contexts are not considered, systems may be recognized as
% equivalent even though they do not show the same behaviour.  As an
% example, consider two components which both modify parts of the
% environment, albeit in observably differing ways.  The components
% themselves (without environment) both cannot do any steps, and are
% therefore bisimilar, even though their behaviour differs when placed
% into a suitable environment.

In order to obtain a congruence, we can resort to defining
bisimulation on labelled transitions, using as labels the additional
contexts that allow an agent to react~\cite{LM00,DBC-CRS}.

\begin{defi}[Context step (without conditions)~\cite{DBC-CRS}]\label{def-ctxstep}%
\floatingpicspaceright[4]{2.85cm}
\begin{floatingpic}[-.25\baselineskip]{-2.55cm}%
\centering\begin{tikzpicture}

\def\sqw{1.05}
\def\sqh{1.0}

\path[use as bounding box] (-0.1-0.21,0.45) rectangle (2*\sqw+0.1,-\sqh-0.2);

\node (tlempty) at (0,0) {$\obnull$};
\node (i) at (1*\sqw,0) {$I$};
\node (trempty) at (2*\sqw,0) {$\obnull$};
\node (j) at (0,-\sqh) {$J$};
\node (k) at (1*\sqw,-\sqh) {$K$};
\draw[->] (tlempty) -- node[above]{$\ell$} (i);
\draw[->] (trempty) -- node[above]{$r$} (i);
\draw[->] (trempty) -- node[below,pos=0.3]{$a^\prime$} (k);

\draw[->] (tlempty) -- node[left]{$a$} (j);
\draw[->] (i) -- node[left]{$c$} (k);

\draw[->] (j) -- node[above,pos=0.45]{$f$} (k);

\end{tikzpicture}
\end{floatingpic}%
Let $\mathcal S$ be a reactive system and $a \colon \obnull \to J,\ f \colon J \to K,\ a' \colon \obnull \to K$ be arrows.
We write $a \nccstep{f} a'$ whenever $a;f \leadsto a'$ (i.e.\ there exists a rule $(\ell,r) \in \mathcal S$
and an arrow~$c$ such that $a;f = \ell;c,\allowbreak\ a' = r;c$).
Such steps are called \emph{context steps}.
\end{defi}

Intuitively we have to find a context $f$ for the arrow $a$ (which we
want to rewrite) such that we obtain the left-hand side $\ell$ plus
some additional context $c$.
The name \emph{context step} stems from the fact that $a$ might not be
able to do a reaction
on its own, but requires an additional context~$f$. This can be seen in the
following example:

\begin{exa}[Context step (without conditions)]\label{ex-cb-unreliable-pre}
Consider the following reactive system over $\ILC(\catname\Graphfin)$,
i.e.\ all arrows (such as $\ell, r, \dots$) are input-linear cospans of graphs that represent graphs with interfaces.
We model a network of nodes that pass messages (represented by $m$-loops) over
communication channels (represented by $\chanLabel$-edges).
The transmission of a message from the left node to the right node can be represented with the following rule:

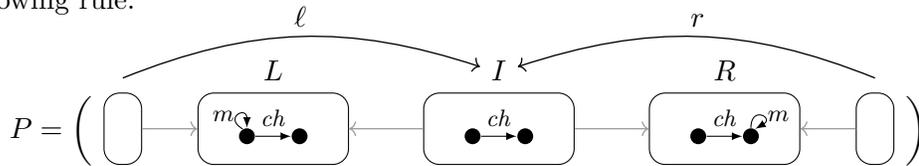
\begin{tikzpicture}[gedge/.append style={font=\footnotesize}]
  \def\drawthebox#1#2{
    \node (#2) at (0,0) [draw,roundbox,minimum width=2cm, minimum height=1.25cm] {};
    \node[anchor=west] at (-1cm,-0.35) {$\scriptstyle{#1}$};
  }
  \def\drawtheboxnolab#1{
    \node (#1) at (0,0.15cm) [draw,roundbox,minimum width=2cm, minimum height=0.95cm] {};
  }
  \def\smolnodes{
    \node[gn] (l) at (-0.35,0.05) {};
    \node[gn] (r) at (0.35,0.05) {};
  }
  \def\smolchan{
    \smolnodes
    \draw[gedge] (l) to node[above]{$\chanLabel$} (r);
  }
  \def\smolchanSh{\begin{scope}[shift={(-0.25,0)}]\smolchan\end{scope}}
  \def\leftmloopt{\draw[gedge] (l.120) to[loop,out=120+30,in=120-30,distance=0.3cm] node[left]{$m\mkern-5mu$} (l.90);}
  \def\rightmloopt{\draw[gedge] (r.90) to[loop,out=60+30,in=60-30,distance=0.3cm] node[right]{$\mkern-2mu m$} (r.60);}
  \begin{scope}[shift={(0,0)}]
    \begin{scope}[shift={(-2,-0.125)}]
      \node (leftnull) at (0,0.15cm) [draw,roundbox,minimum width=0.5cm, minimum height=0.95cm] {};
    \end{scope}

    \begin{scope}[shift={(0,-0.125)}]
      \smolchan\leftmloopt
      \drawtheboxnolab{leftbox}
    \end{scope}

    \begin{scope}[shift={(3,-0.125)}]
      \smolchan
      \drawtheboxnolab{ifbox}
    \end{scope}

    \begin{scope}[shift={(6,-0.125)}]
      \smolchan\rightmloopt
      \drawtheboxnolab{rightbox}
    \end{scope}

    \begin{scope}[shift={(8,-0.125)}]
      \node (rightnull) at (0,0.15cm) [draw,roundbox,minimum width=0.5cm, minimum height=0.95cm] {};
    \end{scope}

    \node[inner xsep=0cm,anchor=east] at (-2.4,0) {$P = \bigg($};
    \node[inner xsep=0cm,anchor=west] at (8.4,0) {$\bigg)$};
    \draw[cospanint] (leftnull) -- (leftbox);
    \draw[cospanint] (ifbox) -- (leftbox);
    \draw[cospanint] (ifbox) -- (rightbox);
    \draw[cospanint] (rightnull) -- (rightbox);
  \end{scope}

  % "explanation" for first one
  \draw[cospanarr] (-2,0.7) to[bend left=20] node[above,overlay]{$\ell$} (3-0.25,0.85);
  \draw[cospanarr] (8,0.7) to[bend right=20] node[above,overlay]{$r$} node[above]{} (3+0.25,0.85);
  % FINAL: ^ check layout. If the previous line is NOT mostly blank, remove the "overlay" key which causes the labels to be included in the bounding box again.
  % (We excluded them here because otherwise the space looks very large)
  \node at (0,0.8) {$L$};
  \node at (3,0.8) {$I$};
  \node at (6,0.8) {$R$};
\end{tikzpicture}

All graph morphisms are induced by edge labels and position of nodes,
i.e.\ the left node is always mapped to the left node.
% wir könnten das stärker rausstellen oder an geeigneten Stellen wiederholen

We can observe that a channel by itself
($a = {\emptyset \rightarrow \inlinechan{chan} \leftarrow \inlinechan{}}$)
cannot do a reaction, since there is no message to be transferred.
However, if a message on the left node is borrowed
(\mbox{$f = {\inlinechan{} \rightarrow \inlinechan{msgleft} \leftarrow \inlinechan{}}$}),
we obtain $a;f$ (\Cref{ctxstep-example-fig-1}),
to which the example rule can be applied (\Cref{ctxstep-example-fig-2}).
As a result, we obtain the context step
$a \nccstep{f} a'$ or~%
${(\emptyset \rightarrow \inlinechan{chan} \leftarrow \inlinechan{})}
\nccstep{(\inlinechan{} \rightarrow \loopOnArrow{\inlinechan{msgleft}} \leftarrow \inlinechan{})}
{(\emptyset \rightarrow \inlinechan{chan,msgright} \leftarrow \inlinechan{})}$.
  \begin{figure}[h]
    % commands for the two figures:
      \def\drawthebox#1#2{
        \node (#2) at ($(\boxSxp,\boxSyp)!.5!(-\boxSxm,-\boxSym)$) [draw,roundbox,minimum width=\boxW cm, minimum height=\boxH cm] {};
        \node[anchor=west] at (-\boxSxm,-0.3) {$\scriptstyle{#1}$}; % -0.2 when using 0.35 below
      }
      \def\arrowfromto#1#2#3{
        % ignore mono for readability (no cospanint#3)
        \draw[cospanint] (#1) edge (#2);
      }
      \def\cospanfromto#1#2#3{% left middle right | name namelabelpos anchor namelabelopts
        \arrowfromto{#1}{#2}{mono}
        \arrowfromto{#3}{#2}{}
      }
      \def\cospanfromtoarr#1#2#3#4#5#6#7{% left middle right | name namelabelpos anchor namelabelopts
        \cospanfromto{#1}{#2}{#3}
        \draw[cospanarr] (#1.#6) to[#7] node[#5]{$#4$} (#3.#6);
      }
      \def\smolnodes{
        \node[gn] (l) at (0,0.05) {};
        \node[gn] (r) at (0.5,0.05) {};
      }
      \def\smolchan{
        \smolnodes
        \draw[gedge] (l) to node[above]{$\chanLabel$} (r);
      }
      \def\smolchanSh{\begin{scope}[shift={(-0.25,0)}]\smolchan\end{scope}}
      \def\leftmloopt{\draw[gedge] (l.120) to[loop,out=120+30,in=120-30,distance=0.3cm] node[left]{$m\mkern-5mu$} (l.90);} % chktex 8
      \def\rightmloopt{\draw[gedge] (r.90) to[loop,out=60+30,in=60-30,distance=0.3cm] node[right]{$\mkern-2mu m$} (r.60);} % chktex 8
      \def\boxat#1#2#3#4#5{% gridx gridy label boxname content
        \begin{scope}[shift={(#1*\gridx,#2*\gridy)}]
          \drawthebox{#3}{#4}
          #5
        \end{scope}
      }
      \def\gridparams{
        \pgfmathsetmacro\gridx{2.1}
        \pgfmathsetmacro\gridy{-1.4}
        \pgfmathsetmacro\boxSxp{0.85}
        \pgfmathsetmacro\boxSxm{0.55}
        \pgfmathsetmacro\boxSyp{0.5}
        \pgfmathsetmacro\boxSym{0.5}
        \pgfmathsetmacro\boxW{\boxSxp+\boxSxm}
        \pgfmathsetmacro\boxH{\boxSyp+\boxSym}

        \node at ($(2*\gridx,2*\gridy - \boxSym - 0.6)$) {};
        \node at ($(2*\gridx,0*\gridy + \boxSyp + 0.6)$) {};
      }
      \def\polabels#1#2#3#4{
        \def\sqlfs{\scriptsize\color{gray}}
        \node at ($(boxnull)!0.5!(Gp)$) {\sqlfs\textup{#1}};
        \node at ($(L)!0.5!(C)$) {\sqlfs\textup{#2}};
        \node at ($(G)!0.5!(F)$) {\sqlfs\textup{#3}};
        \node at ($(Gp)!0.5!(K)$) {\sqlfs\textup{#4}};
      }
    \begin{subfigure}[b]{0.42\textwidth}
      \centering
      \begin{tikzpicture}[gedge/.append style={font=\scriptsize}, cospanarr/.append style={font=\scriptsize},y=1.00cm]% PAGEBREAK-ADJUST
        \gridparams
        \boxat00{\obnull}{boxnull}{}
        \boxat01{}{C0}{\smolchan}
        \boxat02{}{N0}{\smolnodes}
        \boxat12{}{Nm}{\smolnodes\leftmloopt}
        \boxat22{}{N0r}{\smolnodes}
        \boxat11{}{G}{\smolchan\leftmloopt}
        \cospanfromto{boxnull}{C0}{N0}
        \draw[cospanarr] (boxnull.south west) to[bend right=12] node[left]{$a$} (N0.north west);
        \cospanfromto{N0}{Nm}{N0r}
        \draw[cospanarr] (N0.south east) to[bend right=10] node[below]{$f$} (N0r.south west);
        \cospanfromto{boxnull}{G}{N0r}
        \draw[cospanarr] (boxnull.east) to[bend left=25] node[above right]{$a;f$} (N0r.north);

        \arrowfromto{C0}{G}{}
        \arrowfromto{Nm}{G}{}
      \end{tikzpicture}
      \caption{A channel $a$ together with borrowed context $f$, resulting in $a;f$.}%
      \label{ctxstep-example-fig-1}
    \end{subfigure}\hfill%\hspace{1cm}
    \begin{subfigure}[b]{0.56\textwidth}
      \centering
      \begin{tikzpicture}[gedge/.append style={font=\scriptsize}, cospanarr/.append style={font=\scriptsize},y=1.00cm]% PAGEBREAK-ADJUST
        \gridparams
        %\boxat00{\obnull}{boxnull}{}
        %\boxat40{\obnull}{boxnullr}{}
        \node (boxnull) at ($(\boxSxp,0)!.25!(-\boxSxm,0) + (0,\boxSyp)!.5!(0,-\boxSym) + (0*\gridx,0)$) [draw,roundbox,minimum width=0.5*\boxW cm, minimum height=\boxH cm] {};
        \node (boxnullr) at ($(\boxSxp,0)!.75!(-\boxSxm,0) + (0,\boxSyp)!.5!(0,-\boxSym) + (4*\gridx,0)$) [draw,roundbox,minimum width=0.5*\boxW cm, minimum height=\boxH cm] {};
        \node[anchor=west] at ($(-\boxSxm,-0.3) + (0*\gridx,0) + (\boxW*0.5,0)$) {$\scriptstyle{\obnull}$};
        \node[anchor=west] at ($(-\boxSxm,-0.3) + (4*\gridx,0)$) {$\scriptstyle{\obnull}$};

        \boxat10{L}{L}{\smolchan\leftmloopt}
        \boxat20{I}{I}{\smolchan}
        \boxat30{R}{R}{\smolchanSh\rightmloopt}
        \boxat11{}{G}{\smolchan\leftmloopt}
        \boxat21{C}{C}{\smolchan}
        \boxat22{}{K}{\smolnodes}
        \boxat31{}{H}{\smolchanSh\rightmloopt}
        \cospanfromto{boxnull}{L}{I} %{\ell}{above}{north}{bend left=10}
        \cospanfromto{boxnullr}{R}{I}
        \cospanfromto{boxnull}{G}{K}
        \cospanfromto{boxnullr}{H}{K}
        \cospanfromto{I}{C}{K}
        \arrowfromto{L}{G}{}
        \arrowfromto{C}{G}{}
        \arrowfromto{C}{H}{}
        \arrowfromto{R}{H}{}
        \draw[cospanarr] (boxnull.north east) to[bend left=10] node[above]{$\ell$} (I.north west);
        \draw[cospanarr] (boxnullr.north west) to[bend right=10] node[above]{$r$} (I.north east);
        \draw[cospanarr] (boxnull.south) to[bend right=25] node[below left]{$a;f$} (K.west);
        \draw[cospanarr] (boxnullr.south) to[bend left=25] node[below right]{$a'$} (K.east);
        \draw[cospanarr] (I.south east) to[bend left=12,pos=0.3] node[right]{$c$} (K.north east);
      \end{tikzpicture}
      \caption{$a;f$ can now do a reaction to $a' = r;c$.\\~}%
      \label{ctxstep-example-fig-2}
    \end{subfigure}
    \caption{Visualization of the context step described in \Cref{ex-cb-unreliable-pre}.}%
    \label{ctxstep-example-fig}
  \end{figure}
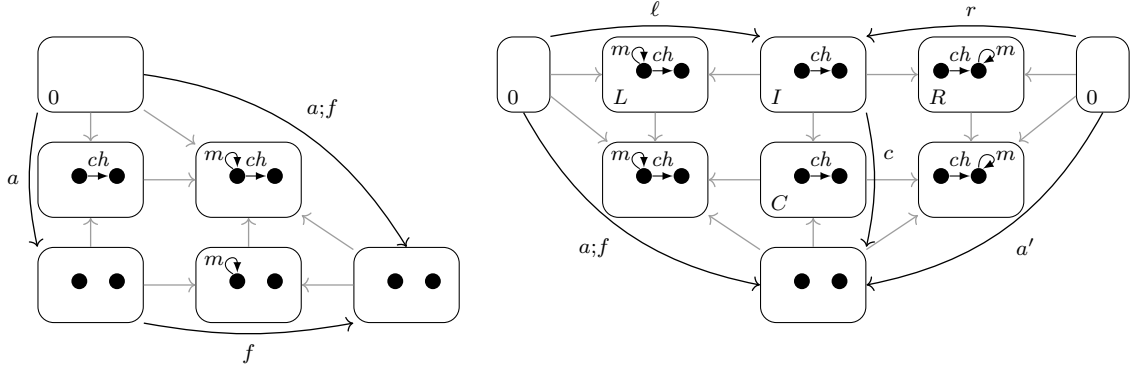
\end{exa}

A bisimulation relation over $\rightarrow_C$ is called \emph{saturated
  bisimulation}, as it checks all contexts. Consequently, saturated
bisimilarity $\simC$ (${\sim_{S\mkern-2mu AT}}$ in~\cite{DBC-CRS}) is
a congruence~\cite{bkm:saturated,DBC-CRS}, i.e., it is closed under
contextualization. In other words $a \mathrel{\simC} b$ implies $a;c \mathrel{\simC} b;c$
for all contexts $c$.

\subsection{Representative Squares}\label{subsec-rsq}

  Checking bisimilarity of context steps is impractical
  because the transition system is generally infinitely branching:
  usually, $f$ can be chosen from an infinite set of possible
  contexts, which all have to be checked.  Most of these contexts are
  larger than necessary, that is, they contain elements that do not
  actively participate in the reduction.  (In
  \Cref{ex-cb-unreliable-pre}, contexts can be arbitrarily large, as
  % ... arbitrarily large graphs can be used for $f$, as long as ...
  long as they have an $m$-loop on the left node.)  An improvement
  would be to check only the minimal contexts from which all
  other context steps can be derived.

When checking which contexts are required to make a rule applicable,
%(in the example above, at least an $m$-loop)
in the reaction diagram (\Cref{def-ctxstep}) the arrows $a, \ell$ are
given and we need to check for possible values of~$f$ (which generate
matching $c, a'$).  To derive a set of contexts $f$ which is as small
as possible --- preferably finite --- \cite{CRS,DBC-CRS} introduced the % chktex 2
notion of representative squares, which describe a way to represent
all possible squares that close a pair $a,\ell$ by a smaller set of
squares (the so-called representative squares).  We can then limit
bisimilarity checking to just the steps using representative squares,
which, if this smaller set is indeed finite, leads to a finitely
branching transition system.

\begin{figure}[t]
  \captionsetup{width=.9\linewidth}% better text flow
  \centering\begin{tikzpicture}[x=1.10cm,y=1.10cm]% PAGEBREAK-ADJUST

  \def\sqw{2}
  \def\sqh{2}
  \def\sqhsp{0.3} % how much to inset C, B compared to the full D' square size
  \def\sqhspi{0.4} % how much to additionally inset D (in addition to \sqhsp)
  \def\sqhsr{0.25} % shrink of the left side. was 0 in masters thesis, but a tiny offset is acceptable (it should however not be exactly halfway between D and D', because this might suggest that D' is "split" to some smaller object D and some larger-than-D' object that is also called D'
  \def\sqhsrr{0.25} % like sqhsr but for vertical size of the first diagram

  % %%% LEFT - PRE R-STEP %%%
  \begin{scope}[shift={(0,0)}]
    \node (a) at (0,0) {$A$};
    \node (b) at (\sqw-\sqhsr,0) {$B$};
    \node (c) at (0,-\sqh+\sqhsrr) {$C$};
    \node (ds) at (\sqw-\sqhsr,-\sqh+\sqhsrr) {$D'$};
    \draw[->] (a) -- node[above]{$\alpha_1$} (b);
    \draw[->] (a) -- node[left]{$\alpha_2$} (c);

    \draw[->] (b) -- node[right]{$\delta_1$} (ds);
    \draw[->] (c) -- node[below]{$\delta_2$} (ds);
  \end{scope}

  \node at (\sqw+1,-0.5*\sqh) {$\rightarrow$};

  % %%% RIGHT %%%
  \begin{scope}[shift={(\sqw+2.5,0)}]
    \node (a) at (0,0) {$A$};
    \node (b) at (\sqw-\sqhsp,0) {$B$};
    \node (c) at (0,-\sqh+\sqhsp) {$C$};
    \node (d) at (\sqw-\sqhsp-\sqhspi,-\sqh+\sqhsp+\sqhspi) {$D$};
    \node (ds) at (\sqw,-\sqh) {$D'$};
    \draw[->] (a) -- node[above]{$\alpha_1$} (b);
    \draw[->] (a) -- node[left]{$\alpha_2$} (c);
    \draw[->] (b) -- node[left,pos=0.25]{$\beta_1$} (d);
    \draw[->] (c) -- node[above,pos=0.3]{$\beta_2$} (d);

    \draw[->] (d.center)+(5pt,-5pt) -- (ds);
    \node at ($(d.center)+(12pt,-4pt)$) {$\gamma$};
    \draw[->] (b) -- node[right]{$\delta_1$} (ds);
    \draw[->] (c) -- node[below]{$\delta_2$} (ds);
  \end{scope}
  \end{tikzpicture}%
  \caption{Every commuting square of the category (left) can be reduced to a representative square in $\kappa$ and an arrow $\gamma$ which extends the representative square to the original square~(right).}%
  \label{repr-sq-init}
\end{figure}
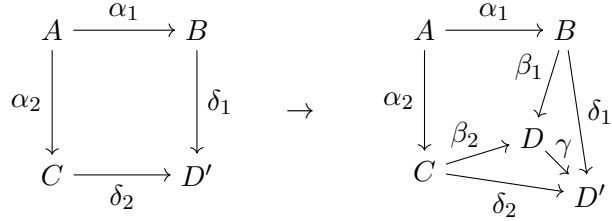

\begin{defi}[Representative squares~\cite{CRS}]\label{def-rsq}%
A class $\kappa$ of commuting squares in a category $\catname C$ is \emph{represen\-tative} if $\kappa$ satisfies the following condition:
for each commuting square $(\alpha_1, \alpha_2, \delta_1, \delta_2)$ in $\catname C$
there exists a commuting square $(\alpha_1, \alpha_2, \beta_1, \beta_2)$ in $\kappa$ and an arrow $\gamma$,
such that $\delta_1 = \beta_1;\gamma,\ \delta_2 = \beta_2;\gamma$.
This situation is depicted in \Cref{repr-sq-init}.

For two arrows $\alpha_1 \colon A \to B,\ \alpha_2 \colon A \to C$,
we define $\kappa(\alpha_1,\alpha_2)$ as the set of pairs of arrows $(\beta_1,\beta_2)$ which,
together with $\alpha_1,\alpha_2$, form representative squares in $\kappa$.
\end{defi}

The original paper on reactive systems~\cite{LM00} used the (more restrictive)
notion of idem pushouts \mbox{instead} of representative
squares. Unfortunately, the universal property of idem pushouts leads
to complications, in particular for cospan categories, where one has
to resort to the theory of bicategories in order to be able to express
this requirement. For the purposes of this paper, we stick to the
simpler notion of representative squares, in order to keep our
results independent of the concrete class of squares chosen.

The question arises which constructions yield suitable classes of representative squares, ideally with finite $\kappa(\alpha_1,\alpha_2)$, in order to represent all possible contexts $\delta_1,\delta_2$ with a finite set of representative contexts $\beta_1,\beta_2$.
Pushouts can be used when they exist~\cite{DBC-CRS}, however, they do not exist for $\ILC(\catname\Graphfin)$.%
% \footnote{Consider the cospan $n \colon \emptyset \rightarrow N \leftarrow \emptyset$, where $N$ is a single isolated node. No pushout exists for $(n,n)$.}

For cospan categories over adhesive categories, borrowed context
diagrams --- initially introduced as an extension of DPO
rewriting~\cite{EK04} --- can be used as representative squares. Before
we can introduce such diagrams, we first need the notion of jointly
epi.

\begin{defi}[Jointly epi]
  A pair of arrows $f \colon B \to D,\ g \colon C \to D$ is
  \emph{jointly epi} (\emph{JE}) if for each pair of arrows
  $d_1, d_2 \colon D \to E$ the following holds: if $f;d_1 = f;d_2$
  and $g;d_1 = g;d_2$, then $d_1 = d_2$.
\end{defi}

In $\catname\Graphfin$ jointly epi equals jointly surjective, meaning
that each node or edge of $D$ is required to have a preimage under $f$
or $g$ or both (it contains only images of $B$ or~$C$).

This criterion is similar to, but weaker than a pushout:
For jointly epi graph morphisms $d_1 \colon B \to D,\ d_2 \colon C \to D$,
there are no restrictions on which elements of $B,C$ can be merged in~$D$.
However, in a pushout constructed from morphisms $a_1 \colon A \to B,\ a_2 \colon A \to C$,
elements in~$D$ can (and must) only be merged if they have a common preimage in $A$.
(Hence every pushout generates a pair of jointly epi arrows, but not vice versa.)

\begin{defi}[Borrowed context diagram~\cite{DBC-CRS}]
  A commuting diagram in the category $\ILC(\catname C)$, where
  $\catname C$ is adhesive, is a \emph{borrowed context diagram}
  whenever it has the form of the diagram shown in \Cref{fig-bc-diag},
  and the four squares in the base category $\catname C$ are pushout
  (PO), pullback (PB) or jointly~epi (JE) as indicated. In particular
  $L\rightarrowtail G^+$, $G\rightarrowtail G^+$ must be jointly epi.
\end{defi}

\Cref{fig-bc-venn} shows a more concrete version of
\Cref{fig-bc-diag}, where graphs and their overlaps are depicted by
Venn diagrams (assuming that all morphisms are injective). Because of
the two pushout squares, this diagram can be interpreted as
composition of cospans $a;f = \ell;c = D \rightarrow G^+ \leftarrow K$
with extra conditions on the top left and the bottom right square.
The top left square fixes an overlap $G^+$ of $L$ and $G$, while $D$
is contained in the intersection of $L$ and $G$ (shown as a hatched
area). Being jointly epi ensures that it really is an overlap and does
not contain unrelated elements.
%\todo{wir sollten über $D$ reden. (unten im cospan $a;f=\ell;c$ ists
%schon geändert)}
The top right pushout corresponds to the left
pushout of a DPO rewriting diagram. It contains a total match of $L$
in $G^+$.  Then, the bottom left pushout gives us the minimal borrowed
context $F$ such that applying the rule becomes possible.
%The top left and the bottom left squares together ensure that the contexts to be considered are not larger than necessary.
The bottom right pullback ensures that the interface $K$ is as large
as possible.  %\todo{\Cref{fig-bc-venn} referenzieren}

We
will discuss an example of a borrowed context diagram below
(\Cref{repr-step-example}).  For additional examples, we refer to~\cite{EK04}.

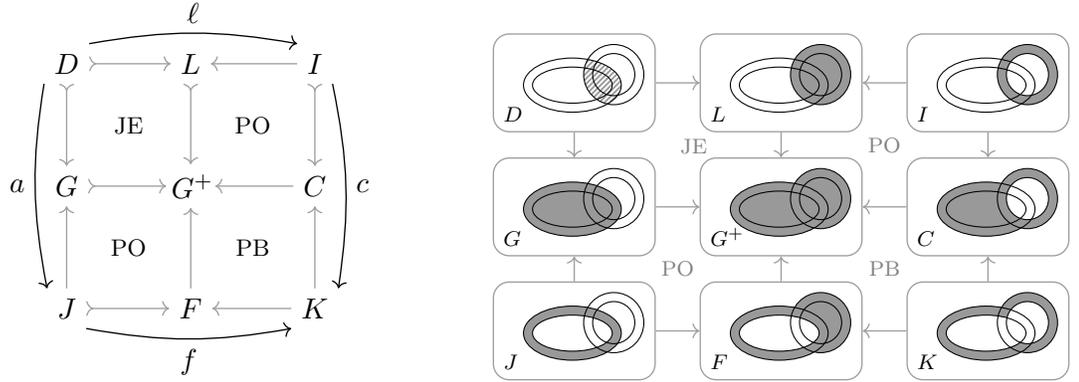
\begin{figure}[t]
  \begin{subfigure}[b]{0.45\textwidth}
    \centering\begin{tikzpicture}[x=1.10cm,y=1.10cm]% PAGEBREAK-ADJUST
      \def\sqw{1.5}
      \def\sqh{1.48}
      \pgfmathsetmacro\sqww{2*\sqw}
      \pgfmathsetmacro\sqhh{2*\sqh}

      % very tight fit (the normal bounding box leaves a bit of space above and below the top and bottom cospan labels, which makes it hard to properly align it with the rounded boxes in the other figure)
      \path[use as bounding box] (-0.7,0.8) rectangle (2*\sqw+0.7,-2*\sqh-0.87);

      \node (d) at (0,0) {$D$};
      \node (l) at (\sqw,0) {$L$};
      \node (i) at (\sqww,0) {$I$};
      \node (g) at (0,-\sqh) {$G$};
      \node (gp) at (\sqw,-\sqh) {$\mkern-3mu G^+\mkern-5mu$};
      \node (c) at (\sqww,-\sqh) {$C$};
      \node (j) at (0,-\sqhh) {$J$};
      \node (f) at (\sqw,-\sqhh) {$F$};
      \node (k) at (\sqww,-\sqhh) {$K$};

      \def\sqlfs{\footnotesize}
      \node at (0.5*\sqw,-0.5*\sqh) {\sqlfs\textup{JE}};
      \node at (1.5*\sqw,-0.5*\sqh) {\sqlfs\textup{PO}};
      \node at (0.5*\sqw,-1.5*\sqh) {\sqlfs\textup{PO}};
      \node at (1.5*\sqw,-1.5*\sqh) {\sqlfs\textup{PB}};

      \draw[cospanintmono]
        (d) edge (l) edge (g)
        (l) edge (gp)
        (g) edge (gp)
        (i) edge (c)
        (j) edge (f);
      \draw[cospanint]
        (j) edge (g)
        (k) edge (f) edge (c)
        (f) edge (gp)
        (i) edge (l)
        (c) edge (gp);

      \draw[cospanarr] (d.north east) to[bend left =10] node[above]{$\ell$} (i.north west);
      \draw[cospanarr] (d.south west) to[bend right=10] node[left ]{$a$} (j.north west);
      \draw[cospanarr] (i.south east) to[bend left =10] node[right]{$c$} (k.north east);
      \draw[cospanarr] (j.south east) to[bend right=10] node[below]{$f$} (k.south west);
    \end{tikzpicture}%
    \caption{Structure of a borrowed context diagram.
    The inner, lighter arrows are morphisms of the base category $\catname C$,
    while the outer arrows are morphisms of $\ILC(\catname C)$.
    }%
    \label{fig-bc-diag}%
  \end{subfigure}\hfill%\hspace{1cm}
  \begin{subfigure}[b]{0.52\textwidth}
    \centering
    \begin{tikzpicture}
      \def\heightAdjustA{1.00}% PAGEBREAK-ADJUST
      \def\heightAdjustB{1.00}% PAGEBREAK-ADJUST
      % general diagram layout (3x3)
      \pgfmathsetmacro\posX{2.75}
      \pgfmathsetmacro\posY{-1.65}
      \pgfmathsetmacro\boxSxp{1.15}
      \pgfmathsetmacro\boxSxm{1.0}
      \pgfmathsetmacro\boxSyp{0.68*\heightAdjustB}
      \pgfmathsetmacro\boxSym{0.62*\heightAdjustB}
      \pgfmathsetmacro\boxW{\boxSxp+\boxSxm}
      \pgfmathsetmacro\boxH{\boxSyp+\boxSym}

      % graph/circle shapes
      \def\circGinner{((0.05,0) ellipse[x radius=0.53, y radius=0.24]}
      \def\circGouter{((0.05,0) ellipse[x radius=0.65, y radius=0.36]}
      \def\circLinner{((0.6,0.14) circle[radius=0.28]}
      \def\circLouter{((0.6,0.14) circle[radius=0.4]}
      \def\allArea{(-\boxSxm,-\boxSym) rectangle (\boxSxp,\boxSyp)} % for clipping

      % draw all outlines
      \newcommand{\drawgl}{%
        \draw \circGinner;
        \draw \circGouter;
        \draw \circLinner;
        \draw \circLouter;
      }

      % fill style
      \tikzset{vennfill/.style={fill=black!40}}

      % shorthand for drawing both circle fills
      % (usually, this follows a series of clip commands)
      \newcommand{\fillboth}{%
        \fill[vennfill] \circGouter;
        \fill[vennfill] \circLouter;
      }

      % draw a G-L overlap (with box around it and so on)
      % at given position with specified fill commands
      \newcommand{\boxat}[5]{% x y name label fillcommands
        \begin{scope}[shift={(#1*\posX,#2*\posY*\heightAdjustA)}]
          \node (#3) at ($(\boxSxp,\boxSyp)!.5!(-\boxSxm,-\boxSym)$) [draw,roundbox,cospanintcommon,minimum width=\boxW cm, minimum height=\boxH cm] {};
          \node[anchor=south west] at (-\boxSxm,-\boxSym) {$\scriptstyle{#4}$};

          \begin{scope}[even odd rule]
            #5
          \end{scope}
          \drawgl
        \end{scope}
      }

      \boxat00DD{% fill the intersection of L and G, but in a different shade
        \clip \circLouter;
        \clip \circGouter;
        \foreach \d in {0,1,...,20} {
          \draw[black!40, line width=0.7pt] (0, -1cm + 0.3pt + \d*2pt) -- +(1cm,1cm);
        }}
        %\fill[fill=black!30] \circGouter;}

      \boxat10LL{% fill all of L
        \fill[vennfill] \circLouter;}

      \boxat20II{% fill only the interface of L
        \clip \circLinner \allArea;
        \fill[vennfill] \circLouter;}

      \boxat01GG{% fill all of G
        \fill[vennfill] \circGouter;}

      \boxat11{Gp}{G^+}{% fill all of L and G
        \fillboth}

      \boxat21CC{% fill all except the inside of L
        \clip \circLinner \allArea;
        \fillboth}

      \boxat02JJ{% fill only the interface of G
        \clip \circGinner \allArea;
        \fill[vennfill] \circGouter;}

      \boxat12FF{% fill all except the inside of G
        \clip \circGinner \allArea;
        \fillboth}

      \boxat22KK{% fill both interfaces, but not the inside
        \clip \circGinner \allArea;
        \clip \circLinner \allArea;
        \fillboth}

      % connecting the boxes
      \def\arrowfromto#1#2#3{
        % ignore mono for readability (no cospanint#3)
        \draw[cospanint] (#1) edge (#2);
      }
      \def\cospanfromto#1#2#3{% left middle right | name namelabelpos anchor namelabelopts
        \arrowfromto{#1}{#2}{mono}
        \arrowfromto{#3}{#2}{}
      }
      \def\polabels#1#2#3#4{
        \def\sqlfs{\scriptsize\color{gray}}
        \node at ($(boxnull)!0.5!(Gp)$) {\sqlfs\textup{#1}};
        \node at ($(L)!0.5!(C)$) {\sqlfs\textup{#2}};
        \node at ($(G)!0.5!(F)$) {\sqlfs\textup{#3}};
        \node at ($(Gp)!0.5!(K)$) {\sqlfs\textup{#4}};
      }

      \cospanfromto{D}{L}{I} %{\ell}{above}{north}{bend left=10}
      \cospanfromto{D}{G}{J} %{a}{left}{west}{bend right=10}
      % \draw[cospanarr] (boxnull.north east) to[bend left=10] node[above]{$\ell$} (I.north west);
      % \draw[cospanarr] (boxnull.south west) to[bend right=12] node[left]{$a$} (J.north west);
      % \draw[cospanarr] (J.south east) to[bend right=10] node[below]{$f$} (K.south west);
      % \draw[cospanarr] (I.south east) to[bend left=12] node[right]{$c$} (K.north east);
      \cospanfromto{J}{F}{K}
      \cospanfromto{I}{C}{K}
      \cospanfromto{L}{Gp}{F}
      \cospanfromto{G}{Gp}{C}
      \polabels{JE}{PO}{PO}{PB}
    \end{tikzpicture}
    \caption{Borrowed context diagrams represented as Venn diagrams.
    The outer circles represent graphs $L,G$, and
    the area between the inner and outer circles represents their interfaces $I,J$.}%
    \label{fig-bc-venn}
  \end{subfigure}
  \caption{Borrowed context diagrams}
\end{figure}

For cospan categories over adhesive categories, borrowed context
diagrams form a represen\-tative class of
squares~\cite{CRS}. Furthermore, for some categories (such as
$\catname\Graphfin$), there are --- up to isomorphism --- only finitely
many jointly epi squares for a given span of monos and hence only
finitely many borrowed context diagrams given $a,\ell$ (since pushout
complements along monos in adhesive categories are unique up to
isomorphism).

  This motivates the following finiteness assumption that we will refer
  to in this paper: given $a,\ell$, we require that $\kappa(a,\ell)$
  is finite. (\textsc{Fin})
  %\todoign{Finiteness assumption erweitern ... wobei, wenn wir für \Cref{cec-cb} eh infinite brauchen, hilft uns die zweite finiteness assumption sowieso nicht weiter}

\subsection{Representative Steps}\label{subsec-rstep}

It is possible to define a reaction relation based on representative squares.
By requiring that the left square is representative, we ensure that the contexts $\hat f$ are not larger than necessary:

\begin{defi}[Representative step (without conditions)~\cite{DBC-CRS}]\label{def-reprstep-nc}%
Let \mbox{$a \colon \obnull \to J,\ \hat f \colon J \to K,$} $a' \colon \obnull \to K$ be arrows.
We write $a \ncrstep{\hat f} a'$ if a context step $a \nccstep{\hat f} a'$ is possible
(i.e.\ $a;\hat f \leadsto a'$, i.e.\ for some rule $(\ell,r)$ and some arrow $\hat c$ we have $a;\hat f = \ell;\hat c$ and $r;\hat c = a'$)\linebreak
and additionally $\kappa(a,\ell) \ni (\hat f,\hat c)$ (i.e.\ the arrows $(a,\ell,\hat f,\hat c)$ form a representative square).
Such steps are called \emph{representative steps}.
\end{defi}

\begin{rem}\label{repr-zurueckfuehren}%
  \Cref{def-rsq,def-ctxstep} imply that every context step
  $a \nccstep{f} a'$ (left diagram)
  can be reduced to a representative step
  $a \ncrstep{\hat f} r;\hat c$ (right diagram), a fact used in the
  proofs.

\noindent\mbox{}\hskip 0pt plus 1filll\begin{tikzpicture}
  \def\sqw{2}
  \def\sqh{1.9}
  \def\sqhsp{0.4} % how much to inset C, B compared to the full D' square size
  \def\sqhspi{0.35} % how much to additionally inset D (in addition to \sqhsp)
  \def\sqhsr{0.2} % shrink of the left side. was 0 in masters thesis, but a tiny offset is acceptable (it should however not be exactly halfway between D and D', because this might suggest that D' is "split" to some smaller object D and some larger-than-D' object that is also called D'
  \def\sqrw{1.25} % right square width
  \def\sqtms{-0.25} % shift of middle node in top/left diagram
  \def\sqtmsy{-0.35} % shift of middle node in top/left diagram, y direction

  % %%% LEFT - PRE R-STEP %%%
  \begin{scope}[shift={(0,0)}]
    \node (a) at (0,0) {$\obnull$};
    \node (b) at (\sqw-\sqhsr+\sqtms,0) {$I$};
    \node (c) at (0,-\sqh+\sqhsr-\sqtmsy) {$J$};
    \node (ds) at (\sqw-\sqhsr+\sqtms,-\sqh+\sqhsr-\sqtmsy) {$K$};
    \draw[->] (a) -- node[above]{$\ell$} (b);
    \draw[->] (a) -- node[left]{$a$} (c);

    \draw[->] (b) -- node[right,pos=0.4]{$c$} (ds);
    \draw[->] (c) -- node[above]{$f$} (ds);

    \node (r0) at (\sqw-\sqhsr + \sqrw,0) {$\obnull$};
    \draw[->] (r0) -- node[above]{$r$} (b);
    \draw[->] (r0) -- node[right]{$a'$} (ds);
  \end{scope}

  \node at (4.5,-0.45*\sqh) {$\rightarrow$};

  % %%% RIGHT %%%
  \begin{scope}[shift={(6,0)}]
    \node (a) at (0,0) {$\obnull$};
    \node (b) at (\sqw-\sqhsp,0) {$I$};
    \node (c) at (0,-\sqh+\sqhsp) {$J$};
    \node (d) at (\sqw-\sqhsp-\sqhspi,-\sqh+\sqhsp+\sqhspi) {$K'$};
    \node (ds) at (\sqw,-\sqh) {$K$};
    \draw[->] (a) -- node[above]{$\ell$} (b);
    \draw[->] (a) -- node[left]{$a$} (c);
    \draw[->] (b) -- node[left,pos=0.3]{$\hat c$} (d);
    \draw[->] (c) -- node[above,pos=0.3]{$\hat f$} (d);

    \draw[->] (d.center)+(5pt,-5pt) -- (ds);
    \node at ($(d.center)+(12pt,-3pt)$) {$\hat g$};
    \draw[->] (b) to[bend left=15] node[right,pos=0.4]{$c$} (ds);
    \draw[->] (c) -- node[below]{$f$} (ds);

    \node (r0) at (\sqw-\sqhsr + \sqrw,0) {$\obnull$};
    \draw[->] (r0) -- node[above]{$r$} (b);
    \draw[->] (r0) -- node[right]{$a'$} (ds);
  \end{scope}

\end{tikzpicture}\hskip 0pt plus 1filll\mbox{} % chktex 31

For this, we construct the representative square $(a,\ell,\hat f,\hat c) \in \kappa$ (which, according to \Cref{def-rsq}, always exists) from the square $(a,\ell,f,c)$ describing the context step.
We obtain arrows $\hat f, \hat c$ and an arrow $\hat g$ which completes $\hat f, \hat c$ to $f,c$ (i.e.\ $\hat f ; \hat g = f,\ \hat c ; \hat g = c$).
\end{rem}

\begin{exa}[Representative steps]\label{repr-step-example}
  Reconsider the reactive system described in Example~\ref{ex-cb-unreliable-pre}, %dvips breaks link
  i.e., a message $m$ can be transferred along a channel $\chanLabel$.
  One possible context step allows a channel $\inlinechan{chan}$ to borrow a message $\inlinechan{msgleft}$ (depicted in \Cref{repr-step-example-bcdfig}) and do a transfer:
  ${(\emptyset \rightarrow \inlinechan{chan} \leftarrow \inlinechan{})}
  \nccstep{(\inlinechan{} \rightarrow \loopOnArrow{\inlinechan{msgleft}} \leftarrow \inlinechan{})}
  {(\emptyset \rightarrow \inlinechan{chan,msgright} \leftarrow \inlinechan{})}$.

  Another possible context step is to borrow an additional message on the right node, i.e.\ %
  ${(\emptyset \rightarrow \inlinechan{chan} \leftarrow \inlinechan{})}
  \nccstep{(\inlinechan{} \rightarrow \loopOnArrow{\inlinechan{msgleft,msgright}} \leftarrow \inlinechan{})}
  {(\emptyset \rightarrow \inlinechan{chan,msgright=2} \leftarrow \inlinechan{})}$
  (depicted in \mbox{\Cref{repr-step-example-bigctxfig}}).
  Clearly, this is a valid context step, but the right message is not required by the rule, and we do not want to consider such steps in our analysis (by adding yet more messages, we obtain infinitely many context steps).

  \begin{figure}[b]
    \def\heightAdjustA{1.00}% PAGEBREAK-ADJUST
    \def\heightAdjustB{1.00}% PAGEBREAK-ADJUST
    % commands for the two figures:
      \def\drawthebox#1#2{
        \node (#2) at ($(\boxSxp,\boxSyp)!.5!(-\boxSxm,-\boxSym)$) [draw,roundbox,minimum width=\boxW cm, minimum height=\boxH cm] {};
        \node at (-0.35,-0.25) {$\scriptstyle{#1}$}; % -0.2 when using 0.35 below
      }
      \def\arrowfromto#1#2#3{
        % ignore mono for readability (no cospanint#3)
        \draw[cospanint] (#1) edge (#2);
      }
      \def\cospanfromto#1#2#3{% left middle right | name namelabelpos anchor namelabelopts
        \arrowfromto{#1}{#2}{mono}
        \arrowfromto{#3}{#2}{}
      }
      \def\cospanfromtoarr#1#2#3#4#5#6#7{% left middle right | name namelabelpos anchor namelabelopts
        \cospanfromto{#1}{#2}{#3}
        \draw[cospanarr] (#1.#6) to[#7] node[#5]{$#4$} (#3.#6);
      }
      \def\smolnodes{
        \node[gn] (l) at (0,0) {};
        \node[gn] (r) at (0.5,0) {};
      }
      \def\smolchan{
        \smolnodes
        \draw[gedge] (l) to node[below]{$\chanLabel$} (r);
      }
      \def\leftmloopt{\draw[gedge] (l.150) to[loop,out=135+30,in=135-30,distance=0.3cm] node[left]{$m\mkern-5mu$} (l.120);} % chktex 8
      \def\rightmloopt{\draw[gedge] (r.60) to[loop,out=45+30,in=45-30,distance=0.3cm] node[right]{$\mkern-2mu m$} (r.30);} % chktex 8
      \def\boxat#1#2#3#4#5{% gridx gridy label boxname content
        \begin{scope}[shift={(#1*\gridx,#2*\gridy)}]
          \drawthebox{#3}{#4}
          #5
        \end{scope}
      }
      \def\gridparams{
        \pgfmathsetmacro\gridx{2.1}
        \pgfmathsetmacro\gridy{-1.4*\heightAdjustA}
        \pgfmathsetmacro\boxSxp{1.1}
        \pgfmathsetmacro\boxSxm{0.6}
        \pgfmathsetmacro\boxSyp{0.5*\heightAdjustB}
        \pgfmathsetmacro\boxSym{0.5*\heightAdjustB}
        \pgfmathsetmacro\boxW{\boxSxp+\boxSxm}
        \pgfmathsetmacro\boxH{\boxSyp+\boxSym}
      }
      \def\polabels#1#2#3#4{
        \def\sqlfs{\scriptsize\color{gray}}
        \node at ($(boxnull)!0.5!(Gp)$) {\sqlfs\textup{#1}};
        \node at ($(L)!0.5!(C)$) {\sqlfs\textup{#2}};
        \node at ($(G)!0.5!(F)$) {\sqlfs\textup{#3}};
        \node at ($(Gp)!0.5!(K)$) {\sqlfs\textup{#4}};
      }
    \begin{subfigure}[b]{0.48\textwidth}
      \centering
      \begin{tikzpicture}[gedge/.append style={font=\scriptsize}, cospanarr/.append style={font=\scriptsize}]
        \gridparams
        \pgfmathsetmacro\boxSxp{0.9}
        \pgfmathsetmacro\boxW{\boxSxp+\boxSxm}
        \boxat00{\obnull}{boxnull}{}
        \boxat10{L}{L}{\smolchan\leftmloopt}
        \boxat20{I}{I}{\smolchan}
        \boxat21{C}{C}{\smolchan}
        \boxat01{G}{G}{\smolchan}
        \boxat11{\;G^+}{Gp}{\smolchan\leftmloopt}
        \boxat02{J}{J}{\smolnodes}
        \boxat12{F}{F}{\smolnodes\leftmloopt}
        \boxat22{K}{K}{\smolnodes}
        \cospanfromto{boxnull}{L}{I} %{\ell}{above}{north}{bend left=10}
        \cospanfromto{boxnull}{G}{J} %{a}{left}{west}{bend right=10}
        \draw[cospanarr] (boxnull.north east) to[bend left=10] node[above]{$\ell$} (I.north west);
        \draw[cospanarr] (boxnull.south west) to[bend right=12] node[left]{$a$} (J.north west);
        \draw[cospanarr] (J.south east) to[bend right=10] node[below]{$f$} (K.south west);
        \draw[cospanarr] (I.south east) to[bend left=12] node[right]{$c$} (K.north east);
        \cospanfromto{J}{F}{K}
        \cospanfromto{I}{C}{K}
        \cospanfromto{L}{Gp}{F}
        \cospanfromto{G}{Gp}{C}
        \polabels{JE}{PO}{PO}{PB}
      \end{tikzpicture}
      \caption{Borrowed context diagram for a channel \mbox{borrowing} a message on the left node\\~}%
      \label{repr-step-example-bcdfig}
    \end{subfigure}\hfill%\hspace{1cm}
    \begin{subfigure}[b]{0.48\textwidth}
      \centering
      \begin{tikzpicture}[gedge/.append style={font=\scriptsize}, cospanarr/.append style={font=\scriptsize}]
        \gridparams
        \boxat00{\obnull}{boxnull}{}
        \boxat10{L}{L}{\smolchan\leftmloopt}
        \boxat20{I}{I}{\smolchan}
        \boxat21{C}{C}{\smolchan\rightmloopt}
        \boxat01{G}{G}{\smolchan}
        \boxat11{\;G^+}{Gp}{\smolchan\leftmloopt\rightmloopt}
        \boxat02{J}{J}{\smolnodes}
        \boxat12{F}{F}{\smolnodes\leftmloopt\rightmloopt}
        \boxat22{K}{K}{\smolnodes}
        \cospanfromto{boxnull}{L}{I} %{\ell}{above}{north}{bend left=10}
        \cospanfromto{boxnull}{G}{J} %{a}{left}{west}{bend right=10}
        \draw[cospanarr] (boxnull.north east) to[bend left=10] node[above]{$\ell$} (I.north west);
        \draw[cospanarr] (boxnull.south west) to[bend right=12] node[left]{$a$} (J.north west);
        \draw[cospanarr] (J.south east) to[bend right=10] node[below]{$f$} (K.south west);
        \draw[cospanarr] (I.south east) to[bend left=12] node[right]{$c$} (K.north east);
        \cospanfromto{J}{F}{K}
        \cospanfromto{I}{C}{K}
        \cospanfromto{L}{Gp}{F}
        \cospanfromto{G}{Gp}{C}
        \polabels{}{PO}{PO}{}
      \end{tikzpicture}
      \caption{Commuting diagram for a channel additionally borrowing a message on the right, which is not needed for the reaction}%
      \label{repr-step-example-bigctxfig}
    \end{subfigure}
    \caption{Diagrams for the two steps described in \Cref{repr-step-example}.}
  \end{figure}
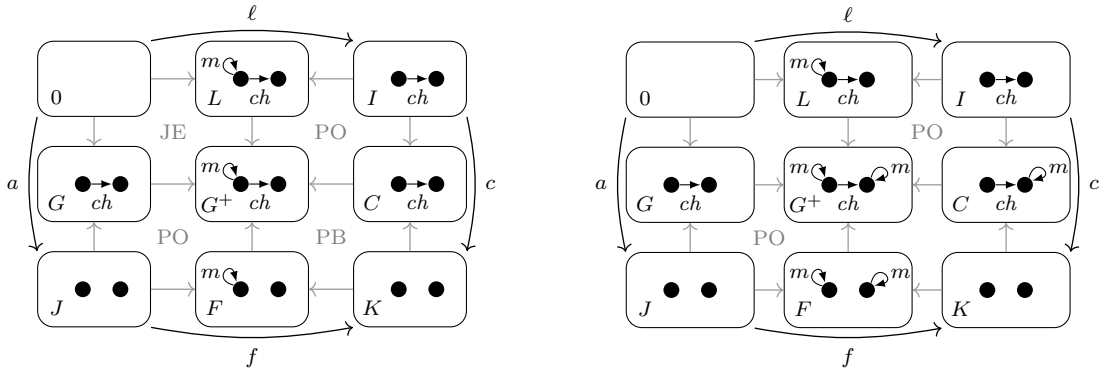

  However, the second context step is not a representative step
  (assuming that representative squares correspond to borrowed context
  diagrams).
  We try to construct a borrowed context diagram: First we fill in the graphs given by $a$, $f$ and $\ell$,
  then we construct the bottom left pushout, we obtain $G^+ = \inlinechan{chan,msgleft,msgright}$ as depicted in \Cref{repr-step-example-bigctxfig}.
  Then however the top left square is not jointly epi, since neither $L = \inlinechan{chan,msgleft}$ (from $\ell$) nor $G = \inlinechan{chan}$ (from $a$) provide a preimage for the right $m$-loop $\inlinechan{msgright}$.

  On the other hand, the first context step is representative, since there $G^+ = \inlinechan{chan,msgleft}$ does not contain the problematic right $m$-loop and it is possible to complete the borrowed context diagram as shown in \Cref{repr-step-example-bcdfig}.
  (To obtain the result of the context step, the right-hand side $a'$ is constructed just as for context steps (see \Cref{ex-cb-unreliable-pre}), which is not depicted here.)
\end{exa}

In a \emph{semi-saturated bisimulation}, $\rightarrow_R$-steps are
answered by $\rightarrow_C$-steps (for every \linebreak\mbox{$(a,b) \in R$} and step
$a \ncrstep{f} a'$ there is $b \nccstep{f} b'$ such that
$(a',b') \in R$ and vice versa).  The resulting bisimilarity $\simR$ is
identical~\cite{DBC-CRS} to saturated bisimilarity (i.e.
$\simR = \simC$) and therefore also a congruence. Whenever
  (\textsc{Fin}) holds, $\sim_R$ is amenable to mechanization, since
  we have to consider only finitely many $\rightarrow_R$-steps
  ($\rightarrow_R$ is finitely branching).

\begin{rem}
Note that answering $\rightarrow_R$-steps with $\rightarrow_R$-steps gives a different, finer notion of behavioural equivalence than answering $\rightarrow_R$-steps with $\rightarrow_C$-steps.
As an example, consider the reactive system with two rules:
$\inlinerr{a,b} \Rightarrow \inlinerr{a,b}$ and $\inlinerr{c}
\Rightarrow \inlinerr{c}$, where the single node is in the rule interface.
Both rules replace the graphs with themselves,
hence, any rewriting step does not change the graph at all.

By exhaustive enumeration of all representative steps, it is easy to see that $\inlinerr{a,c}, \inlinerr{c}$ are semi-saturated bisimilar.
It is important to keep in mind that the answering step does not need to use the same rule.
For example, a step $\inlinerr{a,c} \ncrstep{\loopOnArrow{\inlinerr{b}}} \inlinerr{a,b,c}$ (using the first rule) can be answered by a step $\inlinerr{c} \nccstep{\loopOnArrow{\inlinerr{b}}} \inlinerr{b,c}$ (using the second rule).
Because the resulting graphs both contain a $c$-loop, they are also bisimilar, since any subsequent steps (using either rule) can always be answered by applying the second rule.

However, the step $\inlinerr{c} \nccstep{\loopOnArrow{\inlinerr{b}}} \inlinerr{b,c}$ is not a representative step, because it borrows more than what is necessary to apply the second rule.
There is also no other representative step that originates from $\inlinerr{c}$ and borrows exactly $\inlinerr{b}$.
Hence, under a notion of bisimulation where~$\rightarrow_R$-steps are answered by $\rightarrow_R$-steps, $\inlinerr{a,c},\inlinerr{c}$ are not bisimilar.
\end{rem}

\section{Conditions for Reactive Systems}%
\label{sec:reasys-cond}

%\subsection{Motivation}

The reactive systems defined so far cannot represent rules where a
certain component is required to be absent: whenever a reaction
$a \leadsto a'$ is possible, a reaction $a;c \leadsto a';c$ (with
additional context $c$) is also possible, with no method to prevent
this.  Restricting rule applications can be useful, e.g.\ to model
access to a shared resource, which may only be accessed if no other
entity is currently using it.

For graph transformation systems, application conditions with a
first-order logic flavour have been studied extensively (e.g.\ in~\cite{NegativeAC,PennemannNAC}) and generalized to reactive systems in~\cite{CRS}. If we interpret such conditions in
  $\ILC(\catname\Graphfin)$, we obtain a logic that subsumes
  first-order logic (for more details on expressiveness
  see~\cite{CRS}).

In this section, we summarize the definitions from~\cite{CRS} and
define shifting of conditions as partial evaluation.  We then
summarize the changes that are necessary to extend reactive systems
with conditions. We illustrate the concepts of this chapter
with various examples. An example for conditional reactive systems will be
discussed later (\Cref{ex-cb-unreliable}).  For an additional example, we refer to~\cite{CRS}.

\subsection{Conditions and Satisfaction}

\begin{defi}[Condition~\cite{CRS}]\label{def-condition}
  Let $\catname C$ be a category and $A$ be an object of $\catname C$.
  The set of conditions over $A$ is defined inductively as:
\begin{itemize}
\item $\condtrue_A \defeq (A, \forall, \emptyset)$ and
$\condfalse_A \defeq (A, \exists, \emptyset)$ are conditions over $A$ (base case)
\item $\mathcal{A} = (A, \mathcal{Q}, S)$ is a condition over $A$,
   where
\begin{itemize}
\item $A = \Ro(\mathcal{A})$ is the \emph{root object} of $\mathcal A$,
\item $\mathcal{Q} \in \{\forall, \exists\}$ is a quantifier and
\item $S$ is a finite set of pairs $(h, \mathcal{A}^\prime)$, where
  $h \colon A \to A'$ is an arrow and
  $\mathcal A'$ is a condition over~$A'$.
  % + child condition(?)
\end{itemize}
\end{itemize}
\end{defi}

\noindent% used to avoid mis-association of "Note that..." to the preceding bullet point
Note that conditions can be represented as finite trees.

\begin{defi}[Satisfaction~\cite{CRS}]\label{def-cond-satisf}
Let $\mathcal A$ be a condition over $A$. For an arrow $a \colon A \to B$ and a condition $\mathcal A$  we define the satisfaction relation $a \models \mathcal A$ as follows:
\begin{itemize}
\item $a \models (A, \forall, S)$ iff for every pair $(h, \mathcal{A}^\prime) \in S$ and every arrow $g \colon \Ro(\mathcal{A}^\prime) \to B$ we have: if~$a = h;g$, then $g \models \mathcal{A}^\prime$.
\item $a \models (A, \exists, S)$ iff there exists a pair $(h, \mathcal{A}^\prime) \in S$ and an arrow $g \colon \Ro(\mathcal{A}^\prime) \to B$ such that $a = h;g$ and $g \models \mathcal{A}^\prime$.
\end{itemize}

\noindent
From the above it follows that
$\condtrue_A$ is satisfied by every arrow with domain $A$, and
$\condfalse_A$ is satisfied by no arrow.

We write $\mathcal A \models \mathcal B$ ($\mathcal{A}$ implies $\mathcal{B}$) if for every arrow $c$ with $\dom(c) = \Ro(\mathcal A) = \Ro(\mathcal B)$ we have: if $c \models \mathcal A$, then $c \models \mathcal B$.
Two conditions are equivalent ($\mathcal A \equiv \mathcal B$) if $\mathcal A \models \mathcal B$ and $\mathcal B \models \mathcal A$.
\end{defi}

With satisfaction defined in this way,
universal conditions implicitly implement conjunction, i.e.,
$(A, \forall, \{ (f_1, \mathcal A_1), (f_2, \mathcal A_2), \dots \})$
can be understood as
$\forall(f_1, \mathcal A_1) \land \forall(f_2, \mathcal A_2) \land \dots$
\linebreak[4]
(with a meaning analogous to the conditions of~\cite{NegativeAC,PennemannNAC});
similarly, existential conditions implement disjunction
($\exists(f_1, \mathcal A_1) \lor \exists(f_2, \mathcal A_2) \land \dots$).

Based on this, we can define negation of conditions, and conjunctions and disjunctions of arbitrarily quantified conditions such as $\exists(\dots) \land \exists(\dots)$ (note the flipped logical connective):

\begin{prop}[Boolean operations~\cite{CRS}]\label{cond-bool-ops}%
Consider the following Boolean operations: % on conditions:

\begin{itemize}
\item  $\neg(A, \forall, S) \defeq (A, \exists, \{ (h, \neg\mathcal{A}^\prime) \mid (h, \mathcal{A}^\prime) \in S \})$, $\neg(A, \exists, S) \defeq (A, \forall, \{ (h, \neg\mathcal{A}^\prime) \mid (h, \mathcal{A}^\prime) \in S \})$
\item $\mathcal{A} \lor \mathcal{B} \defeq (A, \exists, \{ (\id_A, \mathcal{A}), (\id_A, \mathcal{B}) \})$ for two conditions $\mathcal{A}, \mathcal{B}$ over $A$
\item $\mathcal{A} \land \mathcal{B} \defeq (A, \forall, \{ (\id_A, \mathcal{A}), (\id_A, \mathcal{B}) \})$ for two conditions $\mathcal{A}, \mathcal{B}$ over $A$
\end{itemize}

\noindent
\pagebreak[2]% avoid breaking the list of operations before its end
These operations satisfy the standard laws of propositional logic, i.e.
$a \models \neg\mathcal A$ if and only if $a \notmodels \mathcal A$;
$a \models \mathcal A \lor \mathcal B$ if and only if $a \models \mathcal A \;\lor\; a \models \mathcal B$,
analogously for $\mathcal A \land \mathcal B$.
Conjunction and disjunction can be extended in the obvious way to arbitrary, rather than binary, conjunctions and disjunctions.
\end{prop}

% \begin{proof}
% We show the correctness of $\mathcal{A} \land \mathcal{B}$ only, the others are similar.
% For each arrow $a$ there is a decomposition $a = {\id_A};a$, and there is no other arrow $b \neq a$ that also satisfies $a = {\id_A};b$.
% For $\mathcal{A} \land \mathcal{B} = (A, \forall, \{ (\id_A, \mathcal{A}), (\id_A, \mathcal{B}) \})$, applying the definition of satisfaction yields:
% $a$~satisfies $\mathcal{A} \land \mathcal{B}$ if and only if for each pair $(h, \mathcal{A}^\prime) \in S$ we have $a \models \mathcal{A}^\prime$ -- which is exactly the case if $a \models \mathcal{A}$ and $a \models \mathcal{B}$.
% \end{proof}

\begin{exa}[Examples of conditions]\label{examples-conditions}
\quad% don't put the list bullet in the middle of the line; using \quad instead of ~ because of weird spacing otherwise
\begin{itemize}
\item $\mathcal{A}_{node}$ recognizes graphs that contain at least one node:
\[
\mathcal{A}_{node}
=
\Big( \emptyset, \exists, \Big\{
( \emptyset \rightarrow \gn \leftarrow \gn,~ \condtrue_{\gn} )
\Big\} \Big)
\]
The condition is checked as follows: any arrow that satisfies the condition
must be decomposable into two arrows, the first of which is given in the condition and contributes the required node,
and the second optionally provides additional elements.
Since the output interface is not empty, the second arrow is free to connect edges to the required node,
i.e.\ the condition matches both isolated and non-isolated nodes.

\item The condition $\mathcal{A}_{iso}$ recognizes graphs that contain an isolated node:
\[
\mathcal{A}_{iso}
=
\Big( \emptyset, \exists, \Big\{
( \emptyset \rightarrow \gn \leftarrow \emptyset,~ \condtrue_\emptyset )
\Big\} \Big)
% =
% \Big( \emptyset, \exists, \Big\{
% \big( \emptyset \rightarrow \gn \leftarrow \emptyset,~ (\emptyset, \forall, \emptyset) \big)
% \Big\} \Big)
\]

As the outer interface of $h = \emptyset \rightarrow \gn \leftarrow \emptyset$ is empty, $h;g$ has to contain an isolated node ($g$ can only connect an edge to the node provided by $h$ if it is contained in the interface).

\item $\mathcal{A}_{ab}$ recognizes the graphs where
for all occurrences of an $a$-edge, there also exists a $b$-edge in the opposite direction:
  \newcommand{\AabInterface}{%
    \begin{tikzpicture}[baseline=-0.5ex]
      \node[gn] (l) at (0,0) {};
      \node[gn] (r) at (0.35,0) {};
    \end{tikzpicture}
  }

  \[ \begin{aligned}
    \mathcal{A}_{ab} =
    \Big( \emptyset, \forall, \Big\{ (
      \emptyset
      \ &\rightarrow\ %
        \begin{tikzpicture}[baseline=-0.5ex]
          \node[gn] (l) at (0,0) {};
          \node[gn] (r) at (0.6,0) {};
          \draw[gedge] (l) to node[below]{$a$} (r);
        \end{tikzpicture}
      \ \leftarrow\ %
      \AabInterface,
      \ \mathcal{A}^\prime
    ) \Big\} \Big) \\
  \mathcal{A}^\prime =
    \Big(
      \AabInterface,
      \exists,
      \Big\{ (
        \AabInterface
        \ &\rightarrow\ %
          \begin{tikzpicture}[baseline=-0.5ex]
            \node[gn] (l) at (0,0) {};
            \node[gn] (r) at (0.6,0) {};
            \draw[gedge] (r) to node[below]{$b$} (l);
          \end{tikzpicture}
        \ \leftarrow\ %
        \AabInterface,
        \condtrue_{\text{\textbullet\,\textbullet}}%cheat
      ) \Big\} \Big)
  \end{aligned} \]
\end{itemize}

\noindent% used only to visually distinguish this paragraph from its last list point!
Note that in the examples above, the root object of the condition is empty, since we only consider isolated conditions.
When using conditions in a transformation rule, we would use the interface of the rule instead.
This ensures that the condition is evaluated at the same position where the rule is applied, and not in any other position.
\end{exa}

\subsection{Shifting as Partial Evaluation of Conditions}

When evaluating conditions, it is sometimes known that a given context
is guaranteed to be present. In this case, a condition can be
rewritten, using representative squares, under the assumption that
this context is provided by the environment. This operation is known
as shift~\cite{PennemannNAC}:

\begin{defi}[Shift of a condition~\cite{CRS}]\label{shift-def}\label{shift-law-def}%
Given a fixed class of representative squares $\kappa$,
the \emph{shift of a condition $\mathcal A = (A, \mathcal Q, S)$ along an arrow
$c \colon A \to B$} is inductively defined as follows:
\[
\mathcal{A}_{\downarrow c} \defeq \Big( B, \mathcal Q, \Big\{
( \beta, \mathcal{A}^\prime_{\downarrow\alpha} ) \:\Big|\: (h, \mathcal{A}^\prime) \in S,\ (\alpha,\beta) \in \kappa(h, c)
\Big\} \Big)
\]

The shift operation can be understood as a partial evaluation of
$\mathcal A$ under the assumption that $c$ is already present.
It satisfies
$ c;d \models \mathcal{A} \iff d \models \mathcal{A}_{\downarrow c} $.
% that is, if an arrow is composed from $\varphi$ and some remaining
% part $c$, the composed arrow $\varphi;c$ satisfies $\mathcal A$ if and
% only if the remaining part $c$ satisfies the shifted condition
% $\mathcal{A}_{\downarrow \varphi}$.
\end{defi}

\noindent\mbox{}\hfill\begin{tikzpicture}
  \def\sqw{1.75}
  \node (z) at (-1*\sqw,0) {$Z$};
  \node (a) at (0,0) {$A$};
  \node (b) at (1*\sqw,0) {$B$};
  \node (x) at (2*\sqw,0) {$X$};
  \draw[->,fancydotted] (z) -- node[above]{$a$} (a);
  \draw[->] (a) -- node[above]{$c$} (b);
  \draw[->] (b) -- node[above]{$d$} (x);
  \node[condtri,dart tip angle=60,shape border rotate=90] at (a.south) {$\mathcal A\vphantom{{}_{\downarrow}}\mkern3mu$};
  \node[condtri,dart tip angle=60,shape border rotate=90] at (b.south) {$\mkern-5mu\mathcal A_{\downarrow c}\mkern-6mu$};
  \node at (0,-1.4) {}; % PAGEBREAK-ADJUST
\end{tikzpicture}\hfill\mbox{} % chktex 31

The typical case, which we will encounter throughout the rest of this paper,
is that a condition on the context of some arrow $a$ is given,
this arrow is then placed into some environment $c$ (which might not fully satisfy the condition, but possibly parts of it),
and we are interested in a condition that an additional context $d$ has to satisfy.
(For instance, if $\mathcal A$ requires the existence of two elements and $c$ already provides one of them,
then $d$ only needs to add the other one, which is reflected in $\mathcal A_{\downarrow c}$.)

The representation of the shifted condition may differ depending on the class of representative squares chosen.
However, no matter which class is chosen, the resulting conditions are equivalent to each other.
Furthermore, if we assume that (\textsc{Fin}) holds,
shifting a finite condition will again result in a finite condition.

Representative squares as well as shift play a major role in the
diagrammatic~proofs.
The shift operation satisfies a few equivalences that we will use in the proofs of our theorems:

\begin{thm}[Shift laws {\cite[Proposition 17]{CRS}}]\label{shift-laws}%
\begin{align*}
(\mathcal{A} \lor \mathcal{B})_{\downarrow c} &\equiv \mathcal{A}_{\downarrow c} \lor \mathcal{B}_{\downarrow c} &
\mathcal{A}_{\downarrow\id} &\equiv \mathcal{A} &
\mathcal{A} \models \mathcal{B} &\implies \mathcal{A}_{\downarrow c} \models \mathcal{B}_{\downarrow c} \\
(\mathcal{A} \land \mathcal{B})_{\downarrow c} &\equiv \mathcal{A}_{\downarrow c} \land \mathcal{B}_{\downarrow c} &
\condtrue_{\downarrow c} &\equiv \condtrue &
\mathcal{A}_{\downarrow c_1;c_2} &\equiv (\mathcal{A}_{\downarrow c_1})_{\downarrow c_2} \\
\neg(\mathcal{A}_{\downarrow c}) &\equiv (\neg\mathcal{A})_{\downarrow c} &
\condfalse_{\downarrow c} &\equiv \condfalse
\end{align*}
\end{thm}

\begin{exa}[Simplifying conditions by shifting]\label{ex-shift-simplify}
  \newcommand\mlRight{\draw[gedge] (n) to[loop,out=45,in=-45,distance=0.4cm] node[right]{\footnotesize$b$} (n);}
  \newcommand\mlLeft{\draw[gedge] (n) to[loop,out=45+90,in=-45-90,distance=0.4cm] node[left]{\footnotesize$b$} (n);} % chktex 8
  \newcommand\mbrloop{\tikz[baseline=-0.5ex]{\gn\mlRight}}

  Let the following condition $\mathcal A_{nb}$ be given, which
  requires that the interface node does not have a $b$-loop attached:
  \[ \mathcal{A}_{nb}
  = \Big( \gn, \forall, \Big\{ ( \gn \rightarrow \mbrloop \leftarrow \gn ,~ \condfalse ) \Big\}
  \Big) \]
  Furthermore let the cospan $c = \gn \rightarrow \mbrloop \leftarrow \gn$.
  We now compute the result of the shift $\mathcal{A}_{nb\downarrow c}$,
  i.e., the condition $\mathcal{A}_{nb}$ under the assumption that $c$ is already given.
  We expect the resulting condition to be equivalent to $\condfalse$, since the presence of the $b$-loop in $c$ already violates $\mathcal{A}_{nb}$.
  We will show that this is indeed the case.
  By \Cref{shift-law-def} we have:

  \begin{align*}
    \mathcal{A}_{nb\downarrow c}
    &=
    \big( \gn, \forall, \big\{ ( \beta, \mathcal{A}^\prime_{\downarrow\alpha} ) \;\big|\; (h, \mathcal{A}^\prime) \in S,~ (\alpha,\beta) \in \kappa(h,c) \big\} \big) \\
    &=
    \big( \gn, \forall, \big\{
      ( \beta, \mathcal{A}^\prime_{\downarrow\alpha})
      \;\big|\;
      \mathcal{A}^\prime = \condfalse,~
      h = \gn \rightarrow \smash\mbrloop \leftarrow \gn ,~
      (\alpha,\beta) \in \kappa(h,c)
    \big\} \big)
  \intertext{
  We can obtain possible $\alpha,\beta$ by enumerating the borrowed context diagrams where $h,c$ are already given.
  As seen in \Cref{shift-ex-bcds}, there are two possible choices for the jointly epi square in the top left:
  the $b$-loops of $h$ and $c$ can be mapped to two different loops in the center graph (\Cref{shift-ex-bcd-disjoint})
  or they can be mapped to a single loop (\Cref{shift-ex-bcd-overlap}).
  The remaining pushout and pullback squares are then uniquely determined.
  We therefore obtain:}
    &= \big( \gn, \forall, \big\{
      \big( \gn \rightarrow \smash\mbrloop \leftarrow \gn, (\condfalse)_{\downarrow\alpha_1} \big),
      \big( \gn \rightarrow \gn \leftarrow \gn, (\condfalse)_{\downarrow\alpha_2} \big)
    \\ & \hspace{100pt}\hspace{45pt} \big|\;
    \alpha_1 = \gn \rightarrow \smash\mbrloop \leftarrow \gn,\ %
    \alpha_2 = \gn \rightarrow \gn \leftarrow \gn
    \big\} \big)
  \intertext{Shifting $\condfalse$ along any arrow $\alpha_i$ again results in $\condfalse$:}
    &= \big( \gn, \forall, \big\{ ( \gn \rightarrow \smash\mbrloop \leftarrow \gn, \condfalse ), ( \gn \rightarrow \gn \leftarrow \gn, \condfalse ) \big\} \big)
  \end{align*}
  Furthermore, since $\gn \rightarrow \gn \leftarrow \gn$ is an identity cospan,
  the condition is equivalent to \mbox{$\ldots \land \condfalse \equiv \condfalse$},
  which is the expected result.
  \begin{figure}[t]
    \def\heightAdjustA{1.00}% PAGEBREAK-ADJUST
    \def\heightAdjustB{1.00}% PAGEBREAK-ADJUST
    % commands for the two figures:
      \def\drawthebox#1#2{
        \node (#2) at ($(\boxSxp,\boxSyp)!.5!(-\boxSxm,-\boxSym)$) [draw,roundbox,minimum width=\boxW cm, minimum height=\boxH cm] {};
        %\node at (-0.35,-0.25) {$\scriptstyle{#1}$}; % -0.2 when using 0.35 below
      }
      \def\arrowfromto#1#2#3{
        % ignore mono for readability (no cospanint#3)
        \draw[cospanint] (#1) edge (#2);
      }
      \def\cospanfromto#1#2#3{% left middle right | name namelabelpos anchor namelabelopts
        \arrowfromto{#1}{#2}{mono}
        \arrowfromto{#3}{#2}{}
      }
      \def\cospanfromtoarr#1#2#3#4#5#6#7{% left middle right | name namelabelpos anchor namelabelopts
        \cospanfromto{#1}{#2}{#3}
        \draw[cospanarr] (#1.#6) to[#7] node[#5]{$#4$} (#3.#6);
      }
      \def\leftmloopt{\draw[gedge] (l.150) to[loop,out=135+30,in=135-30,distance=0.3cm] node[left]{$m\mkern-5mu$} (l.120);} % chktex 8
      \def\rightmloopt{\draw[gedge] (r.60) to[loop,out=45+30,in=45-30,distance=0.3cm] node[right]{$\mkern-2mu m$} (r.30);} % chktex 8
      \def\boxat#1#2#3#4#5{% gridx gridy label boxname content
        \begin{scope}[shift={(#1*\gridx,#2*\gridy)}]
          \drawthebox{#3}{#4}
          #5
        \end{scope}
      }
      \def\gridparams{
        \pgfmathsetmacro\gridx{2.0}
        \pgfmathsetmacro\gridy{-1.3*\heightAdjustA}
        \pgfmathsetmacro\boxSxp{0.7}
        \pgfmathsetmacro\boxSxm{0.7}
        \pgfmathsetmacro\boxSyp{0.42*\heightAdjustB}
        \pgfmathsetmacro\boxSym{0.42*\heightAdjustB}
        \pgfmathsetmacro\boxW{\boxSxp+\boxSxm}
        \pgfmathsetmacro\boxH{\boxSyp+\boxSym}
      }
      \def\polabels#1#2#3#4{
        \def\sqlfs{\scriptsize\color{gray}}
        \node at ($(boxnull)!0.5!(Gp)$) {\sqlfs\textup{#1}};
        \node at ($(L)!0.5!(C)$) {\sqlfs\textup{#2}};
        \node at ($(G)!0.5!(F)$) {\sqlfs\textup{#3}};
        \node at ($(Gp)!0.5!(K)$) {\sqlfs\textup{#4}};
      }
    \begin{subfigure}[b]{0.48\textwidth}
      \centering
      \begin{tikzpicture}[gedge/.append style={font=\scriptsize}, cospanarr/.append style={font=\scriptsize}]
        \gridparams
        \pgfmathsetmacro\boxW{\boxSxp+\boxSxm}
        \boxat00{\obnull}{boxnull}{\gn}
        \boxat10{L}{L}{\gn\mlRight}
        \boxat20{I}{I}{\gn}
        \boxat21{C}{C}{\gn\mlLeft}
        \boxat01{G}{G}{\gn\mlLeft}
        \boxat11{\;G^+}{Gp}{\gn\mlLeft\mlRight}
        \boxat02{J}{J}{\gn}
        \boxat12{F}{F}{\gn\mlRight}
        \boxat22{K}{K}{\gn}
        \cospanfromto{boxnull}{L}{I} %{\ell}{above}{north}{bend left=10}
        \cospanfromto{boxnull}{G}{J} %{a}{left}{west}{bend right=10}
        \draw[cospanarr] (boxnull.north east) to[bend left=10] node[above]{$h$} (I.north west);
        \draw[cospanarr] (boxnull.south west) to[bend right=12] node[left]{$c$} (J.north west);
        \draw[cospanarr] (J.south east) to[bend right=10] node[below]{$\beta_1$} (K.south west);
        \draw[cospanarr] (I.south east) to[bend left=12] node[right]{$\alpha_1$} (K.north east);
        \cospanfromto{J}{F}{K}
        \cospanfromto{I}{C}{K}
        \cospanfromto{L}{Gp}{F}
        \cospanfromto{G}{Gp}{C}
        \polabels{JE}{PO}{PO}{PB}
      \end{tikzpicture}
      \caption{$b$-loops of $h$ and $c$ are kept separate in the center graph.}%
      \label{shift-ex-bcd-disjoint}
    \end{subfigure}\hfill%\hspace{1cm}
    \begin{subfigure}[b]{0.48\textwidth}
      \centering
      \begin{tikzpicture}[gedge/.append style={font=\scriptsize}, cospanarr/.append style={font=\scriptsize}]
        \gridparams
        \boxat00{\obnull}{boxnull}{\gn}
        \boxat10{L}{L}{\gn\mlRight}
        \boxat20{I}{I}{\gn}
        \boxat21{C}{C}{\gn}
        \boxat01{G}{G}{\gn\mlLeft}
        \boxat11{\;G^+}{Gp}{\gn\mlRight}
        \boxat02{J}{J}{\gn}
        \boxat12{F}{F}{\gn}
        \boxat22{K}{K}{\gn}
        \cospanfromto{boxnull}{L}{I} %{\ell}{above}{north}{bend left=10}
        \cospanfromto{boxnull}{G}{J} %{a}{left}{west}{bend right=10}
        \draw[cospanarr] (boxnull.north east) to[bend left=10] node[above]{$h$} (I.north west);
        \draw[cospanarr] (boxnull.south west) to[bend right=12] node[left]{$c$} (J.north west);
        \draw[cospanarr] (J.south east) to[bend right=10] node[below]{$\beta_2$} (K.south west);
        \draw[cospanarr] (I.south east) to[bend left=12] node[right]{$\alpha_2$} (K.north east);
        \cospanfromto{J}{F}{K}
        \cospanfromto{I}{C}{K}
        \cospanfromto{L}{Gp}{F}
        \cospanfromto{G}{Gp}{C}
        \polabels{JE}{PO}{PO}{PB}
      \end{tikzpicture}
      \caption{$b$-loops of $h$ and $c$ are mapped to a single loop.}%
      \label{shift-ex-bcd-overlap}
    \end{subfigure}
    \caption{Borrowed context diagrams for $h,c$ as given in \Cref{ex-shift-simplify}.}%
    \label{shift-ex-bcds}
  \end{figure}
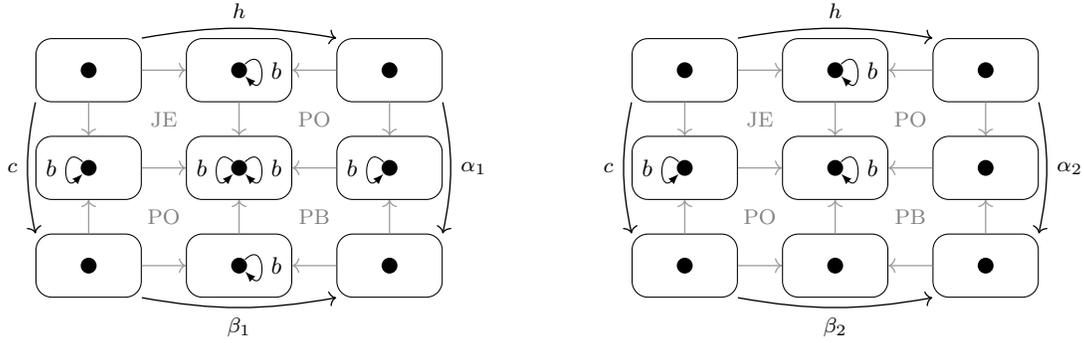
\end{exa}

\subsection{Conditional Reactive Systems}

We now extend reactive systems with application conditions:

\begin{defi}[Conditional reactive system~\cite{CRS}]
A \emph{rule with condition} is a triple $(\ell,r,\mathcal R)$
where $\ell, r \colon \obnull \to I$ are arrows % with common codomain
and $\mathcal R$ is a condition with root object $I$.
A \emph{conditional reactive system} is a set of rules with conditions.
\end{defi}

As the root object $I$ of the condition is the codomain of the rule arrow,
it is also the domain of the reactive context,
which has to satisfy the rule condition in order to be able to apply the rule:

\begin{defi}[Reaction]
Let $a, a'$ be arrows of a conditional reactive system with rules $\mathcal S$.\linebreak
We say that \emph{$a$ reduces to $a'$} ($a \leadsto a'$)
whenever there exists a rule $(\ell,r,\mathcal R) \in \mathcal S$ with $\ell,r \colon \obnull \to I$
and a reactive context $c \colon I \to J$
such that $a = \ell;c$, $a' = r;c$ and additionally $c \models \mathcal R$.
\end{defi}

In order to define a bisimulation for conditional reactive systems that is also a congruence,
it is necessary to enrich labels with conditions derived from
the application conditions. Since we can not assume that the full
context is present, the application condition might refer to currently
unknown parts of the context and this has to be suitably integrated
into the label.

\begin{defi}[Context/representative step with conditions~\cite{DBC-CRS}]\label{def-condcstep}%
\floatingpicspacerightafter[2]{3.5cm}
\begin{floatingpic}[-3\baselineskip-4pt]{-3.2cm}%
\centering\begin{tikzpicture}

\def\sqw{1.35}
\def\sqh{1.3}

\path[use as bounding box] (-0.33,1) rectangle (2*\sqw+0.1,-\sqh-1);

\node (tlempty) at (0,0) {$\obnull$};
\node (i) at (1*\sqw,0) {$I$};
\node (trempty) at (2*\sqw,0) {$\obnull$};
\node (j) at (0,-\sqh) {$J$};
\node (k) at (1*\sqw,-\sqh) {$K$};
\draw[->] (tlempty) -- node[above]{$\ell$} (i);
\draw[->] (trempty) -- node[above]{$r$} (i);
\draw[->] (trempty) -- node[below,pos=0.4]{$a^\prime$} (k);

\draw[->] (tlempty) -- node[left]{$a$} (j);
\draw[->] (i) -- node[left]{$c$} (k);

\draw[->] (j) -- node[above]{$f$} (k);
\node[condtri,dart tip angle=50,shape border rotate=270] at (i.north) {$\mathcal R$};
\node[condtri,dart tip angle=50,shape border rotate=90] at (k.south) {$\mkern-3mu\mathcal A$};

\node (pcdst) at (2*\sqw,-1.3*\sqh) {};
\draw[->,fancydotted] (k) -- node[below]{$d$} (pcdst);

\end{tikzpicture}
\end{floatingpic}%
Let $\mathcal S$ be a conditional reactive system, let
$a \colon \obnull \to J,\ f \colon J \to K,\ a' \colon \obnull \to K$
be arrows and $\mathcal A$ be a condition over~$K$.  We
write $a \cstep{f}{\mathcal A} a'$ whenever there exists a rule~%
\mbox{$(\ell,r,\mathcal R) \in \mathcal S$} and an arrow~$c$ such that
$a;f = \ell;c,\ a' = r;c$ (i.e.\ the reaction is \mbox{possible without}
conditions) and furthermore
$\mathcal A \models \mathcal{R}_{\downarrow c}$ (a condition on an additional context $d$ as explained below).
Such steps are called \emph{context steps}.

\floatingpicspaceright{3.5cm}
We write $a \rstep{f}{\mathcal A} a'$ whenever
$a \cstep{f}{\mathcal A} a'$, $\kappa(a,\ell) \ni (f,c)$ and
$\mathcal A = \mathcal{R}_{\downarrow c}$.  Such steps are called
\emph{representative~steps}.
\end{defi}

\floatingpicspaceright[1]{3.5cm}
Conditions are represented graphically in the form of ``arrowhead
shapes'' depicted next to the root object.  Intuitively
$a \cstep{f}{\mathcal A} a'$ means that $a$ can make a step to $a'$
when borrowing $f$, if the yet unknown context $d$ beyond $f$ satisfies
condition $\mathcal{A}$ (since this context $d$ does not directly
participate in the reduction, we call it \emph{passive context}).

The intuition behind this requirement is that $\mathcal A$ should
allow only the contexts that are allowed by the rule condition
$\mathcal R$ (thereby checking that the rule can actually be
applied).
Since $\mathcal A$ is a condition over an additional context $d$ that is
beyond the reaction context $c$, and $c$ might partially satisfy
$\mathcal R$, we shift $\mathcal R$ over $c$ to obtain a condition that
only requires the parts of $\mathcal R$ that are still missing.
For context steps, $\mathcal A$ may also be stronger, hence $\models$.

In the case of a representative step, we
require that a context step is possible, the borrowed context is
minimal, and the condition on the passive context is not stronger than
necessary.

In the proofs, we will make extensive use of the following construction
to obtain a representative step for a given context step:

\begin{rem}\label{cond-repr-zurueckfuehren}%
\Cref{def-rsq,def-condcstep} imply,
analogously to \Cref{repr-zurueckfuehren},
that every context~step $a \cstep{f}{\mathcal A} a'$ (left diagram)
can be reduced to a representative step \mbox{$a \rstep{\hat f}{\mathcal{R}_{\downarrow \hat c}} r;\hat c$} (right diagram),
with $\mathcal R$ being the condition of the rule $(\ell,r,\mathcal R)$ that enables the given context step.

  
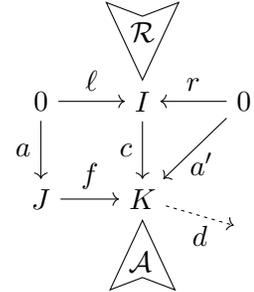
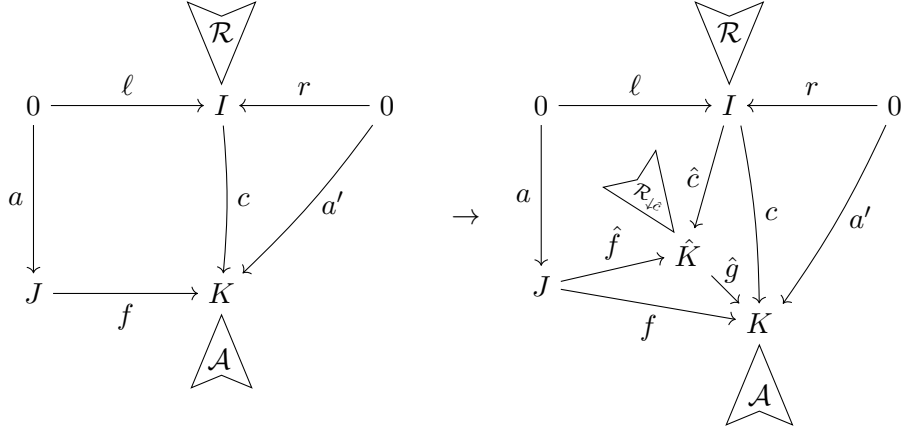
\begin{center}\begin{tikzpicture}[x=1.00cm,y=1.00cm]% PAGEBREAK-ADJUST
    % %%% LEFT - PRE R-STEP %%%
    \begin{scope}[shift={(0,0)}]
      \node (tlempty) at (0,0) {$\obnull$};
      \node (i) at (2.5,0) {$I$};
      \node (trempty) at (4.7,0) {$\obnull$};
      \node (j) at (0,-2.5) {$J$};
      \node (k) at (2.5, -2.5) {$K$};
      \draw[->] (tlempty) -- node[above]{$\ell$} (i);
      \draw[->] (trempty) -- node[above]{$r$} (i);
      \draw[->] (trempty) to[bend left=5] node[right]{$a'$} (k);

      \draw[->] (tlempty) -- node[left]{$a$} (j);
      \draw[->] (i) to[bend left=5] node[right]{$c$} (k);

      \draw[->] (j) -- node[below]{$f$} (k);

      \node[condtri,shape border rotate=270] at (i.north) {$\mathcal{R}$};
      \node[condtri,shape border rotate=90] at (k.south) {$\mkern-3mu\mathcal{A}$};
    \end{scope}

    \node at (5.75,-1.5) {$\rightarrow$};

    % %%% RIGHT %%%
    \begin{scope}[shift={(6.75,0)}]
      \node (tlempty) at (0,0) {$\obnull$};
      \node (i) at (2.5,0) {$I$};
      \node (trempty) at (4.7,0) {$\obnull$};
      \node (j) at (0,-2.4) {$J$};
      \node (kp) at (1.95,-1.95) {$\hat K$};
      \node (k) at (2.9, -2.9) {$K$};
      \draw[->] (tlempty) -- node[above]{$\ell$} (i);
      \draw[->] (trempty) -- node[above]{$r$} (i);
      \draw[->] (trempty) to[bend left=5] node[right]{$a'$} (k.40);

      \draw[->] (tlempty) -- node[left]{$a$} (j);
      \draw[->] (i.300) to[bend left=5] node[right]{$c$} (k);

      \draw[->] (j) -- node[below]{$f$} (k);

      \draw[->] (i) -- node[left]{$\hat c$} (kp);
      \draw[->] (j) -- node[above]{$\hat f$} (kp);
      \draw[->] (kp) -- node[right,pos=0.1]{\raisebox{9pt}{$\hat g$}} (k);

      \node[condtri,shape border rotate=270] at (i.north) {$\mathcal R$};
      \node[condtri,shape border rotate=90] at (k.south) {$\mathcal A$};
      % shift is applied before rotate that's why the coordinates are so odd - too lazy to fix
      \node[rotate around={-55:(kp.center)},condtri,shape border rotate=0,scale=0.7] at (kp.west) {%
        \raisebox{-5pt}[5pt][2pt]{%
          \rotatebox{25}{$\mathcal{R}_{\downarrow \hat c}$}%
        }\kern-2pt
      };
        % \node[condtri,dart tip angle=60,shape border rotate=90] at (ii.south) {$\kern -0.5mm\mathcal{R}_i\kern -0.5mm$};
        % \node[rotate around={20:(k.center)},...]
    \end{scope}
  \end{tikzpicture}
  \end{center}
\end{rem}

We will also make use of the following context-step-rewriting lemma:

\begin{lem}\label{lemma-ctxtrans-composition}%
  A context step $a;d \cstep{f}{\mathcal{A}} a^\prime$ is equivalent
  to a step $a \cstep{d;f}{\mathcal{A}} a^\prime$.  In particular,
  $a;d \cstep{f}{\mathcal{A}} a^\prime$  if and only if
  $a \cstep{d;f}{\mathcal{A}} a^\prime$.
\end{lem}

\begin{proof}
  According to \Cref{def-condcstep}, for a step
  $a;d \cstep{f}{\mathcal{A}} a^\prime$ there exists a rule
  $(\ell,r,\mathcal{R})$ such that
  $(a;d);f = \ell;c,\ a^\prime = r;c,\ \mathcal{A} \models
  \mathcal{R}_{\downarrow c}$ for some arrow $c$.  Since composition
  is associative, we rewrite this to $a;(d;f) = \ell;c$, which
  immediately results in the definition of
  $a \cstep{d;f}{\mathcal A} a^\prime$.
\end{proof}

We now extend (semi-)saturated bisimilarity to rules with conditions: % chktex 36

% A notion of bisimulation for conditional reactive systems needs to
% take conditions into consideration:

% , otherwise the definition is similar in spirit to saturated
% bisimilarity for standard reactive systems:

\begin{defi}[(Semi-)Saturated bisimilarity~\cite{DBC-CRS}]\label{def-bisim-c}% % chktex 36
  Let $\mathcal{S}$ be a conditional reactive system.
  A \emph{saturated bisimulation} is a relation $R$,
  relating arrows $a,b \colon \obnull \to J$, such that: for
  all $(a,b) \in R$ and for every context step
  $a \cstep{f}{\mathcal A} a^\prime$ there exist answering moves
  $b \cstep{f}{\mathcal{B}_i} b_i^\prime$, $i\in I$, such that
  $(a^\prime, b_i^\prime) \in R$ and
  $\mathcal{A} \models \bigvee_{i \in I} {\mathcal{B}_i}$, where $I$
  is a finite index set;
  and, vice versa, for every context step
  $b \cstep{f}{\mathcal B} b'$ there exist answering moves
  $a \cstep{f}{\mathcal A_j} a_j^\prime$, $j \in J$, such that
  $(a_j^\prime, b') \in R$ and
  $\mathcal B \models \bigvee_{j \in J} {\mathcal A_j}$.

  Two arrows $a, b$ are called \emph{saturated bisimilar}
  ($(a, b) \in \simC$) whenever there exists a saturated bisimulation
  $R$ with $(a, b) \in R$.  Similarly, for \emph{semi-saturated
    bisimilarity} we require that $\rightarrow_R$-steps of $a$ can be
  answered by $\rightarrow_C$-steps of $b$,
  and vice versa for $\rightarrow_R$-steps of $b$.
  Saturated and semi-saturated bisimilarity agree and both are
  congruences~\cite{DBC-CRS}.
\end{defi}

\section{Conditional Bisimilarity}%
\label{sec:condbisim}

We will now introduce our new results on conditional bisimilarity: as
stated earlier, our motivation is to extend the notion of saturated
bisimilarity, which is often too strict, since it requires that two system
states behave identically in all possible contexts. However, sometimes
it is enough to ensure behavioural equivalence only in specific
environments.

Hence we now replace standard bisimilarity, which is a binary
relation, by a ternary relation --- called conditional relation --- with
tuples of the form $(a,b,\mathcal{C})$.
Then, a conditional bisimulation is a conditional relation, where a tuple $(a,b,\mathcal C)$ can
be read as: $a,b$
are bisimilar in all contexts satisfying $\mathcal{C}$.

% We present new original results from here onwards.

% In verification, a common task is to check whether two components are
% interchangeable or whether a component conforms to a specification.
% Saturated bisimilarity allows checking for identical behaviour in all
% contexts.  In some situations, however, this is too strict: it might
% be known that under some specific conditions the component deviates
% from the specification, but bisimulation shall be proved for all other
% cases.

% In this section, we introduce conditional bisimulations that allow us
% to restrict the set of contexts to be checked for bisimilarity.

\subsection{Definition, Properties and Examples}

\begin{defi}[Conditional relation]\label{def-condrel}%
  A \emph{conditional relation} is a set of triples
  $(a,b,\mathcal C)$, where $a,b \colon \obnull \to J$ are arrows with
  identical target and $\mathcal C$ is a condition over~$J$.
\end{defi}

Note that for a triple $(a,b,\mathcal C)$,
the root object of the condition $\mathcal C$ is not the source of $a$
(as is the case for satisfaction), but the target $\codom(a)$.  This
is because we do not state a condition on the arrows $a,b$ themselves,
but on the context in which they are embedded ($a;f$ resp.\ $b;f$ for
some context $f$), so the condition is over $\dom(f) = \codom(a) = \codom(b)$.

\begin{defi}[Closure under contextualization, $u(R)$, conditional congruence]%
\label{def-condrel-u}%

  If $R$ is a conditional relation, then:
  \begin{itemize}
  \item
    $R$ is \emph{reflexive} if
    $(a,a,\mathcal C) \in R$ for all $a, \mathcal C$ with $\codom(a)=\Ro(\mathcal C)$
  \item
    $R$ is \emph{symmetric} if
    $(a,b,\mathcal C)\in R$ implies $(b,a,\mathcal C)\in R$
  \item
    $R$ is \emph{transitive} if
    $(a,b,\mathcal C) \in R$ and $(b,c,\mathcal C) \in R$ implies
    $(a,c,\mathcal C) \in R$
  \item
    $R$ is \emph{closed under contextualization} if
    $(a,b,\mathcal C) \in R$ implies $(a;d, b;d, \mathcal C_{\downarrow d}) \in R$
  \item
    $R$ is a \emph{conditional congruence} if
    it is an equivalence (reflexive, symmetric, transitive) and closed under contextualization
  \end{itemize}

  \noindent
  For a conditional relation $R$, $u(R)$ is its closure under contextualization, that is,\linebreak
  \mbox{$u(R) \defeq \{ (a;d, b;d, \mathcal{C}_{\downarrow d}) \mid (a,b,\mathcal{C}) \in R,
  \allowbreak\ a,b \colon \obnull \to J ,\  d \colon J \to K \}$.}
\end{defi}

Closure under contextualization means that whenever $a,b$ are
related under a context satisfying $\mathcal{C}$, then they are
still related when we contextualize under $d$, where however the
condition has to be shifted since we commit to the fact that the
context is of the form $d;c$ for some additional context $c$.

We will now introduce one of the central definitions of this paper. Here,
given a conditional reactive system we will describe when two arrows
$a,b$ are bisimilar in all contexts that satisfy a condition
$\mathcal{C}$.

\begin{defi}[Conditional bisimulation]\label{def-cb-c}%
  We fix a conditional reactive system.
A conditional bisimulation $R$ is a conditional relation such that the following holds:
for each triple $(a,b,\mathcal{C}) \in R$ and each context step $a \cstep{f}{\mathcal A} a^\prime$,
there are answering steps $b \cstep{f}{\mathcal{B}_i} b_i^\prime$,
$i\in I$
(where $I$ is possibly infinite),
and conditions $\mathcal{C}_i^\prime$
such that $(a^\prime, b_i^\prime, \mathcal{C}_i^\prime) \in R$
and $\mathcal{A} \land \mathcal{C}_{\downarrow f} \models \bigvee_{i \in I} \left( \mathcal{C}_i^\prime \land \mathcal{B}_i \right)$;
and vice versa\footnote{For each triple $(a,b,\mathcal{C}) \in R$ and each context step $b \cstep{f}{\mathcal B} b'$,
there are answering steps $a \cstep{f}{\mathcal{A}_j} a_j^\prime$
and conditions $\mathcal{C}_j^\prime$
such that $(a_j', b', \mathcal{C}_j^\prime) \in R$
and $\mathcal{B} \land \mathcal{C}_{\downarrow f} \models \bigvee_{j \in J} \left( \mathcal{C}_j^\prime \land \mathcal{A}_j \right)$}.
Two arrows are \emph{conditionally bisimilar \mbox{under $\mathcal C$}} ($(a, b, \mathcal{C})
\in \csimC$) whenever a conditional bisimulation $R$ with $(a, b,
\mathcal{C}) \in R$ exists.\footnote{Note that since conditional
  bisimulations are closed under union, $\csimC$ is itself a
  conditional bisimulation.}
\end{defi}
\begin{figure}[ht]%
\centering\begin{tikzpicture}

  \def\sq{1.75}

  \node (tlempty) at (0,0) {$\obnull$};
  \node (i) at (1*\sq,0) {$I$};
  \node (trempty) at (2*\sq,0) {$\obnull$};
  \node (j) at (0,-1*\sq) {$J$};
  \node (k) at (1*\sq,-1*\sq) {$K$};
  \node (blempty) at (0,-2*\sq) {$\obnull$};
  \node (ii) at (1*\sq,-2*\sq) {$I_i$};
  \node (brempty) at (2*\sq,-2*\sq) {$\obnull$};
  \draw[->] (tlempty) -- node[above]{$\ell$} (i);
  \draw[->] (trempty) -- node[above]{$r$} (i);
  \draw[->] (trempty) -- node[below,pos=0.25]{$a^\prime$} (k.30);
  \draw[->] (blempty) -- node[above]{$\ell_i$} (ii);
  \draw[->] (brempty) -- node[above]{$r_i$} (ii);
  \draw[->] (brempty) -- node[above,pos=0.25]{$b_i^\prime$} (k.-30);

  \draw[->] (tlempty) -- node[left]{$a$} (j);
  \draw[->] (blempty) -- node[left]{$b$} (j);
  \draw[->] (i) -- node[left]{$c$} (k);
  \draw[->] (ii) -- node[left]{$e_i$} (k);

  \draw[->] (j) -- node[above]{$f$} (k);

  \node[condtri,dart tip angle=60,shape border rotate=270] at (i.north) {$\mathcal{R}$};
  \node[condtri,dart tip angle=60,shape border rotate=0] at (j.west) {$\kern 0.5mm \mathcal{C}$};
  \node[condtri,dart tip angle=60,shape border rotate=90] at (ii.south) {$\kern -0.5mm\mathcal{R}_i\kern -0.5mm$};
  \node[rotate around={5:(k.center)},condtri,shape border rotate=180] at (k.east) {\kern-3pt\raisebox{0pt}[0.5\baselineskip][0.2\baselineskip]{$\mathcal{C}_i^\prime$}};
  \node[rotate around={-20:(k.center)},condtri,shape border rotate=270,scale=0.9] at (k.80) {$\mathcal{A}$};
  \node[rotate around={20:(k.center)},condtri,shape border rotate=90,scale=0.9] at (k.-80) {\raisebox{0pt}[1.8pt][1.5pt]{$\mathcal{B}_i$}};

  \node (pcdst) at (2.25*\sq,-1.4*\sq) {};
  \draw[->,fancydotted] (k.-10) -- node[above,pos=0.85]{$d$} (pcdst);
\end{tikzpicture}%
\caption{A context step $a \protect\cstep{f}{\mathcal A} a'$ (top half) and a single answer step $b \protect\cstep{f}{\mathcal B_i} b_i^\prime$ (bottom half) in conditional bisimulation}%
\label{fig-cbdef}
\end{figure}

The situation for one answer step is depicted in \Cref{fig-cbdef}.
Since the definition is rather complex, we will discuss its various
aspects in the following remarks.

\begin{rem}[Logical implication]
In \Cref{def-cb-c}, the implication $\mathcal A \land \mathcal C_{\downarrow f} \models \bigvee_{i \in I} (\mathcal C_i^\prime \land \mathcal B_i)$
is to be understood as follows:
For every step, we have a borrowed context~$f$ and an additional passive context $d$ (as explained below \Cref{def-condcstep}).
The condition~$\mathcal C$ from the triple refers to the full context of $a$ (i.e.\ both the borrowed context $f$ and the passive context $d$, hence $f;d \models \mathcal C$ or equivalently $d \models \mathcal C_{\downarrow f}$),
while~$\mathcal A$, coming from the context step, only refers to the passive context $d$ (hence $d \models \mathcal A$).

Every environment $d$ that is valid for the context step of $a$
(i.e.\ which satisfies $\mathcal A \land \mathcal C_{\downarrow f}$)
must also be valid for some answering step of $b$, i.e.\ satisfies at least one $\mathcal B_i$.
Depending on the context, different answering steps may be chosen, and the resulting pair $a',b'_i$ might only be conditionally bisimilar for some contexts, which is indicated by the condition $\mathcal C_i^\prime$.
\end{rem}

\begin{rem}[Necessity of multiple answering steps]
As for saturated bisimilarity,
we need to allow several answering moves for a single
step of $a$: the answering step taken by $b$ might depend on the
context, using different rules for contexts satisfying different
conditions $\mathcal{B}_i$. We just have to ensure that all answering
step conditions taken together (disjunction on the right-hand side) fully cover the conditions under which the
step of $a$ is feasible (left-hand side).
As an example for this, consider the following example (originally presented in~\cite[remark after Definition 15]{DBC-CRS}).
Assume three rules:
  \[\bgroup%
  \renewcommand{\arraystretch}{1.5}%
  \setlength{\arraycolsep}{0.3pt}%
  \begin{array}{rlrccclr} % x.|R|=... | L | ... | R | ... | text...
    1.\ &  R_A   &{}= \big(
      \emptyset \rightarrow & \inlinechan{chan=a}  & \leftarrow \inlinechan{},\  %
      \emptyset \rightarrow & \inlinechan{} & \leftarrow \inlinechan{},
      \condtrue_{\mathifnodessubscript} \big)
      &\text{(unconditionally delete an $a$-edge)} \\
    2.\ & R_{B1} &{}= \big(
      \emptyset \rightarrow & \inlinechan{chan=b}  & \leftarrow \inlinechan{},\  %
      \emptyset \rightarrow & \inlinechan{} & \leftarrow \inlinechan{},\ %
      \mathcal A_q \big)
      &\text{\parbox[t]{61.5mm}{(delete a $b$-edge if the context satisfies some condition $\mathcal A_q$)}} \\
    3.\ & R_{B2} &{}= \big(
      \emptyset \rightarrow & \inlinechan{chan=b}  & \leftarrow \inlinechan{},\  %
      \emptyset \rightarrow & \inlinechan{} & \leftarrow \inlinechan{},\ %
      \neg\mathcal A_q \big)
      &\text{\parbox[t]{61.5mm}{(delete a $b$-edge in contexts not satisfying $\mathcal A_q$)}}
  \end{array}\egroup\]
(Hence rules 2 and 3 together allow a $b$-edge to be deleted in any context, since every context satisfies either $\mathcal A_q$ or $\neg\mathcal A_q$.
The condition $\mathcal A_q$ can be chosen arbitrarily, as long as it
is not equivalent to $\condtrue$ or $\condfalse$.)
Then, an $a$-edge is conditionally bisimilar to a $b$-edge under $\condtrue$ (all contexts):
a step that deletes the $a$-edge can be answered by deleting the $b$-edge,
but depending on the context that the step happens in, a different rule has to be chosen:
two answering steps with $\mathcal B_1 = \mathcal A_q,\ \mathcal B_2 = \neg\mathcal A_q$ are required, and together they cover $\condtrue$.
(The other direction --- deleting the $b$-edge using either rule 2 or 3 being answered by deleting the $a$-edge using rule 1 --- does not require multiple answering steps in this example.)
\end{rem}

\begin{rem}[Infinitely many answering steps]
  The definition explicitly permits an infinite index set $I$ for the
  answering steps (this is in contrast to saturated bisimilarity, cf.\
  \Cref{def-bisim-c}, which required finite $I$). If we do not
  consider conditional bisimilarity and the finiteness assumption
  (\textsc{Fin}) holds, it does not make a difference whether we
  consider finite or infinite index sets, since there are only
  finitely many possible answering steps~\cite{DBC-CRS}. However, in
  the presence of conditions, it might make a difference.

  Since the logic does not support infinite disjunctions,
  $\mathcal A \land \mathcal C_{\downarrow f} \models \bigvee_{i \in
    I} {\mathcal C_i' \land \mathcal B_i}$ means that for every $d$
  with $d \models \mathcal A \land \mathcal C_{\downarrow f}$, there
  exists $i \in I$ such that
  $d \models \mathcal C_i' \land \mathcal B_i$.

  In many practical applications, it may be useful to restrict to a
  variant of the definition that permits only finitely many answering
  steps.  Our theorems are valid for either variant (finite or
  infinite), except for the proof of \Cref{cec-cb} which in its
  current version requires $I$ to be infinite.
%
% In the rest of the paper, similar infinite disjunctions should be interpreted as described above whenever the variant of conditional bisimulation with infinite $I$ is being used.
\end{rem}

\def\AIneN{\mathcal A_{\overline{n}}}%
\def\AUneN{\mathcal A_{U \nexists n}}%
\begin{exa}[Message passing over unreliable channels]\label{ex-cb-unreliable}
We now work in the category of input-linear cospans of graphs, i.e.,
$\ILC(\catname\Graphfin)$.

We extend our previous example (cf.\ \Cref{ex-cb-unreliable-pre}) of networked
nodes, introducing different types of channels.
A channel can be reliable or unreliable, indicated by an $\relLabel$-edge or $\unrelLabel$-edge respectively.
Sending a message over a reliable channel always succeeds (rule~$P_R$), while an unreliable channel only transmits a message if there is no noise (indicated by a parallel $\noiseLabel$-edge) in the environment that disturbs the transmission (rule~$P_U$).
The~reactive system has the following rules with application conditions, where condition $\AUneN$ states that the unreliable channel must not have an $\noiseLabel$-edge in parallel:

\begin{tikzpicture}[gedge/.append style={font=\footnotesize}]
  \def\drawthebox#1#2{
    \node[anchor=center] (#2) at (0,0) [draw,roundbox,minimum width=2cm, minimum height=1.25cm] {};
    \node[anchor=west] at (-1cm,-0.35) {$\scriptstyle{#1}$};
  }
  \def\smolnodes{
    \node[gn] (l) at (-0.35,0.05) {};
    \node[gn] (r) at (0.35,0.05) {};
  }
  \def\smolchan#1{
    \smolnodes
    \draw[gedge] (l) to node[below]{$#1$} (r);
  }
  \def\smolchanSh{\begin{scope}[shift={(-0.25,0)}]\smolchan\end{scope}}
  \def\leftmloopt{\draw[gedge] (l.120) to[loop,out=120+30,in=120-30,distance=0.3cm] node[left]{$m\mkern-5mu$} (l.90);}
  \def\rightmloopt{\draw[gedge] (r.90) to[loop,out=60+30,in=60-30,distance=0.3cm] node[right]{$\mkern-2mu m$} (r.60);}
  \def\PRPUrule#1#2#3#4{% yshift chanlabel capitallabel cond
  \begin{scope}[shift={(0,#1)}]
    \begin{scope}[shift={(-2,-0.125)}]
      \node (leftnull) at (0,0) [draw,roundbox,minimum width=0.5cm, minimum height=1.25cm] {};
    \end{scope}

    \begin{scope}[shift={(0,-0.125)}]
      \smolchan{#2}\leftmloopt
      \drawthebox{}{leftbox}
    \end{scope}

    \begin{scope}[shift={(3,-0.125)}]
      \smolchan{#2}
      \drawthebox{{#3}_0}{ifbox}
    \end{scope}

    \begin{scope}[shift={(6,-0.125)}]
      \smolchan{#2}\rightmloopt
      \drawthebox{}{rightbox}
    \end{scope}

    \begin{scope}[shift={(8,-0.125)}]
      \node (rightnull) at (0,0) [draw,roundbox,minimum width=0.5cm, minimum height=1.25cm] {};
    \end{scope}

    \node[inner xsep=0cm,anchor=east] at (-2.4,-0.125) {$P_{#3} = \Bigg($};
    \node[inner xsep=0cm,anchor=west] at (8.4,-0.125) {$, #4 \Bigg)$};
    \draw[cospanint] (leftnull) -- (leftbox);
    \draw[cospanint] (ifbox) -- (leftbox);
    \draw[cospanint] (ifbox) -- (rightbox);
    \draw[cospanint] (rightnull) -- (rightbox);
  \end{scope}
  }
  \PRPUrule0{\relLabel}R{\condtrue_{R_0}}
  \PRPUrule{-1.5}{\unrelLabel}U{\AUneN}

  % "explanation" for first one
  \draw[cospanarr] (-2,0.7) to[bend left=20] node[below,pos=0.65]{$\ell$} (3-0.25,0.85);
  \draw[cospanarr] (8,0.7) to[bend right=20] node[below,pos=0.65]{$r$} node[above]{} (3+0.25,0.85);
  % FINAL: ^ check layout. If the previous line is NOT mostly blank, remove the "overlay" key which causes the labels to be included in the bounding box again.
  % (We excluded them here because otherwise the space looks very large)
  % If it IS mostly blank, add "overlay" key to the node[] that labels the edge
  \node at (0,0.8) {$L$};
  \node at (3,0.8) {$I$};
  \node at (6,0.8) {$R$};

  % condition
  \begin{scope}[shift={(0,-3.2)}]
    \node[inner xsep=0cm,anchor=east] at (-2.4,-0.125) {$\AUneN = \Big($};
    \node[inner xsep=0cm,anchor=west] at (-2.4,-0.125) {$U_0, \forall, \Big\{ \Big( $};
    \def\smolnodes{
      \node[gn] (l) at (-0.35,0.05) {};
      \node[gn] (r) at (0.35,0.05) {};
    }
    \begin{scope}[shift={(0.3,-0.125)}]
      \smolchan{\unrelLabel}
      \drawthebox{U_0}{leftbox}
    \end{scope}

    \begin{scope}[shift={(3.3,-0.125)}]
      \smolchan{\unrelLabel}
      \draw[gedge] (l.45) to[bend left=10] node[above]{$\noiseLabel$} (r.135);
      \drawthebox{U_N}{ifbox}
    \end{scope}

    \begin{scope}[shift={(6.3,-0.125)}]
      \smolchan{\unrelLabel}
      \drawthebox{U_0}{rightbox}
    \end{scope}
    \draw[cospanint] (leftbox)  -- (ifbox);
    \draw[cospanint] (rightbox) -- (ifbox);
    \node[inner xsep=0cm,anchor=west] at (7.5,-0.125) {$, \condfalse_{U_0} \Big) \Big\} \Big) $};
  \end{scope}
\end{tikzpicture}

Hence the application condition $\AUneN$ says that the
context must \emph{not} be decomposable into
$U_0\rightarrow U_N\leftarrow U_0$ and some other cospan, which
is only the case if the $\unrelLabel$-edge in the interface of $P_U$ has no parallel $\noiseLabel$-edge. In other words:
there is no noise.

We compare the behaviour of a reliable channel
($r \defeq \emptyset \rightarrow \inlinechan{reliable} \leftarrow \inlinechan{}$)
to that of an unreliable channel
($u \defeq \emptyset \rightarrow \inlinechan{unreliable} \leftarrow \inlinechan{}$).
It is easy to see that they are not saturated bisimilar: $r$ can do
a step by borrowing a message on the left
($f \defeq \inlinechan{} \rightarrow \inlinechan{msgleft} \leftarrow \inlinechan{}$) without further
restrictions (i.e.\ using an environment condition
$\mathcal A = \condtrue$). But $u$ is unable to answer this step,
because the corresponding rule is only applicable if no $\noiseLabel$-edge is
present.

However, $r$ and $u$ are conditionally bisimilar under the
assumption that no $\noiseLabel$-edge is present between the two nodes
($\mathcal C = \AIneN$, where $\AIneN \defeq \big( \inlinechan{}, \forall, \big\{ (
{\inlinechan{} \rightarrow \inlinechan{noise} \leftarrow \inlinechan{}},
\allowbreak\condfalse_{\mathifnodessubscript} ) \big\} \big)$),
i.e.\ there exists a conditional bisimulation that contains
$(r, u, \AIneN)$.  A direct proof is hard, since the
proof involves checking infinitely many context steps, since messages
accumulate on the right-hand side.  However, in
\Cref{ex-cb-repr-unreliable} we will use an argument based on
representative steps to construct a proof.
\end{exa}

\begin{rem}[Condition strengthening]\label{cond-strengthening}
  It holds that $(a,b,\mathcal C')\in \csimC$,
  $\mathcal C \models \mathcal C'$ implies
  $(a,b,\mathcal{C})\in \csimC$. (This is due to the fact that
  $\mathcal C \models \mathcal C'$ implies
  $\mathcal C_{\downarrow f}\models \mathcal C'_{\downarrow f}$ which,
  in \Cref{def-cb-c}, implies
  $\mathcal A \land \mathcal C_{\downarrow f}\models \mathcal A \land
  \mathcal C'_{\downarrow f}$ for any condition $\mathcal A$ and arrow
  $f$.)
% To prove that $a$ and $b$ are conditionally bisimilar under $\mathcal C$, it is sufficient to find a conditional bisimulation containing a triple $(a,b,\mathcal C')$ with a weaker condition $\mathcal C'$ (i.e. $\mathcal C \models \mathcal C'$).
% $\mathcal C \models \mathcal C'$ implies $\mathcal C_{\downarrow f}\models \mathcal C'_{\downarrow f}$
% which (in \Cref{def-cb-c}) implies $\mathcal A \land \mathcal C_{\downarrow f}\models \mathcal A \land \mathcal C'_{\downarrow f}$
% for any condition $\mathcal A$ and arrow $f$.
\end{rem}

%\subsection{Soundness}
Our motivation for introducing the notion of conditional bisimilarity
was to check whether two systems are behaviourally equivalent when
they are put into a context that satisfies some condition
$\mathcal C$.  It is not immediately obvious that our definition can
be used for this purpose, since all context steps are checked, not
just the ones that actually satisfy $\mathcal C$.

Hence we now show that our definition is sound, i.e.\ if two systems
are conditionally bisimilar, then they show identical behaviour under
all contexts that satisfy $\mathcal C$.

\begin{thm}\label{satz-brauchbar}%
Let $R$ be a conditional bisimulation. Then
$ R^\prime = \{ (a;d, b;d) \mid (a,b,\mathcal{C}) \in R \allowbreak\land d \models \mathcal{C} \} $
is a bisimulation for the reaction relation
%\footnote{Using a single dummy label for all transitions.}
$\leadsto$.
\end{thm}

\begin{proof}
To prove that $R'$ is a bisimulation, we need to show that
if $a;d \mathrel{R^\prime} b;d$ and $a;d \leadsto a^\prime$, then there exists $b^\prime$ such that $b;d \leadsto b^\prime$ and $a^\prime \mathrel{R^\prime} b^\prime$;
and vice versa.
Equivalently, if $a;d$ can do a step, then $b;d$ can answer this step (and vice versa) and the result is again contained in the bisimulation $R'$.
We show only one direction ($a;d$ answered by $b;d$), the other one can be done analogously.

Now let some $(a;d, b;d) \in R^\prime$ be given, for which there must exist a triple $(a, b, \mathcal{C}) \in R$.
Consider a step $a;d \leadsto a^\prime$. This step is due to some rule $(\ell, r, \mathcal{R})$, shown graphically in \Cref{fig-brauchbar-ctxstep}.

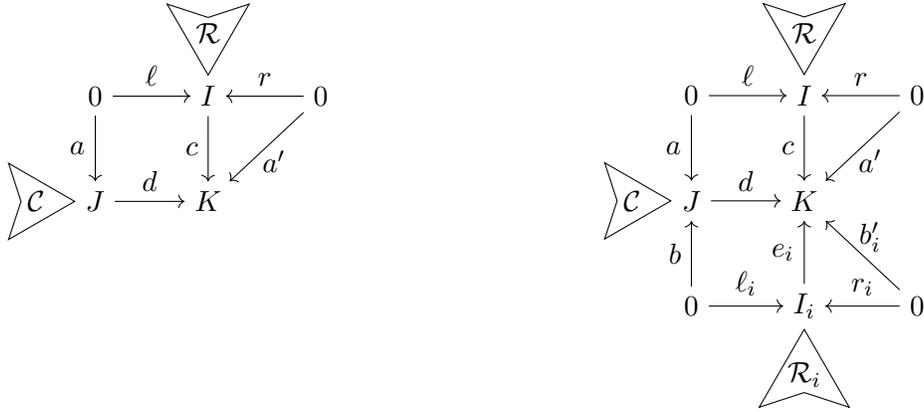
\begin{figure}[hb]
  \begin{subfigure}[t]{0.48\textwidth}
    \centering\begin{tikzpicture}
      \def\sqw{1.5}
      \def\sqh{1.4}

      \node (tlempty) at (0,0) {$\obnull$};
      \node (i) at (1*\sqw,0) {$I$};
      \node (trempty) at (2*\sqw,0) {$\obnull$};
      \node (j) at (0,-1*\sqh) {$J$};
      \node (k) at (1*\sqw,-1*\sqh) {$K$};
      \draw[->] (tlempty) -- node[above]{$\ell$} (i);
      \draw[->] (trempty) -- node[above]{$r$} (i);
      \draw[->] (trempty) -- node[below,pos=0.4]{$a^\prime$} (k);

      \draw[->] (tlempty) -- node[left]{$a$} (j);
      \draw[->] (i) -- node[left]{$c$} (k);

      \draw[->] (j) -- node[above]{$d$} (k);

      \node[condtri,dart tip angle=60,shape border rotate=270] at (i.north) {$\mathcal{R}$};
      \node[condtri,dart tip angle=60,shape border rotate=0] at (j.west) {$\kern 0.5mm \mathcal{C}$};
      % invisible:
      \node (ii) at (1*\sqw,-2*\sqh) {$\phantom{I_i}$};
      \node[condtri,dart tip angle=60,shape border rotate=90,draw=none] at (ii.south) {$\kern -0.5mm\phantom{\mathcal{R}_i}\kern -0.5mm$};
    \end{tikzpicture}%
    \caption{A reaction $a;d \leadsto a'$ for some context $d$ that satisfies $\mathcal C$.}%
    \label{fig-brauchbar-ctxstep}%
  \end{subfigure}\hfill%
  \begin{subfigure}[t]{0.48\textwidth}
    \centering\begin{tikzpicture}
      \def\sqw{1.5}
      \def\sqh{1.4}

      \node (tlempty) at (0,0) {$\obnull$};
      \node (i) at (1*\sqw,0) {$I$};
      \node (trempty) at (2*\sqw,0) {$\obnull$};
      \node (j) at (0,-1*\sqh) {$J$};
      \node (k) at (1*\sqw,-1*\sqh) {$K$};
      \node (blempty) at (0,-2*\sqh) {$\obnull$};
      \node (ii) at (1*\sqw,-2*\sqh) {$I_i$};
      \node (brempty) at (2*\sqw,-2*\sqh) {$\obnull$};
      \draw[->] (tlempty) -- node[above]{$\ell$} (i);
      \draw[->] (trempty) -- node[above]{$r$} (i);
      \draw[->] (trempty) -- node[below,pos=0.4]{$a^\prime$} (k);
      \draw[->] (blempty) -- node[above]{$\ell_i$} (ii);
      \draw[->] (brempty) -- node[above]{$r_i$} (ii);
      \draw[->] (brempty) -- node[above,pos=0.4]{$b_i^\prime$} (k);

      \draw[->] (tlempty) -- node[left]{$a$} (j);
      \draw[->] (blempty) -- node[left]{$b$} (j);
      \draw[->] (i) -- node[left]{$c$} (k);
      \draw[->] (ii) -- node[left]{$e_i$} (k);

      \draw[->] (j) -- node[above]{$d$} (k);

      \node[condtri,dart tip angle=60,shape border rotate=270] at (i.north) {$\mathcal{R}$};
      \node[condtri,dart tip angle=60,shape border rotate=0] at (j.west) {$\kern 0.5mm \mathcal{C}$};
      \node[condtri,dart tip angle=60,shape border rotate=90] at (ii.south) {$\kern -0.5mm\mathcal{R}_i\kern -0.5mm$};
    \end{tikzpicture}%
    \caption{The same step, together with answering steps provided by the conditional bisimulation.}%
    \label{fig-brauchbar-answers}%
  \end{subfigure}%
  \caption{The steps considered in the proof of \Cref{satz-brauchbar}}%
  \label{fig-brauchbar-proof}
\end{figure}

We have $c \models \mathcal R$ (otherwise the rule would not be applicable and therefore the step $a;d \leadsto a'$ would not be possible)
and $d \models \mathcal C$ (follows from the given $(a;d, b;d) \in R^\prime$ by construction of $R'$).
To make them usable for the answering steps, we transform $\mathcal R, \mathcal C$ to be conditions over $K$.
Trivially $d = d;\id_K$, so using \Cref{shift-law-def} we rewrite $d;\id_K \models \mathcal{C}$ to $\id_K \models \mathcal{C}_{\downarrow d}$.
Analogously, we rewrite $c;\id_K \models \mathcal{R}$ to $\id_K \models \mathcal{R}_{\downarrow c}$.

We set $\mathcal{A} \defeq \mathcal{R}_{\downarrow c}$ and interpret this diagram as $a \cstep{d}{\mathcal{A}} a^\prime$.
Since $R$ is a conditional bisimulation and $(a,b,\mathcal{C}) \in R$, $b$ can answer the step of $a$, using one of possibly several rules $(\ell_i, r_i, \mathcal{R}_i)$ depending on the given context $d$.
Setting $\mathcal{B}_i \defeq \mathcal{R}_{i \downarrow e_i}$, we get steps $b \cstep{d}{\mathcal{B}_i} b_i^\prime$ and we extend the diagram as shown in \Cref{fig-brauchbar-answers}.

Generally, not every answering step that is possible for our given triple $(a,b,\mathcal{C})$ is a suitable answering step for the given context $d$.
But since $R$ is a conditional bisimulation, we know that $\mathcal{R}_{\downarrow c} \land \mathcal{C}_{\downarrow d} \models \bigvee_{i \in I} \left( \mathcal{C}_i^\prime \land \mathcal{B}_i \right)$.
Previously we derived $\id_K \models \mathcal{R}_{\downarrow c} \land \mathcal{C}_{\downarrow d}$.
Therefore, $\id_K$ also satisfies $\bigvee_{i \in I} \left( \mathcal{C}_i^\prime \land \mathcal{B}_i \right)$, that is, $\id_K$ satisfies $\mathcal{C}_i^\prime \land \mathcal{B}_i = \mathcal{C}_i^\prime \land \mathcal{R}_{i \downarrow e_i}$ for some $i$. From now on, we only consider answering steps for which this is indeed the case.

Using \Cref{shift-law-def} we rewrite $\id_K \models \mathcal{R}_{i \downarrow e_i}$ to $e_i;\id_K = e_i \models \mathcal{R}_i$, which means that the rule $(\ell_i, r_i, \mathcal{R}_i)$ can actually be applied, that is, $b;d \leadsto b_i^\prime = r_i;e_i$. So $b$ has a suitable answering step.

To show that $R'$ is a bisimulation, we only have to show that $(a^\prime, b_i^\prime) \in R'$.
As $R$ is a conditional bisimulation, for the given answering step we know that $(a^\prime, b_i^\prime, \mathcal{C}_i^\prime) \in R$.
Previously we had $\id_K \models \mathcal{C}_i^\prime$, therefore, the requested pair $(a^\prime;\id_K, b_i^\prime;\id_K) = (a^\prime, b_i^\prime)$ is added during the construction of $R'$.
\end{proof}

\begin{rem}\label{converse-brauchbar}
Note that the converse of \Cref{satz-brauchbar} (if $R'$ is a bisimulation, then $R$ is a conditional bisimulation) does not hold.
Consider the following counterexample:
  \[\bgroup%
  \renewcommand{\arraystretch}{1.5}%
  \setlength{\arraycolsep}{0.1em}%
  \begin{array}{rlccccclccccl}
    R_1    &= \big( & \inlinenlg{} & \rightarrow & \inlinenlg{a,x,fw} & \leftarrow & \inlinenlg{x,fw}&,
    \                 \inlinenlg{} & \rightarrow & \inlinenlg{e,x,fw} & \leftarrow & \inlinenlg{x,fw}&,\  \condtrue_X \big) \\
    R_2    &= \big( & \inlinenlg{} & \rightarrow & \inlinenlg{b,fw}   & \leftarrow & \inlinenlg{fw}&,
    \                 \inlinenlg{} & \rightarrow & \inlinenlg{e,fw}   & \leftarrow & \inlinenlg{fw}&,\  \mathcal A_{\exists X} \big) \\
    \mathcal A_{\exists X} &= \mathrlap{\Big(
      \inlinenlg{}, \exists, \big\{ (
        \inlinenlg{} \rightarrow \inlinenlg{x} \leftarrow \inlinenlg{x},\  \condtrue_X
      ) \big\}
    \Big)}
  \end{array}\egroup\]
  Here $X = \inlinenlg{x}$. In this case, an $a$-loop can be replaced with an $e$-loop if an $x$-loop is present, ensured by requiring (and retaining) it in the rule $R_1$.
  A $b$-loop can also be replaced with an $e$-loop, also if an $x$-loop is present, this time ensured by an application condition.

  Now consider the conditional relation
  $R = \{ (\inlinenlg{a}, \inlinenlg{b}, \condtrue) \} \cup \{ (G,G,\condtrue) \mid \text{$G$ is}\allowbreak\text{a graph} \}$\footnote{
  The reflexive triples $(G,G,\condtrue)$ are needed because
  $\inlinenlg{a}, \inlinenlg{b}$ can both be transformed to $\inlinenlg{e}$ in the presence of $\inlinenlg{x}$ and we require
  that the resulting (identical) graphs are related.}
  (all graphs are seen as cospans with empty interfaces)
  and the accompanying relation
  $R' = \{ (a;d, b;d) \mid (a,b,\mathcal C) \in R \land d \models \mathcal C \}$.

  Clearly, the graphs $\inlinenlg{a}$ and $\inlinenlg{b}$ are bisimilar under all contexts,
  and therefore $R'$ is a bisimulation:
  either the context contains an $x$-loop, then they both reduce to a graph that contains $\inlinenlg{e,x}$ and possibly further context (both steps reach the same graph),
  or the context does not contain an $x$-loop, in which case neither rule is applicable.

  However, $R$ is not a conditional bisimulation, the violating triple being $(\inlinenlg{a},\inlinenlg{b},\condtrue)$:
  the step $\inlinenlg{b} \cstep{\id}{\mathcal A_{\exists X}} \inlinenlg{e}$ cannot be answered by $\inlinenlg{a}$,
  since in $R_1$, the $x$-loop is directly participating in the
  reaction, but $\mathcal A_{\exists X}$ only guarantees its existence
  in a passive environment (i.e.\ it is \emph{not} participating in the reaction).
  ($\inlinenlg{a}$ could only do a step by borrowing $\inlinenlg{x}$, but this does not constitute a valid answering step for the step of $\inlinenlg{b}$ where $\id$ (i.e.\ no additional elements) has been borrowed.)
\end{rem}

Next, we will show that conditional bisimilarity $\csimC$ is a conditional congruence.
This is an important plausibility check, since reactive systems have
been introduced with the express purpose to define and reason about
bisimulation congruences.

\begin{lem}\label{cb-congruence}%
Conditional bisimilarity $\csimC$ is a conditional congruence.
\end{lem}

\begin{proof}
We show that $\csimC$ is:

\begin{description}
% quick hack to make the align*'s in two nested lists look a bit better
% I'm not sure about how to fix this in a general way, but the main issue is
% that align* seems to be unaware of the fact that it's in a (nested) list,
% and centers the equation over the whole pagewidth instead, which causes it
% to appear offset to the left.
% (Feel free to tweak the exact number used below, or to apply a proper fix
% for align*/list interaction :))
\csname@fleqntrue\endcsname\csname@mathmargin\endcsname=2.3cm
\item[reflexive] We prove that $R = \{(a,a,\mathcal{C})\mid \codom(a) =
  \Ro(\mathcal{C})\}$ is a conditional bisimulation.

  Any context step $a \cstep{f}{\mathcal A} a'$ can be trivially
  answered by the exact same step, setting
  $\mathcal B_1 \defeq \mathcal A,\ \mathcal C_1^\prime \defeq
  \mathcal C_{\downarrow f},\ b_1^\prime \defeq a'$, where
  $I = \{1\}$, and we have
  $(a',b'_1,\mathcal{C}'_1) = (a', a', \mathcal C_{\downarrow f}) \in
  R$.

\item[symmetric]
  Let $R$ be a conditional bisimulation.
  It is easily seen that $R^{-1}$, due to the symmetric nature of the definition, is also a conditional bisimulation.
  Then, $R \subseteq \csimC$ implies $R^{-1} \subseteq \csimC$,
  which proves symmetry of $\csimC$.

\item[transitive] Let $R_1, R_2$ be conditional bisimulations that are
  closed under condition strengthening, i.e.\ %
  $(a,b,\mathcal C')\in R_i$ and $\mathcal C \models \mathcal C'$
  implies $(a,b,\mathcal C) \in R_i$.  We show that
  $R_1R_2 \defeq \{ (a,c,\mathcal D) \mid \text{there exists $b$ such that }
  (a,b,\mathcal D) \in R_1 \,\text{ and }\, (b,c,\mathcal D) \in R_2 \}$
  is a conditional bisimulation.

Then, since $\csimC$ is a conditional bisimulation closed under
condition strengthening (\Cref{cond-strengthening}),
$(a,b,\mathcal D) \in \csimC, (b,c,\mathcal D) \in \csimC$ implies
$(a,c,\mathcal D) \in \csimC\csimC$ and, since we show that
$\csimC\csimC$ is a conditional bisimulation,
$a,c$ are conditionally bisimilar under $\mathcal D$, which proves
transitivity of $\csimC$.

Consider a triple $(a,c,\mathcal D) \in R_1R_2$,
which by construction of $R_1R_2$ results from some $(a,b,\mathcal D) \in R_1,\ (b,c,\mathcal D) \in R_2$.
Also consider a step
$a \cstep{f}{\mathcal A} a'$.
Then, $R_1R_2$ fulfills the requirements of a conditional bisimulation:

\begin{enumerate}
\item \emph{Answering steps $c \cstep{f}{\mathcal C_{i,j}} c_{i,j}^\prime$:}
Since $(a,b,\mathcal D) \in R_1$ and $R_1$ is a conditional bisimulation,
we know that there exist answering steps
$b \cstep{f}{\mathcal B_i} b_i^\prime$ such that for all $i \in I$,
$(a', b_i^\prime, \mathcal D_i^\prime) \in R_1$ and
$\mathcal A \land \mathcal D_{\downarrow f} \models \bigvee_{i \in I} \left( \mathcal D_i^\prime \land \mathcal B_i \right)$.

Additionally, since $(b,c,\mathcal D) \in R_2$, for each
$b \cstep{f}{\mathcal B_i} b_i^\prime$ there exist answering steps
$c \cstep{f}{\mathcal C_{i,j}} c_{i,j}^\prime$ such that for all $j \in J_i$ we have
$(b_i^\prime, c_{i,j}^\prime, \mathcal D_{i,j}^\prime) \in R_2$ and
$\mathcal B_i \land \mathcal D_{\downarrow f} \models \bigvee_{j \in J_i} \big( \mathcal D_{i,j}^\prime \land \mathcal C_{i,j} \big)$.

We now collect all answering steps $c \cstep{f}{\mathcal C_{i,j}} c_{i,j}^\prime$
and use them as answering steps for the original step $a \cstep{f}{\mathcal A} a'$.

\item \emph{$(a', c_{i,j}', \mathcal D_{i,j}^\prime) \in R_1R_2$:}
Since $(a', b_i^\prime, \mathcal D_i^\prime) \in R_1$, and $R_1$ is closed
under condition strengthening, and
$\mathcal D_{i,j}^\prime \land \mathcal D_i^\prime \models \mathcal D_i^\prime$,
we also have $(a',\ b_i^\prime,\ \mathcal D_{i,j}^\prime \land
\mathcal D_i^\prime) \in R_1$ for all $j\in J_i$.
Similarly, we obtain $(b_i^\prime,\ c_{i,j}^\prime,\ \mathcal D_{i,j}^\prime \land \mathcal D_i^\prime) \in R_2$.
By construction of $R_1R_2$ we then also have $(a',\ c_{i,j}^\prime,\ \mathcal D_{i,j}^\prime \land \mathcal D_i^\prime) \in R_1R_2$.

\item \emph{$\mathcal A \land \mathcal D_{\downarrow f} \models \bigvee_{i \in I} \left( \mathcal D_i^\prime \land \mathcal C_i^\prime \right)$:}
We now rewrite

\begin{align*}
  \mathcal A \land \mathcal D_{\downarrow f} &\models \bigvee\nolimits_{i \in I} \left( \mathcal D_i^\prime \land \mathcal B_i \right)
  &&\text{(given from $(a,b,\mathcal D) \in R_1$)} \\
  \mathcal A \land \mathcal D_{\downarrow f} &\models \bigvee\nolimits_{i \in I} \left( \mathcal D_i^\prime \land \mathcal B_i \right) \land \mathcal D_{\downarrow f}
  &&\text{(since $\mathcal D_{\downarrow f}$ is given on the left already)}
  \\
  \mathcal A \land \mathcal D_{\downarrow f} &\models \bigvee\nolimits_{i \in I} \left( \mathcal D_i^\prime \land \mathcal B_i \land \mathcal D_{\downarrow f} \right)
  &&\text{(distributivity)}
\intertext{From $(b,c,\mathcal D) \in R_2$ we know $\mathcal B_i \land \mathcal D_{\downarrow f} \models \bigvee_{j \in J_i} \big( \mathcal D_{i,j}^\prime \land \mathcal C_{i,j} \big)$, therefore this also implies:}
  \mathcal A \land \mathcal D_{\downarrow f} &\models \mathrlap{\bigvee\nolimits_{i \in I} \Big( \mathcal D_i^\prime \land \bigvee\nolimits_{j \in J_i} \big( \mathcal D_{i,j}^\prime \land \mathcal C_{i,j} \big) \Big)}
  \\
  \mathcal A \land \mathcal D_{\downarrow f} &\models \mathrlap{\bigvee\nolimits_{i \in I,\, j \in J_i} \big( \mathcal D_i^\prime \land \mathcal D_{i,j}^\prime \land \mathcal C_{i,j} \big)}
\end{align*}

Observe that this implication is of the required form for the previously
derived tuples $(a',\ c_{i,j}^\prime,\ \mathcal D_{i,j}^\prime \land \mathcal D_i^\prime) \in R_1R_2$.

In case $I$ is infinite, the same idea can be applied, but we need to slightly change the notation to prevent the creation of infinite disjunctions:
\begin{align*}
  c \models (\mathcal A \land \mathcal D_{\downarrow f}) &\implies \exists i \in I \colon c \models ( \mathcal D_i' \land \mathcal B_i)
  &&\text{(since $(a,b,\mathcal D) \in R_1$)} \\
  c \models (\mathcal A \land \mathcal D_{\downarrow f}) &\implies \mathrlap{\exists i \in I \colon c \models ( \mathcal D_i' \land \mathcal B_i ) \land c \models \mathcal D_{\downarrow f} }
  &&\text{($\mathcal D_{\downarrow f}$ on left)}
  \\
  c \models (\mathcal A \land \mathcal D_{\downarrow f}) &\implies \exists i \in I \colon c \models ( \mathcal D_i' \land \mathcal B_i \land \mathcal D_{\downarrow f} )
  &&\text{(distributivity)}
\intertext{From $(b,c,\mathcal D) \in R_2$ we know $c \models (\mathcal B_i \land \mathcal D_{\downarrow f}) \implies \exists j \in J_i \colon ( \mathcal D_{i,j}^\prime \land \mathcal C_{i,j} )$, therefore this also implies:}
  c \models (\mathcal A \land \mathcal D_{\downarrow f}) &\implies \mathrlap{\exists i \in I \colon c \models \mathcal D_i' \land \exists j \in J_i \colon c \models ( \mathcal D_{i,j}^\prime \land \mathcal C_{i,j} ) }
  \\
  c \models (\mathcal A \land \mathcal D_{\downarrow f}) &\implies \mathrlap{\exists i \in I, j \in J_i \colon c \models (\mathcal D_i' \land \mathcal D_{i,j}^\prime \land \mathcal C_{i,j} ) }
\end{align*}
\end{enumerate}

\noindent% To visually distinguish it from the previous enumeration
Symmetrically, steps $c \cstep{f}{\mathcal C} c'$ can be answered by $a$.
Therefore, $R_1R_2$ is a conditional bisimulation.

\item[closed under contextualization]
We show that $u(R) = \left\{ (a;d, b;d, \mathcal{C}_{\downarrow d}) \mid (a,b,\mathcal C) \in R \right\}$ is a conditional bisimulation, assuming $R$ is a conditional bisimulation.

Consider a triple $(a;d, b;d, \mathcal{C}_{\downarrow d}) \in u(R)$
and a step $a;d \cstep{f}{\mathcal A} a'$.
This step is due to some rule $(\ell, r, \mathcal R)$.
We have to show that there exist answering steps $b;d \cstep{f}{\mathcal{B}_i} b_i^\prime$
such that $(a', b_i^\prime, \mathcal{C}_i^\prime) \in u(R)$
and $\mathcal A \land (\mathcal{C}_{\downarrow d})_{\downarrow f} \models \bigvee_{i \in I} \left( \mathcal{C}_i^\prime \land \mathcal{B}_i \right)$.

According to \Cref{lemma-ctxtrans-composition}, the given step can be rewritten to $a \cstep{d;f}{\mathcal A} a'$.

By construction of $u(R)$, for the given triple $(a;d, b;d, \mathcal{C}_{\downarrow d}) \in u(R)$ there must exist a triple $(a,b,\mathcal C) \in R$.
As $R$ is a conditional bisimulation, the step $a \cstep{d;f}{\mathcal A} a'$
has answering steps $b \cstep{d;f}{\mathcal B_i} b_i^\prime$
such that $(a', b_i^\prime, \mathcal{C}_i^\prime) \in R$
and $\mathcal A \land \mathcal{C}_{\downarrow d;f} \models \bigvee_{i \in I} \left( \mathcal{C}_i^\prime \land \mathcal{B}_i \right)$.

However, we are not interested in answering steps for the context $d;f$, but rather for the original step $a;d \cstep{f}{\mathcal A} a'$.

\begin{enumerate}
\item \emph{Answering steps $b;d \cstep{f}{\mathcal{B}_i} b_i^\prime$\@:}
According to \Cref{lemma-ctxtrans-composition},
the answering steps $b \cstep{d;f}{\mathcal B_i} b_i^\prime$
can be rewritten to $b;d \cstep{f}{\mathcal B_i} b_i^\prime$.

\item \emph{$(a^\prime, b_i^\prime, \mathcal{C}_i^\prime) \in u(R)$\@:}
Since $(a', b_i^\prime, \mathcal{C}_i^\prime) \in R$,
we have $(a';\id,\ b_i^\prime;\id,\ \mathcal{C}_{i \downarrow \id}^\prime) = (a', b_i^\prime, \mathcal{C}_i^\prime) \in u(R)$.

\item \emph{$\mathcal{A} \land (\mathcal{C}_{\downarrow d})_{\downarrow f} \models \bigvee_{i \in I} \left( \mathcal{C}_i^\prime \land \mathcal{B}_i \right)$\@:}
The implication to be shown is identical to the one that we obtained above,
except for $(\mathcal{C}_{\downarrow d})_{\downarrow f}$, which, however, is
equivalent to $\mathcal C_{\downarrow d;f}$.
\end{enumerate}
Answering steps for $b;d \cstep{f}{\mathcal B} b'$ can be derived analogously.
\qedhere
\end{description}
\end{proof}

\subsection{Alternative Characterization using Fixpoint Theory}\label{subsec-fixpoint-basics}

Behavioural equivalences can be characterized as fixpoints of certain functions on complete lattices~\cite{ps:enhancements-coinductive}.
Before we provide definitions that characterize conditional bisimulation relations as fixpoints, we provide a quick summary of fixpoint theory.
We do not rely on this characterization in the proofs in this section.
However, we will use the theory and the alternative definitions for the proofs of up-to techniques in \Cref{sec:uptocond}.

A \emph{complete lattice} is a partially ordered set $(L,
\sqsubseteq)$ where each subset $Y\subseteq L$ has an infimum,
denoted by $\bigsqcap Y$ and a supremum, denoted by $\bigsqcup Y$.
In this paper, the type of lattices that we consider
contain relations ordered by inclusion,
i.e.\ the elements of the lattice are relations
and thus the functions we consider map relations to relations.

A function $f \colon L \to L$ is \emph{monotone} if for all $l_1, l_2 \in L$,\;
$l_1 \sqsubseteq l_2$ implies $f(l_1) \sqsubseteq f(l_2)$,
\emph{idempotent} if $f \circ f = f$,
and \emph{extensive} if $l \sqsubseteq f(l)$ for all $l \in L$.
When $f$ is monotone, extensive and idempotent it is called an \emph{(upper) closure}.
In this case, $f(L) = \{ f(l) \mid l \in L \}$ is a complete lattice.

  Given some behavioural equivalence, we define a monotone function
  $f \colon L \to L$ in such a way that its greatest fixpoint $\nu f$
  equals the behavioural equivalence.  Behavioural equivalence can
  then be checked by establishing whether some given element of the
  lattice $l \in L$ (for instance a relation consisting of a single
  pair) is under the fixpoint, i.e., if $l \sqsubseteq \nu f$.
  By Tarski's Theorem~\cite{t:lattice-fixed-point},
  $\nu f = \bigsqcup \{ x \mid x \sqsubseteq f(x) \}$, i.e., the
  greatest fixpoint is the supremum of all post-fixpoints. Hence for
  showing that $l\sqsubseteq \nu f$, it is sufficient to prove that
  $l$ is under some post-fixpoint $l'$, i.e.,
  $l \sqsubseteq l'\sqsubseteq f(l')$.

Using these preliminaries, we can now give an alternative characterization
of conditional bisimulation using fixpoint theory:

\begin{rem}[Conditional bisimulation function $f_C$]\label{csimc-fc}
  Consider the complete lattice \linebreak $\Condrel$, which is the set of all conditional relations ordered by set inclusion.
  Then, conditional bisimulations can also be seen as post-fixpoints of $f_C$
  (i.e.\ $R$~is a conditional bisimulation if and only if $R \subseteq f_C(R)$),
  and conditional bisimilarity as the greatest fixpoint of~$f_C$
  (i.e.\ $\csimC = \nu f_C$),
  where the monotone function $f_C \colon \Condrel \to \Condrel$ is
  the conditional bisimulation function
  defined by
  \begin{align*}
  f_C(R) \defeq \big\{ (a,b,\mathcal C) \;\big|\; &
    \text{for each
    $a \cstep{f}{\mathcal A} a'$
    there exist
    $b \cstep{f}{\mathcal B_i} b_i'$
    and
    conditions $\mathcal C_i'$
    } \\& \text{%
    such that
    $(a', b_i', \mathcal C_i') \in R$
    and
    $\mathcal A \land \mathcal C_{\downarrow f} \models
    {\textstyle\bigvee\nolimits_{i \in I}} (\mathcal C_i' \land
    \mathcal B_i)$,
    } \\& \text{%
    and vice versa (for each $\smash{b \cstep{f}{\mathcal B} b'}$\dots)%
    }
  \big\}
  \end{align*}
  The correctness of this characterization (i.e.\ that ``$f_C$ is the right function'') can be seen by
  expanding the definition of $f_C$ on the right-hand side of $(a,b,\mathcal C) \in R \implies (a,b,\mathcal C) \in f_C(R)$,
  which results in exactly the definition of a conditional
  bisimulation relation.
\end{rem}

\subsection{Representative Conditional Bisimulations}\label{subsec-repr-cond-bisim}

Checking whether two arrows are conditionally bisimilar, or whether a
given relation is a conditional bisimulation, can be hard in
practice, since we have to check all possible context steps, of which
there are typically infinitely many.

% To reduce the number of contexts to be checked, we use representative steps instead of context steps.
% In many categories, a representative class of squares can be chosen such that $\kappa(a,\ell)$ is always finite.
% Then we only have to check finitely many steps (consisting of finitely
% many contexts and the weakest possible condition). This is analogous
% to the solution chosen in~\cite{LM00}, where the notion of
% representative squares is replaced by idem pushouts.
For saturated bisimilarity, we used representative steps instead of context
steps (cf.\ \Cref{subsec-rsq,subsec-rstep}) to reduce the number of contexts to
be checked.
In this section, we extend our definition of conditional bisimulation to use
representative steps and prove that the resulting bisimilarity is identical to
the one previously defined.

\begin{defi}[Representative conditional bisimulation]\label{def-cb-r}%
  We fix a conditional reactive system.
A \emph{representative conditional bisimulation} $R$ is a conditional relation such that the following holds:
for each triple $(a,b,\mathcal{C}) \in R$ and each representative step $a \rstep{f}{\mathcal A} a^\prime$,
there are answering context steps $b \cstep{f}{\mathcal{B}_i} b_i^\prime$ and conditions $\mathcal{C}_i^\prime$
such that $(a^\prime, b_i^\prime, \mathcal{C}_i^\prime) \in R$
and $\mathcal{A} \land \mathcal{C}_{\downarrow f} \models \bigvee_{i \in I} \left( \mathcal{C}_i^\prime \land \mathcal{B}_i \right)$;
and vice versa.
Two arrows are \emph{representative conditionally bisimilar under $\mathcal C$} ($(a, b, \mathcal{C}) \in \csimR$) whenever a representative conditional bisimulation $R$ with $(a, b, \mathcal{C}) \in R$ exists.
\end{defi}
\begin{rem}\label{csimr-fr}
  Analogously to \Cref{csimc-fc}, we can define
  representative conditional bisimulations as post-fixpoints of $f_R$,
  and representative conditional bisimilarity as the greatest fixpoint
  of $f_R$, where $f_R$ is defined on $\Condrel$ as follows:
  \begin{align*}
  f_R(R) \defeq \big\{ (a,b,\mathcal C) \;\big|\; &
    \text{for each
    $a \rstep{f}{\mathcal A} a'$
    there exist
    $b \cstep{f}{\mathcal B_i} b_i'$
    and
    conditions $\mathcal C_i'$
    } \\& \text{%
    such that
    $(a', b_i', \mathcal C_i') \in R$
    and
    $\mathcal A \land \mathcal C_{\downarrow f} \models {\textstyle\bigvee\nolimits_{i \in I}} (\mathcal C_i' \land \mathcal B_i)$,
    } \\& \text{%
    and vice versa (for each $\smash{b \rstep{f}{\mathcal B} b'}$\dots)%
    }
  \big\}
  \end{align*}
  It is easy to see that $f_C \subseteq f_R$:
  Their definitions differ only in the type of steps which are checked.
  A triple that satisfies the requirements for all context steps (is in $f_C(R)$)
  naturally satisfies them for all representative steps (is in $f_R(R)$),
  since every representative step is also a context step.
\end{rem}

To show that the two conditional bisimilarities using context and representative steps are equivalent (\Cref{simr-simc}) and for the proofs of \Cref{cbuts-equiv,cec-cb,cb-vs-id}, we need the following two lemmas:

\begin{lemC}[{\cite[Lemma 16]{DBC-CRS}}]\label{lemma16-dbccrs}%
Given a context step $a \cstep{f}{\mathcal A} a^\prime$, the borrowed context $f$ can be extended by an additional context $c'$, that is:
\mbox{$a \cstep{f}{\mathcal A} a'$ implies $a \cstep{f;c'}{\mathcal{A}_{\downarrow c'}} a';c'$}
\end{lemC}

\begin{lem}\label{lemma-cbr-cuc}%
Representative conditional bisimilarity is closed under contextualization, that is,
$(a,b,\mathcal{C}) \in \csimR$ implies $(a;d, b;d, \mathcal{C}_{\downarrow d}) \in \csimR$.
\end{lem}

\begin{proof}
We show that $u(R) = \left\{ (a;d, b;d, \mathcal{C}_{\downarrow d}) \mid (a,b,\mathcal C) \in R \right\}$ is a representative conditional bisimulation, assuming $R$ is a representative conditional bisimulation.

Consider a triple $(a;d, b;d, \mathcal{C}_{\downarrow d}) \in u(R)$
and a step $a;d \rstep{f}{\mathcal A} a'$.
This step is due to some rule $(\ell, r, \mathcal R)$.
We have to show that there exist answering steps $b;d \cstep{f}{\mathcal{B}_i} b_i^\prime$
such that $(a', b_i^\prime, \mathcal{C}_i^\prime) \in u(R)$
and $\mathcal A \land (\mathcal{C}_{\downarrow d})_{\downarrow f} \models \bigvee_{i \in I} \left( \mathcal{C}_i^\prime \land \mathcal{B}_i \right)$.

The representative step is of course also a context step ($a;d \cstep{f}{\mathcal A} a'$),
which, according to \Cref{lemma-ctxtrans-composition}, can be rewritten to $a \cstep{d;f}{\mathcal A} a'$.
As a result, the step is not necessarily a representative one anymore. However we can find a matching representative step (see also \Cref{cond-repr-zurueckfuehren}):

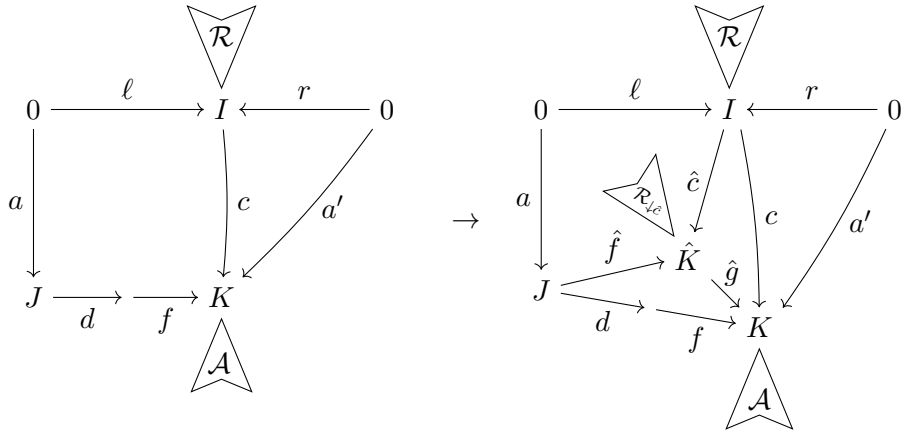
\begin{center}\begin{tikzpicture}
  % %%% LEFT - PRE R-STEP %%%
  \begin{scope}[shift={(0,0)}]
    \node (tlempty) at (0,0) {$\obnull$};
    \node (i) at (2.5,0) {$I$};
    \node (trempty) at (4.7,0) {$\obnull$};
    \node (j) at (0,-2.5) {$J$};
    \node (k) at (2.5, -2.5) {$K$};
    \node[inner sep=2pt] (ki) at ($(j)!0.5!(k)$) {};
    \draw[->] (tlempty) -- node[above]{$\ell$} (i);
    \draw[->] (trempty) -- node[above]{$r$} (i);
    \draw[->] (trempty) to[bend left=5] node[right]{$a'$} (k);

    \draw[->] (tlempty) -- node[left]{$a$} (j);
    \draw[->] (i) to[bend left=5] node[right]{$c$} (k);

    \draw[->] (j) -- node[below]{$d$} (ki);
    \draw[->] (ki) -- node[below]{$f$} (k);

    \node[condtri,shape border rotate=270] at (i.north) {$\mathcal{R}$};
    \node[condtri,shape border rotate=90] at (k.south) {$\mkern-3mu\mathcal{A}$};
  \end{scope}

  \node at (5.75,-1.5) {$\rightarrow$};

  % %%% RIGHT %%%
  \begin{scope}[shift={(6.75,0)}]
    \node (tlempty) at (0,0) {$\obnull$};
    \node (i) at (2.5,0) {$I$};
    \node (trempty) at (4.7,0) {$\obnull$};
    \node (j) at (0,-2.4) {$J$};
    \node (kp) at (1.95,-1.95) {$\hat K$};
    \node (k) at (2.9, -2.9) {$K$};
    \node[inner sep=2pt] (ki) at ($(j)!0.5!(k)$) {};
    \draw[->] (tlempty) -- node[above]{$\ell$} (i);
    \draw[->] (trempty) -- node[above]{$r$} (i);
    \draw[->] (trempty) to[bend left=5] node[right]{$a'$} (k.40);

    \draw[->] (tlempty) -- node[left]{$a$} (j);
    \draw[->] (i.300) to[bend left=5] node[right]{$c$} (k);

    \draw[->] (j) -- node[below]{$d$} (ki);
    \draw[->] (ki) -- node[below]{$f$} (k);

    \draw[->] (i) -- node[left]{$\hat c$} (kp);
    \draw[->] (j) -- node[above]{$\hat f$} (kp);
    \draw[->] (kp) -- node[right,pos=0.1]{\raisebox{9pt}{$\hat g$}} (k);

    \node[condtri,shape border rotate=270] at (i.north) {$\mathcal R$};
    \node[condtri,shape border rotate=90] at (k.south) {$\mathcal A$};
    % shift is applied before rotate that's why the coordinates are so odd - too lazy to fix
    \node[rotate around={-55:(kp.center)},condtri,shape border rotate=0,scale=0.7] at (kp.west) {%
      \raisebox{-5pt}[5pt][2pt]{%
        \rotatebox{25}{$\mathcal{R}_{\downarrow \hat c}$}%
      }\kern-2pt
    };
  \end{scope}
\end{tikzpicture}
\end{center}

Note that $\hat f;\hat g = d;f,\ \hat c;\hat g = c$.
Variables with a hat (e.g.\ $\hat c$) refer to the representative step,
but otherwise play the same role than their unhatted counterparts (e.g.\ $c$), which refer to the original step.
The result is a representative step
$a \rstep{\hat f}{\mathcal{R}_{\downarrow \hat c}} r;\hat c$.

By construction of $u(R)$, for the given triple $(a;d, b;d, \mathcal{C}_{\downarrow d}) \in u(R)$ there must exist a triple $(a,b,\mathcal C) \in R$.
As $R$ is a representative conditional bisimulation, the step $a \rstep{\hat f}{\mathcal{R}_{\downarrow \hat c}} r;\hat c$
has answering steps $b \cstep{\hat f}{\hat{\mathcal{B}_i}} \hat{b_i^\prime}$
such that $(r;\hat c, \hat{b_i^\prime}, \hat{\mathcal{C}_i^\prime}) \in R$
and $\mathcal{R}_{\downarrow \hat c} \land \mathcal{C}_{\downarrow \hat f} \models \bigvee_{i \in I} \left( \hat{\mathcal{C}_i^\prime} \land \hat{\mathcal{B}_i}\right)$.

However we are not interested in answering steps for the representative context $\hat f$, but rather for the original step $a;d \rstep{f}{\mathcal A} a'$, that is, we need answering steps of $b;d$ using context $f$.
So we need
(1) answering steps $b;d \cstep{f}{\mathcal{B}_i} b_i^\prime$
(2) such that $(a^\prime, b_i^\prime, \mathcal{C}_i^\prime) \in u(R)$
and (3) $\mathcal A \land (\mathcal{C}_{\downarrow d})_{\downarrow f} \models \bigvee_{i \in I} \left( \mathcal{C}_i^\prime \land \mathcal{B}_i \right)$.

\begin{enumerate}
\item \emph{Answering steps $b;d \cstep{f}{\mathcal{B}_i} b_i^\prime$\@:}
\Cref{lemma16-dbccrs} and $b \cstep{\hat f}{\hat{\mathcal{B}_i}} \hat{b_i^\prime}$
imply that steps $b \cstep{\hat f;\hat g}{{\hat{\mathcal{B}_i}}_{\downarrow \hat g}} \hat{b_i^\prime};\hat g$ are possible.
Since $d;f = \hat f;\hat g$, these are equivalent to $b \cstep{d;f}{{\hat{\mathcal{B}_i}}_{\downarrow \hat g}} \hat{b_i^\prime};\hat g$.
We set $b_i^\prime \defeq \hat{b_i^\prime};\hat g,\ \mathcal{B}_i \defeq {\hat{\mathcal{B}_i}}_{\downarrow \hat g}$
and get steps $b \cstep{d;f}{\mathcal{B}_i} b_i^\prime$,
which, according to \Cref{lemma-ctxtrans-composition}, are equivalent to steps $b;d \cstep{f}{\mathcal{B}_i} b_i^\prime$.

\item \emph{$(a^\prime, b_i^\prime, \mathcal{C}_i^\prime) \in u(R)$\@:}
Set $\mathcal{C}_i^\prime \defeq \hat{\mathcal{C}_i^\prime}_{\downarrow \hat g}$.
By construction of $u(R)$, $(r;\hat c,\ \hat{b_i^\prime},\ \hat{\mathcal{C}_i^\prime}) \in R$ implies $(r;\hat c;\hat g,\ \hat{b_i^\prime};\hat g,\ \hat{\mathcal{C}_i^\prime}_{\downarrow \hat g}) = (a^\prime, b_i^\prime, \mathcal{C}_i^\prime) \in u(R)$.

\item \emph{$\mathcal{A} \land (\mathcal{C}_{\downarrow d})_{\downarrow f} \models \bigvee_{i \in I} \left( \mathcal{C}_i^\prime \land \mathcal{B}_i \right)$\@:}
Above, we already showed $\mathcal{R}_{\downarrow \hat c} \land \mathcal{C}_{\downarrow \hat f} \models \bigvee_{i \in I} \left( \hat{\mathcal{C}_i^\prime} \land \hat{\mathcal{B}_i} \right)$.
By shifting both sides with $\hat g$ and applying the rules of \Cref{shift-laws}, we get:
\begin{align*}
\mathcal{R}_{\downarrow \hat c} \land \mathcal{C}_{\downarrow \hat f}
&\models \bigvee\nolimits_{i \in I} \left( \hat{\mathcal{C}_i} \land \hat{\mathcal{B}_i} \right)
\\ \Rightarrow
\left( \mathcal{R}_{\downarrow \hat c} \land \mathcal{C}_{\downarrow \hat f} \right)_{\downarrow \hat g}
&\models \left( \bigvee\nolimits_{i \in I} \big( \hat{\mathcal{C}_i^\prime} \land \hat{\mathcal{B}_i} \big) \right)_{\downarrow \hat g}
\\ \Leftrightarrow
\mathcal{R}_{\downarrow \hat c;\hat g} \land \mathcal{C}_{\downarrow \hat f;\hat g}
&\models \bigvee\nolimits_{i \in I} \left( {\hat{\mathcal{C}_i^\prime}}_{\downarrow \hat g} \land {\hat{\mathcal{B}_i}}_{\downarrow \hat g} \right)
\end{align*}

By substituting $\hat c;\hat g = c,\   \hat f;\hat g = d;f,\   \mathcal{C}_i^\prime = {\hat{\mathcal{C}_i^\prime}}_{\downarrow \hat g},\   \mathcal{B}_i = {\hat{\mathcal{B}_i}}_{\downarrow \hat g}$, we obtain
$\mathcal{R}_{\downarrow c} \land \mathcal{C}_{\downarrow d;f}
\models \bigvee_{i \in I} \left( \mathcal{C}_i^\prime \land \mathcal{B}_i \right)$.
Since $\mathcal{A} \models \mathcal{R}_{\downarrow c}$, we have $\mathcal{A} \land \mathcal{C}_{\downarrow d;f} \models \mathcal{R}_{\downarrow c} \land \mathcal{C}_{\downarrow d;f} \models \bigvee_{i \in I} \left( \mathcal{C}_i^\prime \land \mathcal{B}_i \right)$, which was to be shown.

Analogously, answering steps for $b;d \rstep{f}{\mathcal B} b'$ can be constructed.
\qedhere
\end{enumerate}
\end{proof}

\noindent
The following theorem is based on a proof strategy similar to
\Cref{lemma-cbr-cuc}.

\begin{thm}\label{simr-simc}%
Conditional bisimilarity and representative conditional bisimilarity \coincide, that is, $\csimC = \csimR$.
\end{thm}

\begin{proof}
\begin{proofparts}
% c subseteq r
\proofPart{$\csimC \subseteq \csimR$}
  Consider a triple $(a, b, \mathcal C) \in \csimC$.
  By \Cref{def-cb-c}, for each step $a \cstep{f}{\mathcal A} a^\prime$
  there exist answering steps of $b$ with the requirements listed there;
  symmetrically, each step $b \cstep{f}{\mathcal B} b^\prime$ can be answered by $a$.
  To show that $(a, b, \mathcal C) \in \csimR$, the same statement has to be shown for each representative step $a \rstep{f}{\mathcal{A}} a^\prime$
  (and $b \rstep{f}{\mathcal{B}} b^\prime$).
  Since it already holds for all context steps, which are a superset of representative steps (every $\rightarrow_R$ step is also a $\rightarrow_C$ step), \Cref{def-cb-r} is trivially satisfied.

% r subseteq c
\proofPart{$\csimC \supseteq \csimR$}
  By definition of $\csimC$, if a conditional relation $R$ is a conditional bisimulation, then $R \subseteq \csimC$.
  We show that $R = \csimR$ is a conditional bisimulation, i.e.\ that it satisfies the requirements of \Cref{def-cb-c}.

  \def\Relname{R}
  Consider a triple $(a, b, \mathcal C) \in \Relname$ and a context step $a \cstep{f}{\mathcal A} a'$.
  This step is due to some rule $(\ell, r, \mathcal R)$.
  According to \Cref{cond-repr-zurueckfuehren}, this context step can be reduced to a representative step
  $a \rstep{\hat f}{\mathcal{R}_{\downarrow \hat c}} r;\hat c$,
  and there exists $\hat g$ such that $\hat f ; \hat g = f,\ \hat c ; \hat g = c$.
  Again, all variables with a hat (e.g.\ $\hat c$) refer to the representative step,
  but otherwise take the same role as their unhatted counterparts.

  Since $\Relname$ is the representative conditional bisimilarity,
  $\Relname$ is a representative conditional bisimulation.
  Together with $(a, b, \mathcal C) \in \Relname$, this means that
  for the aforementioned step there exist answering steps $b \cstep{\hat f}{\hat{\mathcal{B}_i}} \hat{b_i^\prime}$ and conditions $\hat{\mathcal{C}_i^\prime}$,
  such that $(r;\hat c,\ \hat{b_i^\prime},\ \hat{\mathcal{C}_i^\prime}) \in \Relname$ and
  $\mathcal{R}_{\downarrow \hat c} \land \mathcal{C}_{\downarrow \hat f}
  \models \bigvee_{i \in I} \big( \hat{\mathcal{C}_i^\prime} \land \hat{\mathcal{B}_i} \big)$.

  % todo: deduplication of this proof and the previous one
  However we are not interested in answering steps for the representative step, but rather for the original step $a \cstep{f}{\mathcal A} a'$,
  that is, we need answering steps of $b$ using context $f$.
  So we need
  (1) answering steps $b \cstep{f}{\mathcal{B}_i} b_i^\prime$
  (2) such that $(a^\prime, b_i^\prime, \mathcal{C}_i^\prime) \in \Relname$
  and (3) $\mathcal{A} \land \mathcal{C}_{\downarrow f} \models \bigvee_{i \in I} \left( \mathcal{C}_i^\prime \land \mathcal{B}_i \right)$.
  Analogously to the proof of \Cref{lemma-cbr-cuc}, we have:

  \begin{enumerate}
    \item \emph{Answering steps\@:}
      Having steps $b \cstep{\hat f}{\hat{\mathcal B_i}} \hat{b_i^\prime}$
      implies that steps \mbox{$b \cstep{\hat f;\hat g}{{\hat{\mathcal B_i}}_{\downarrow \hat g}} \hat{b_i^\prime};\hat g$} are possible (\Cref{lemma16-dbccrs}).
      Rewritten as $b \cstep{f}{\mathcal B_i} b_i^\prime$ (where $\mathcal B_i \defeq {\hat{\mathcal B_i}}_{\downarrow \hat g},\ b_i^\prime \defeq \hat{b_i^\prime};\hat g$), we get the desired answering steps for the original step.
    \item \emph{$(a', b_i', \mathcal C_i') \in R$\@:}
      Since $\Relname$ is the representative conditional bisimilarity, by \Cref{lemma-cbr-cuc} we know that $\Relname$ is closed under contextualization.
      Therefore,
      $(r;\hat c,\ \hat{b_i^\prime},\ \hat{\mathcal{C}_i^\prime}) \in \Relname$ implies
      $(r;\hat c;\hat g,\ \hat{b_i^\prime};\hat g,\ \hat{\mathcal C_i^\prime}_{\downarrow \hat g}) = (a^\prime, b_i^\prime, \mathcal C_i^\prime) \in \Relname$,
      i.e.\ the original target $a^\prime$ is conditionally bisimilar to the targets $\hat{b_i^\prime};\hat{g}$ of the answering steps.
    \item \emph{$\mathcal{A} \land \mathcal{C}_{\downarrow f} \models \bigvee_{i \in I} \left( \mathcal{C}_i^\prime \land \mathcal{B}_i \right)$\@:}
      Using the rules of \Cref{shift-laws} we get:
      \begin{align*}
        \mathcal R_{\downarrow \hat c} \land \mathcal C_{\downarrow \hat f}
        & \models \bigvee \left( \hat{\mathcal C_i^\prime} \land \hat{\mathcal B_i} \right)
        \\ \Rightarrow
        (\mathcal R_{\downarrow \hat c})_{\downarrow \hat g} \land (\mathcal{C}_{\downarrow \hat{f}})_{\downarrow \hat g}
        & \models \bigvee \left( \hat{\mathcal C_i^\prime}_{\downarrow \hat g} \land \hat{\mathcal B_i}_{\downarrow \hat g} \right)
        \\ \Leftrightarrow
        \mathcal R_{\downarrow \hat c ;\hat g} \land \mathcal{C}_{\downarrow \hat{f};\hat{g}}
        & \models \bigvee \left( \hat{\mathcal C_i^\prime}_{\downarrow \hat g} \land \hat{\mathcal B_i}_{\downarrow \hat g} \right)
        \\ \Leftrightarrow
        \mathcal R_{\downarrow c} \land \mathcal C_{\downarrow f}
        & \models \bigvee \left( \mathcal C_i^\prime \land \mathcal{B}_i \right)
      \end{align*}

      Since $\mathcal A \models \mathcal R_{\downarrow c}$,
      we have $
        \mathcal A \land \mathcal C_{\downarrow f}
        \models \mathcal R_{\downarrow c} \land \mathcal C_{\downarrow f}
        \models \bigvee \left( \mathcal{C}_i^\prime \land \mathcal{B}_i \right)
      $, which is the required condition for the triple $(a, b, \mathcal C)$.
  \end{enumerate}
  Analogously, we can construct answering steps for $b \cstep{f}{\mathcal B} b'$.
  We have therefore shown that $\Relname = \csimR$ is a conditional bisimulation and therefore $\csimR \subseteq \csimC$.
  \qedhere
\end{proofparts}
\end{proof}

%\todohl{\textbf{L} habs mal versucht ... ist vermutlich zu lang.
%Unterschiede:
%
%bei \Cref{lemma-cbr-cuc} baut man am Anfang erst in Schritte von $a$ um und erhält dadurch einen Kontextschritt $a \cstep{d;f}{\mathcal A} a'$, während man bei \Cref{simr-simc} direkt mit $a \cstep{f}{\mathcal A} a'$ anfängt.
%
%In \Cref{lemma-cbr-cuc} zeigt man nur dass die Nachfolger in $u(R)$ sind, während sie in \Cref{simr-simc} in $R$ sein sollen (tatsächlich werden die Nachfolger dort aber trotzdem aus $u(R)$ gezogen, denn dort wird benutzt, dass für diese Relation $R$ gilt: $u(R) \subseteq R$ gilt).
%
%Man könnte evtl. sagen: Wir starten mit einem $a \cstep{f}{\mathcal A} a'$ und verfahren dann so wie in \Cref{lemma-cbr-cuc}, nur dass wir $d=\id$ annehmen (also $f$ statt $d;f$ benutzen und folglich auch \Cref{lemma-ctxtrans-composition} nie brauchen).
%Im vorletzten Teil des Beweises ($(a',b_i',\mathcal C_i') \in R$) zeigt \Cref{lemma-cbr-cuc} nur $\in u(R)$ aber wir brauchen jetzt $\in R$. aus genau dem Lemma wissen wir aber dass $u(R) \subseteq R$, darum ist das ok.
%}{%
%}%
%
\noindent
We now discuss the notion of representative condition bisimulation in
two examples.

\begin{exa}[Message passing over unreliable channels, continued]\label{ex-cb-repr-unreliable}%
  Consider the reactive system of \Cref{ex-cb-unreliable}.  There
  exists a representative conditional bisimulation $R$ such that
  %$(\emptyset \rightarrow R_0 \leftarrow I_0,\ \emptyset \rightarrow
  %U_0 \leftarrow I_0,\ \AIneN) \in R$
  $(\emptyset \rightarrow \inlinechan{reliable} \leftarrow \inlinechan{},\ %
  \emptyset \rightarrow \inlinechan{unreliable} \leftarrow \inlinechan{},\ \AIneN) \in R$
  (where $\AIneN$ requires that no $\noiseLabel$-edge exists).

  We consider the representative steps that are possible from either
  $\inlinechan{reliable}$ or $\inlinechan{unreliable}$ and only explain the most interesting cases
  (cf.\ \Cref{fig-infinite-cb-inf}).

\begin{itemize}
\item The graph $\inlinechan{reliable}$ can do a step using rule $P_R$ by borrowing a message on
  the left node, that is, $f = \inlinechan{} \rightarrow \inlinechan{msgleft} \leftarrow \inlinechan{}$,
  reacting to $\inlinechan{reliable,msgright}$.  No further restrictions on the environment are
  necessary, so $\mathcal A = \condtrue$.  The graph $\inlinechan{unreliable}$ can answer this step
  using $P_U$ and reacts to $\inlinechan{unreliable,msgright}$, but only if no noise is present
  (environment satisfies $\mathcal B_i = \AIneN$).  We evaluate
  the implication
  $\mathcal A \land \mathcal C_{\downarrow f}
  \equiv {\condtrue} \land {\AIneN}_{\downarrow f}
  \equiv \AIneN
  \models
  \bigvee_{i \in I} \left( \mathcal{C}_i^\prime \land \AIneN \right)
  \equiv \bigvee_{i \in I} \left( \mathcal{C}_i^\prime \land \mathcal{B}_i \right)$,
  setting
  $\mathcal{C}_i^\prime = \AIneN$. (Note that
  ${\AIneN}_{\downarrow f} \equiv \AIneN$ since $\AIneN$ forbids
  the existence of an $\noiseLabel$-edge between the two interface nodes and $f$
  is unrelated, providing an $m$-loop on the left-hand node.)  We now
  require
  $(\emptyset \rightarrow \inlinechan{reliable,msgright} \leftarrow \inlinechan{},\ %
  \emptyset \rightarrow \inlinechan{unreliable,msgright} \leftarrow \inlinechan{},\ \AIneN) \in R$.

\item Symmetrically, $\inlinechan{unreliable}$ can do a step using $P_U$ by borrowing a message on the left node, reacting to $\inlinechan{unreliable,msgright}$ in an environment without noise ($\mathcal A = \AIneN$).
$\inlinechan{reliable}$ can answer this step under any condition $\mathcal B_i$. Then,
the implication is satisfied if we set $\mathcal C_i^\prime = \AIneN$,
so we require again $(\emptyset \rightarrow \inlinechan{reliable,msgright} \leftarrow \inlinechan{},\ %
\emptyset \rightarrow \inlinechan{unreliable,msgright} \leftarrow \inlinechan{},\ \AIneN) \in R$.

\item There are additional representative steps that differ in how
  much of the left-hand side is borrowed, but can be proven
  analogously to the two previously discussed steps.
\end{itemize}

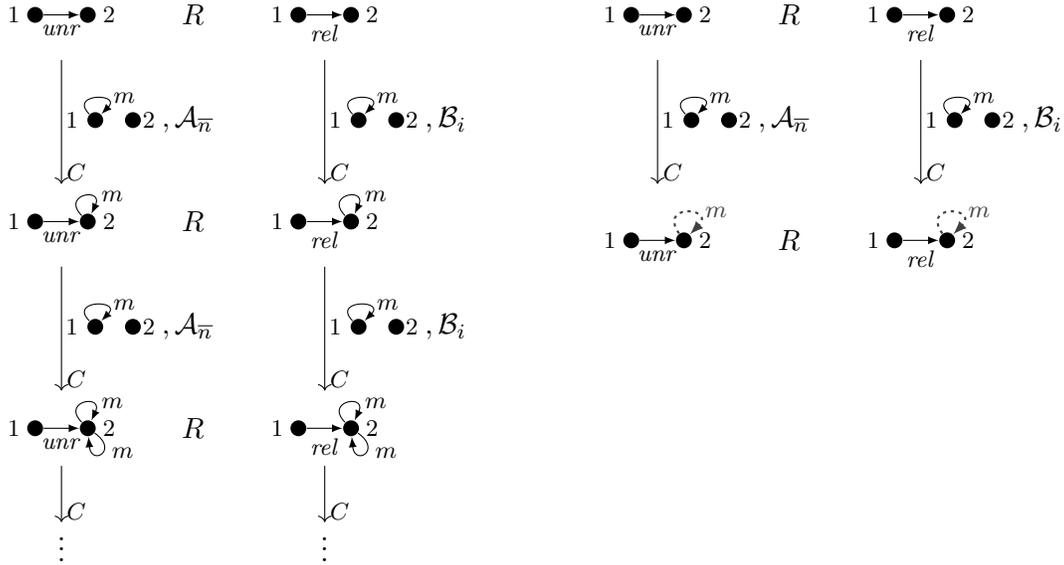
\begin{figure}[ht]
  \begin{subfigure}[b]{0.48\textwidth}
    \centering\begin{tikzpicture}[gedge/.append style={font=\footnotesize}]%[scale=0.8, every node/.style={transform shape}]
      % initial pair
      \foreach \xpos/\chanlabel in {0/\unrelLabel,3.5/\relLabel} {
        \begin{scope}[shift={(\xpos,0)}]
          \node[gn] (u1) at (-0.35,0) {};
          \node[gn] (u2) at (0.35,0) {};
          \node at (-0.35cm-8pt,0pt) {\footnotesize$1$};
          \node at (0.35cm+8pt,0pt) {\footnotesize$2$};
          \draw[gedge] (u1) -- node[below]{$\chanlabel$} (u2);
        \end{scope}
      }
      \node at (1.75,0) {$R$};

      % first reached result
      \foreach \xpos/\chanlabel in {0/\unrelLabel,3.5/\relLabel} {
        \begin{scope}[shift={(\xpos,-2.75)}]
          \node[gn] (u1) at (-0.35,0) {};
          \node[gn] (u2) at (0.35,0) {};
          \node at (-0.35cm-8pt,0pt) {\footnotesize$1$};
          \node at (0.35cm+8pt,0pt) {\footnotesize$2$};
          \draw[gedge] (u1) -- node[below]{$\chanlabel$} (u2);
          \draw[gedge] (u2) to[loop,out=45+90,in=45+20,distance=0.4cm] node[right]{\ $m$} (u2);
        \end{scope}
      }
      \node at (1.75,-2.75) {$R$};

      % second reached result
      \foreach \xpos/\chanlabel in {0/\unrelLabel,3.5/\relLabel} {
        \begin{scope}[shift={(\xpos,-5.5)}]
          \node[gn] (u1) at (-0.35,0) {};
          \node[gn] (u2) at (0.35,0) {};
          \node at (-0.35cm-8pt,0pt) {\footnotesize$1$};
          \node at (0.35cm+8pt,0pt) {\footnotesize$2$};
          \draw[gedge] (u1) -- node[below]{$\chanlabel$} (u2);
          \draw[gedge] (u2) to[loop,out=45+90,in=45+20,distance=0.4cm] node[right]{\ $m$} (u2);
          \draw[gedge] (u2) to[loop,out=45+90+180+10,in=45+180+50,distance=0.4cm] node[right]{$m$} (u2);
        \end{scope}
      }
      \node at (1.75,-5.5) {$R$};

      \foreach \ypos in {0,-2.75} {
        \draw[->] (0,\ypos-0.6) to (0,\ypos-2.25);
        \node at (0.2,\ypos-2.1) {\footnotesize$C$};
        \draw[->] (3.5,\ypos-0.6) to (3.5,\ypos-2.25);
        \node at (3.5+0.2,\ypos-2.1) {\footnotesize$C$};
      }

      % borrow L/R step
      \foreach \ypos in {0,-2.75} {
      \foreach \xpos/\condlabel in {0.7/\AIneN,4.2/\mathcal B_i} {
        \begin{scope}[shift={(\xpos,\ypos-1.4)}]
          \node[gn] (n1) at (-0.25,0) {};
          \node[gn] (n2) at (0.25,0) {};
          \node at (-0.35cm-6pt,0pt) {\footnotesize$1$};
          \node at (0.25cm+6pt,0pt) {\footnotesize$2$};
          \draw[gedge] (n1) to[loop,out=45+90,in=45,distance=0.4cm] node[right]{\ $m$} (n1);
          \node[anchor=west] at (0.5cm,0pt) {$, \condlabel$};
        \end{scope}
      }
      }

      % react dot dot dot
        \draw[->] (0,-5.5-0.5) to (0,-5.5-1.25);
        \draw[->] (3.5,-5.5-0.5) to (3.5,-5.5-1.25);
        \node at (0.2,-5.5-1.1) {\footnotesize$C$};
        \node at (3.5+0.2,-5.5-1.1) {\footnotesize$C$};
        \node at (0,-5.5-1.5) {$\vdots$};
        \node at (3.5,-5.5-1.5) {$\vdots$};
    \end{tikzpicture}%
    \caption{Showing conditional bisimilarity may require an infinite number of steps\\~}%
    \label{fig-infinite-cb-inf}%
  \end{subfigure}\hfill%
  \begin{subfigure}[b]{0.48\textwidth}
    \centering\begin{tikzpicture}[gedge/.append style={font=\footnotesize}]%[scale=0.8, every node/.style={transform shape}]
      \node at (3.5,-5.5-1.5) {$\phantom\vdots$};
      % initial pair
      \foreach \xpos/\chanlabel in {0/\unrelLabel,3.5/\relLabel} {
        \begin{scope}[shift={(\xpos,0)}]
          \node[gn] (u1) at (-0.35,0) {};
          \node[gn] (u2) at (0.35,0) {};
          \node at (-0.35cm-8pt,0pt) {\footnotesize$1$};
          \node at (0.35cm+8pt,0pt) {\footnotesize$2$};
          \draw[gedge] (u1) -- node[below]{$\chanlabel$} (u2);
        \end{scope}
      }
      \node at (1.75,0) {$R$};

      % first reached result
      \foreach \xpos/\chanlabel in {0/\unrelLabel,3.5/\relLabel} {
        \begin{scope}[shift={(\xpos,-3)}]
          \node[gn] (u1) at (-0.35,0) {};
          \node[gn] (u2) at (0.35,0) {};
          \node at (-0.35cm-8pt,0pt) {\footnotesize$1$};
          \node at (0.35cm+8pt,0pt) {\footnotesize$2$};
          \draw[gedge] (u1) -- node[below]{$\chanlabel$} (u2);
          \draw[gedge,thick,dash pattern=on 1.25pt off 1.75pt,black!75] (u2) to[loop,out=35+90,in=45,distance=0.5cm] node[right]{$\color{black!75}\ m$} (u2);
        \end{scope}
      }
      \node at (1.75,-3) {$R$};

        \draw[->] (0,-0.6) to (0,-2.25);
        \node at (0.2,-2.1) {\footnotesize$C$};
        \draw[->] (3.5,-0.6) to (3.5,-2.25);
        \node at (3.5+0.2,-2.1) {\footnotesize$C$};

      % borrow L/R step
      \foreach \xpos/\condlabel in {0.7/\AIneN,4.2/\mathcal B_i} {
        \begin{scope}[shift={(\xpos,-1.4)}]
          \node[gn] (n1) at (-0.25,0) {};
          \node[gn] (n2) at (0.25,0) {};
          \node at (-0.35cm-6pt,0pt) {\footnotesize$1$};
          \node at (0.25cm+6pt,0pt) {\footnotesize$2$};
          \draw[gedge] (n1) to[loop,out=45+90,in=45,distance=0.4cm] node[right]{\ $m$} (n1);
          \node[anchor=west] at (0.5cm,0pt) {$, \condlabel$};
        \end{scope}
      }
    \end{tikzpicture}%
    \caption{Removing unrelated context (dotted edge) before relating the pair would allow stopping after a single step}%
    \label{fig-infinite-cb-upto}%
  \end{subfigure}
  \caption{Showing conditional bisimilarity of two types of channels using the rules \mbox{of \Cref{ex-cb-unreliable}}}%
  \label{fig-infinite-cb}
\end{figure}

\noindent
This means we have to add the pair
\mbox{$(\emptyset \rightarrow \inlinechan{reliable,msgright} \leftarrow \inlinechan{},\ %
\emptyset \rightarrow \inlinechan{unreliable,msgright} \leftarrow \inlinechan{},\ \AIneN)$}
to $R$ and to continue adding pairs
until we obtain a bisimulation: with every step, a new triple with an
additional $m$-loop on the right node is added to the relation,
therefore, the smallest conditional bisimulation has infinite size.
This is visualized in \Cref{fig-infinite-cb-inf}.

However, except for the additional $m$-loop on the right node, which
does not affect rule application, this pair is identical to the
initial one and we can hence use a similar argument.  In
\Cref{sec:uptocond} we show how to make this formal, using up-to
technique.  In summary, we conclude that $\inlinechan{reliable}$ is
conditionally bisimilar to $\inlinechan{unreliable}$ under the
condition~$\AIneN$.
\end{exa}

\begin{exa}[Unreliable channel vs.\ no channel]\label{ex-cb-unreliable-empty}%
  For \Cref{ex-cb-unreliable,ex-cb-repr-unreliable}, it can also be
  shown that under the condition $\neg\AIneN$, the unreliable
  channel $\emptyset \rightarrow \inlinechan{unreliable} \leftarrow \inlinechan{}$ is conditionally bisimilar
  to not having a channel between the two nodes ($\emptyset \rightarrow \inlinechan{} \leftarrow \inlinechan{}$).

In this case, $\inlinechan{unreliable}$ can still do a reaction under $\AIneN$.
Then, $\inlinechan{}$ can answer with an empty set of steps. The implication $\AIneN \land \mathcal C_{\downarrow f} \models \bigvee_{i \in I} \left( \mathcal C_i^\prime \land \mathcal B_i \right)$
is then simplified to $\AIneN \land \neg \AIneN \models
\condfalse$, which is easily seen to be valid.
\end{exa}

\section{Up-to Techniques for Proving Conditional Bisimilarity}%
\label{sec:uptocond}

Our optimizations so far involved replacing context steps by
representative steps, which ensure finite branching and thus greatly
reduce the proof obligations for a single step. However, it can still
happen very easily that the smallest possible bisimulation is of
infinite size, in which case automated proving of conditional
bisimilarity becomes impossible. For instance, in
\Cref{ex-cb-repr-unreliable}, the least conditional bisimulation
relating the two cospans $u,r$ (representing (un)reliable channels) % chktex 36
under $\AIneN$
contains infinitely many triples $(u;m^n,\, r;m^n,\, \AIneN)$
for any number~$n$ of messages on the right node
($m = \inlinechan{} \rightarrow \inlinechan{msgright} \leftarrow \inlinechan{}$).

On the other hand, conditional bisimilarity is closed under contextualization,
hence if $u,r$ are related, we can conclude that $u;m$ and $r;m$ must be
related as well. Intuitively the relation $R = \{(u,r,\AIneN)\}$
is a sufficient witness, since after one step we reach the triple
$(u;m,\, r;m,\, \AIneN)$, from which we can ``peel off'' a common context $m$ to
obtain a triple already contained in $R$
(visualized in \Cref{fig-infinite-cb-upto}).

% Proving bisimilarity often requires checking a possibly infinite number of contexts, as for saturated bisimilarity based on context steps.
% To reduce the number of contexts to be checked, representative steps are used to show semi-saturated bisimilarity,
% which is equivalent to saturated bisimilarity.
% In the previous section, we applied the same approach to conditional bisimilarity
% and showed that this yields an equivalent definition.

% However, some systems are hard to analyze not because of infinitely many contexts, but because the bisimulation relation itself might be of infinite size.
% As a simple example, consider a rule which adds a component to the system, which can be repeatedly applied, giving the infinite reaction sequence $a \leadsto a;c \leadsto a;c;c \leadsto \dots$.
% Note that the additional components $c$ are not part of the borrowed context $f$, but of $a$, and can therefore not be simplified by using representative steps (which only reduce the number of contexts $f$).
% Hence, each of these steps requires finite time to analyze, but the system as a whole has infinitely many steps and therefore still requires infinite time to analyze.
% A similar situation was present in \Cref{ex-cb-repr-unreliable}.

This is an instance of an \emph{up-to technique}, which can be used to
obtain smaller witness relations by identifying and removing redundant
elements from a bisimulation relation.  Instead of requiring the
redundant triple $(u;m,\, r;m,\, \AIneN)$ to be contained in the relation, it is
sufficient to say that \emph{up to} the passive context $m$, the triple is
represented by $(u,r,\AIneN)$, which is already contained in the
relation.  In particular, this specific up-to technique is known as
\emph{up-to context}~\cite{ps:enhancements-coinductive}, a well-known proof
technique for process calculi.%
\footnote{Stated in the language of
  process algebra, a symmetric relation $R$ is a bisimulation up-to context,
  if whenever $(P,Q)\in R$ and $P \xrightarrow{a} P'$,
  then $Q \xrightarrow{a} Q'$, where $P' = C[P'']$, $Q' = C[Q'']$ and
  $(P'',Q'')\in R$.}

Note that in general, a bisimulation up-to context is not a
bisimulation relation.  However, it can be converted into a
bisimulation by closing it under all contexts.

In this section, we show how to adapt this concept to conditional
bisimilarity and~in particular discuss how to deal with the conditions
in a conditional bisimulation \mbox{up-to context}.

\subsection{Up-To Techniques and Fixpoint Theory}
As in \Cref{sec:condbisim}, we will provide definitions and proofs
that are based on fixpoint theory. Hence, we first introduce the
remaining preliminary concepts for implementing up-to techniques using
fixpoint theory, again mostly following~\cite{ps:enhancements-coinductive}.

In \Cref{subsec-fixpoint-basics} we already explained that to show that
some element $l$ of the lattice is contained in behavioural equivalence
($l \sqsubseteq \nu f$), it is sufficient to prove that $l$ is under some
post-fixpoint $l'$ ($l \sqsubseteq l' \sqsubseteq f(l')$).
  The idea of using up-to techniques is now to define a monotone
  function $u$ (the up-to function) and check if
  $l \sqsubseteq \nu(f \circ u)$ by showing that $l$ is a
  post-fixpoint of $f\circ u$.  Typically, the characteristics
  of $u$ should make it easier to prove $l \sqsubseteq f(u(l))$
  than proving $l \sqsubseteq f(l)$.  This is clearly the case when
  $u$ is extensive, since extensiveness of $u$ and monotonicity of $f$
  implies $f(l) \sqsubseteq f(u(l))$ and thus obtaining
  $l \sqsubseteq f(u(l))$ is easier than obtaining
  $l \sqsubseteq f(l)$.

Naturally, for the up-to technique to be useful, it also has to be shown that $\nu(f \circ u) \sqsubseteq \nu f$:

% The general idea of up-to techniques is as follows. Given a monotone
% function $f \colon L \to L$ one is interested in the greatest fixpoint
% $\nu f$.
% % (using an appropriate function $f$ such that $\nu f$ equals the behavioural equivalence).
% In general, the aim is to establish whether some given
% element of the lattice $l \in L$ is under the fixpoint, i.e., if
% $l \sqsubseteq \nu f$
% % (if the elements of the relation $l$ are behaviourally equivalent).
% In turn, since by Tarski's Theorem,
% $\nu f = \bigsqcup \{ x \mid x \sqsubseteq f(x) \}$, this amounts to
% proving that $l$ is under some post-fixpoint $l'$, i.e.,
% $l \sqsubseteq l'\sqsubseteq f(l')$.

% For instance, consider the function
% $\mathit{bis}_T \colon \Rel{\mathbb{S}} \to \Rel{\mathbb{S}}$ (We
% would use $f_C \colon \Condrel \to \Condrel$ here) for bisimilarity on
% a transition system $T$ in Example x.xx [it defined functions
% $sim_T, bis_T \colon \Rel \to \Rel$] .  Given two states
% $s_1, s_2 \in \mathbb{S}$, proving
% $\{ (s_1, s_2) \} \subseteq \nu \mathit{bis}_T$, i.e., showing the two
% states bisimilar, amounts to finding a post-fixpoint, i.e., a relation
% $R$ such that $R \subseteq \mathit{bis}_T(R)$ (namely, a bisimulation)
% such that $\{ (s_1, s_2) \} \subseteq R$.

\begin{defi}[sound up-to function]% Def 16 BKP20b
  \label{def-soundness}
  Let $L$ be a complete lattice and let $f \colon L \to L$ be a monotone
  function. A \emph{sound up-to function} for $f$ is any monotone
  function $u \colon L \to L$ such that
  $\nu (f \circ u) \sqsubseteq \nu f$.
  % It is called \emph{complete} if also the converse inequality
  % $\nu f \sqsubseteq \nu (f \circ u)$ holds.
\end{defi}

% When $u$ is sound, if $l$ is a post-fixpoint of $f \circ u$ (i.e., $l \sqsubseteq f(u(l))$),
% we have
% $l \sqsubseteq \nu (f \circ u) \sqsubseteq \nu f$.
% The idea is that the characteristics of $u$ should make it easier to
% prove that $l$ is a post-fixpoint of $f \circ u$ than proving that it
% is for $f$.

% Note that extensiveness also implies
% ``completeness'' of the up-to function: since $f \sqsubseteq f \circ u$
% clearly $\nu f \sqsubseteq \nu (f \circ u)$.
% We remark that for up-to functions, since the interest is for underapproximating fixpoints, the terms soundness and completeness are somehow reversed with respect to their meaning in abstract interpretation.

Instead of soundness, we will use the stronger notion of compatibility:

% A common sufficient condition ensuring soundness of up-to
% functions is compatibility~\cite{p:complete-lattices-up-to}.

\begin{defi}[$f$-compatibility of $u$]% BKP20b, Def 17
  \label{def:compatibility}
  Let $L$ be a complete lattice and let $f \colon L \to L$ be a monotone
  function. A monotone function $u \colon L \to L$ is \emph{$f$-compatible} if
  $u \circ f \sqsubseteq f \circ u$.
\end{defi}

  Compatibility has several advantages over soundness: First,
  compatibility implies soundness, second, $f$-compatible up-to
  techniques can be combined with each other to obtain more powerful
  proof techniques, and third, it implies that the up-to function
  preserves the greatest fixpoint (a kind of congruence
  result):

% The soundness of
% an $f$-compatible up-to function $u$ can be proved by viewing it as an
% abstraction. When $u$ is a closure (i.e., extensive and idempotent),
% $u(L)$ is a complete lattice.
%
% {\cite[Lemma 18]{BKP20b}}]

\begin{lem}[Compatibility implies soundness]%
  \label{le:up-to-closure}
  Let $f \colon L \to L$ be a monotone function and let $u \colon L \to L$ be an
  $f$-compatible closure.
  Then $\nu f = \nu (f\circ u)$.
\end{lem}

\begin{proof}
  $\nu (f \circ u) \sqsubseteq \nu f$ by~\cite[Theorem 6.3.9]{pous_sangiorgi_2011}.
  $\nu f \sqsubseteq \nu (f \circ u)$ because $u$ is a closure and therefore extensive.
\end{proof}

Note that the first inclusion is sufficient to show soundness, but we will later also use the second one to simplify some of the proofs.

This lemma gives rise to the following proof rule:
to show that $l$ is in the behavioural equivalence ($l \sqsubseteq \nu f$),
it is sufficient to show that $l$ is a post-fixpoint of $f \circ u$
(i.e., \mbox{$l \sqsubseteq f(u(l))$}), which implies
$l \sqsubseteq \nu(f \circ u) = \nu f$.

\begin{lem}[Compositionality of $f$-compatible functions {\cite[Proposition 1.6]{p:complete-lattices-up-to}}]
  If $u_1, u_2$ are $f$-compatible, then $u_1 \circ u_2$ is also $f$-compatible.
\end{lem}

\begin{prop}%
  \label{compat-preserves-fp}
  Let $f \colon L \to L$ be a monotone function and
  let $u \colon L \to L$ be an $f$-compatible closure.
  Then $u(\nu f) = \nu f$.
\end{prop}

\begin{proof}
  $\nu f \sqsubseteq u(\nu f)$ holds because $u$ is a closure and
  therefore extensive.

  Since $u$ is $f$-compatible,~\cite[Remark
  1.5]{p:complete-lattices-up-to} guarantees
  $u(\nu f) \sqsubseteq \nu f$.
%
  % $\nu f$ is a fixpoint, therefore $u(\nu f) = u(f(\nu f))$,
  % and because of $f$-compatibility, $u(f(\nu f)) \sqsubseteq f(u(\nu f))$.
  % In total, $u(\nu f) \sqsubseteq f(u(\nu f))$,
  % i.e.\ $u(\nu f)$ is a post-fixpoint of $f$.
  % By Tarski's Theorem, the latter implies $u(\nu f) \sqsubseteq \nu f$.
%
%\begin{proofparts}
%\proofPart{$u(\nu f) \subseteq \nu f$}
\end{proof}

% Given compatibility, we can also give an alternative characterization
% of certain types of up-to relations: $l$ being a post-fixpoint of
% $f \circ u$ is equivalent to the closure of $l$ under $u$ being a
% post-fixpoint of $f$.  The former typically yields a criterion that is
% easier to check for a given $l$, while the latter more directly shows
% that $l$ indeed has the properties given by $f$ (when contextualized
% under $u$).

We will also need the following (straightforward) result.

\begin{lem}[Characterization of post-fixpoints of $f\circ u$]\label{lem-pride}%
  Let $f \colon L \to L$ be a monotone function,
  let $u \colon L \to L$ be an $f$-compatible closure,
  and let $l \in L$.
  Then, $l \sqsubseteq f(u(l)) \iff u(l) \sqsubseteq f(u(l))$.
\end{lem}

\begin{proof}
\begin{proofparts}
% <= <= <=
\proofPart{$l \sqsubseteq f(u(l)) \Rightarrow u(l) \sqsubseteq f(u(l))$}
  Since $u$ is monotone, $l \sqsubseteq f(u(l))$ implies $u(l) \sqsubseteq u(f(u(l)))$.
  By $f$-compatibility of $u$ ($u \circ f \sqsubseteq f \circ u$),
  then also $u(l) \sqsubseteq f(u(u(l)))$,
  which by idempotence of $u$ is equivalent to $u(l) \sqsubseteq f(u(l))$.

%  Therefore, every relation that our initial definition recognizes as a CBUC
%  indeed represents a conditional bisimulation (when closed under contextualization).

% => => =>
\proofPart{$l \sqsubseteq f(u(l)) \Leftarrow u(l) \sqsubseteq f(u(l))$}
  $u$ is extensive, therefore $l \sqsubseteq u(l)$.
  Combined, we obtain $l \sqsubseteq u(l) \sqsubseteq f(u(l))$
  and hence our desired result.
%
%  As a consequence, all relations that should intuitively be a CBUC are
%  recognized by \Cref{def-cbut-51} as such.
  \qedhere
\end{proofparts}
\end{proof}

\subsection{Conditional Bisimilarity Up-To Context}

We start our investigation of conditional bisimilarity up-to context with the idea of a relation that can be extended to a conditional bisimulation.
To show, using such a conditional bisimulation up-to context $R$, that
a pair of arrows is conditionally bisimilar, it is not in general
necessary to find this pair in $R$, but one can instead extend a pair in $R$ to the pair under review.
As~this extension might provide parts of the context that the original condition referred to, it is necessary to shift the associated condition over the extension.

\begin{defi}[Conditional bisimulation up-to context (CBUC)]\label{def-cbut-51}%
  A conditional relation $R$ is a \emph{conditional bisimulation
  up-to context} if the following holds:
  for each triple $(a,b,\mathcal{C}) \in R$ and each context step
  $a \cstep{f}{\mathcal A} a^\prime$, there are \mbox{answering} steps
  $b \cstep{f}{\mathcal{B}_i} b_i^\prime$, $i\in I$, and conditions
  $\mathcal{C}_i^\dprime$ such that for each $i\in I$ there exists
  $(a_i^\dprime, b_i^\dprime, \mathcal{C}_i^\dprime) \in R$ with
  \mbox{$a^\prime = a_i^\dprime;j_i$},\ %
  \raisebox{0pt}[12pt][0pt]{}% <-- increases spacing above this line. this reduces the visual clash between \dprime of this line and _i of the previous one; also, it makes for a more consistent line spacing within the definition
  $b_i^\prime = b_i^\dprime;j_i$ for some
  arrow $j_i$  and additionally
  \mbox{$\mathcal{A} \land \mathcal{C}_{\downarrow f} \models \bigvee_{i \in
    I} \big( \mathcal{C}_{i \downarrow j_i}^{\prime\prime} \land
    \mathcal{B}_i \big)$};
  vice versa for steps $b \cstep{f}{\mathcal B} b'$.
\end{defi}
The situation for one answer step is depicted in \Cref{fig-blume}.
The weakest possible $\mathcal A, \mathcal B_i$ can be derived from the rule
conditions as $\mathcal A = \mathcal R_{\downarrow c},\ \mathcal B_i
= {\mathcal R_i}_{\downarrow e_i}$.

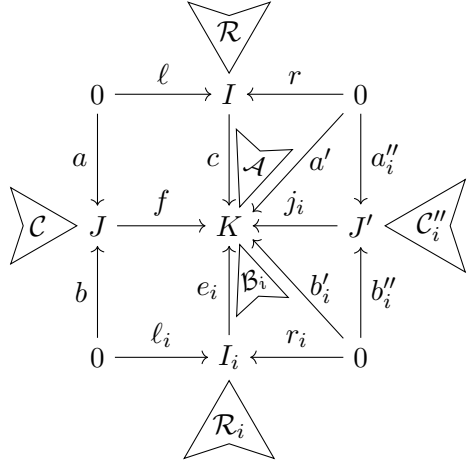
\begin{figure}[ht]%
\centering\begin{tikzpicture}

  \def\sq{1.75}

  \node (tlempty) at (0,0) {$\obnull$};
  \node (i) at (1*\sq,0) {$I$};
  \node (trempty) at (2*\sq,0) {$\obnull$};
  \node (j) at (0,-1*\sq) {$J$};
  \node (k) at (1*\sq,-1*\sq) {$K$};
  \node (jp) at (2*\sq,-1*\sq) {$J'$};
  \node (blempty) at (0,-2*\sq) {$\obnull$};
  \node (ii) at (1*\sq,-2*\sq) {$I_i$};
  \node (brempty) at (2*\sq,-2*\sq) {$\obnull$};
  \draw[->] (tlempty) -- node[above]{$\ell$} (i);
  \draw[->] (trempty) -- node[above]{$r$} (i);
  \draw[->] (trempty) -- node[below,pos=0.25]{$a^\prime$} (k.30);
  \draw[->] (blempty) -- node[above]{$\ell_i$} (ii);
  \draw[->] (brempty) -- node[above]{$r_i$} (ii);
  \draw[->] (brempty) -- node[above,pos=0.25]{$b_i^\prime$} (k.-30);

  \draw[->] (tlempty) -- node[left]{$a$} (j);
  \draw[->] (blempty) -- node[left]{$b$} (j);
  \draw[->] (i) -- node[left]{$c$} (k);
  \draw[->] (ii) -- node[left]{$e_i$} (k);
  \draw[->] (trempty) -- node[right]{$a_i^\dprime$} (jp);
  \draw[->] (brempty) -- node[right]{$b_i^\dprime$} (jp);

  \draw[->] (j) -- node[above]{$f$} (k);
  \draw[->] (jp) -- node[above]{$j_i$} (k);

  \node[condtri,dart tip angle=60,shape border rotate=270] at (i.north) {$\mathcal{R}$};
  \node[condtri,dart tip angle=60,shape border rotate=0] at (j.west) {$\kern 0.5mm \mathcal{C}$};
  \node[condtri,dart tip angle=60,shape border rotate=90] at (ii.south) {$\kern -0.5mm\mathcal{R}_i\kern -0.5mm$};
  \node[condtri,dart tip angle=60,shape border rotate=180] at (jp.east) {$\kern -1mm \mathcal{C}_i^\dprime$};
  \node[rotate around={-20:(k.center)},condtri,shape border rotate=270,scale=0.9] at (k.80) {$\mathcal{A}$};
  \node[rotate around={20:(k.center)},condtri,shape border rotate=90,scale=0.9] at (k.-80) {\raisebox{0pt}[1.8pt][1.5pt]{$\mathcal{B}_i$}};
\end{tikzpicture}%
\caption{A single answer step in conditional bisimulation up-to context}%
\label{fig-blume}
\end{figure}

Compared to a regular conditional bisimulation, which directly relates
the results of the answering steps
$(a^\prime, b_i^\prime, \mathcal{C}_i^\prime)$, in a CBUC it is
sufficient to relate some pair
$(a_i^\dprime, b_i^\dprime, \mathcal{C}_i^\dprime)$, where
$a_i^\dprime, b_i^\dprime$ are obtained from $a^\prime, b_i^\prime$ by
removing an identical context $j_i$.
(The conditional bisimilarity of the actual successors
$a^\prime, b_i^\prime$ can then be derived by contextualizing the relation, i.e.,
we use \Cref{thm-cbut-intu} and refer to a triple in $u(R)$ that is contextualized under $j_i$.)

\begin{rem}\label{def-cbuc-post-fixpoint}
  A CBUC can also be defined based on the closure under contextualization~$u$ (see \Cref{def-condrel-u},
  note that $u$ is easily seen to be a closure\footnote{Monotonicity: trivial.
  Extensiveness: $(a,b,\mathcal C) \in R$ implies $(a;\id,b;\id,\mathcal C_{\downarrow \id}) = (a,b,\mathcal C) \in u(R)$.
  Idempotence:
  $u \subseteq u \circ u$ by extensiveness and monotonicity,
  $u \circ u \subseteq u$ because subsequent contextualization first under $d_1$ then $d_2$ can also be done in a single step of $d_1;d_2$.
  }):
  A conditional relation $R$ is a CBUC if and only if
  $R \subseteq f_C(u(R))$
  (i.e.\ $R$ is a post-fixpoint of $f_C \circ u$).
  This can be seen by expanding the definitions of $f_C$ and $u$ on the right-hand side of $(a,b,\mathcal C) \in R \implies (a,b,\mathcal C) \in f_C(u(R))$,
  which results in exactly the definition of a CBUC\@.
\end{rem}

We now show that this up-to technique is useful or \emph{sound} (\Cref{def-soundness}),
that is, all
elements recognized as bisimilar by the up-to technique are actually
bisimilar~\cite{sangiorgi,ps:enhancements-coinductive}.
In fact, we prove the stronger result that the technique is $f_C$-compatible,
which (as outlined in \Cref{subsec-fixpoint-basics}) not only implies soundness,
but also makes it possible to combine our technique with other $f_C$-compatible
up-to techniques.

\begin{thm}[$u$ is $f_C$-compatible]\label{thm-compat-fc}
  Let $R$ be a conditional relation. Then it holds that
  $u(f_C(R)) \subseteq f_C(R)$. This implies that $u(f_C(R)) \subseteq
  f_C(u(R))$, i.e., $u$ is $f_C$-compatible.
\end{thm}

\begin{proof}
  We show that $(a,b,\mathcal C) \in u(f_C(R))$ implies $(a,b,\mathcal C) \in f_C(R)$.

  \newcommand{\pfeilohne}[2]{{#1}_{\setminus #2}}
  \newcommand{\conditionohne}[2]{{#1}_{\uparrow #2}}
  Within this proof, we use the following notation:
  Given arrows $x,d$,\ $\pfeilohne{x}{d}$ indicates an arrow such that $x = \pfeilohne{x}{d} ; d$.
  Given a condition $\mathcal C$, $\conditionohne{\mathcal C}{d}$ is a condition such that $\mathcal C = (\conditionohne{\mathcal C}{d})_{\downarrow d}$.

  \begin{itemize}
  \item
  Since $(a,b,\mathcal C) \in u(f_C(R))$,
  by definition of $u$ this means that
  \begin{quote}
  there exist $d, \pfeilohne{a}{d}, \pfeilohne{b}{d}, \conditionohne{\mathcal C}{d}$
  such that $a = \pfeilohne{a}{d};d,\ b = \pfeilohne{b}{d};d,\ \mathcal C = (\conditionohne{\mathcal C}{d})_{\downarrow d}$, and
  $(\pfeilohne{a}{d}, \pfeilohne{b}{d}, \conditionohne{\mathcal C}{d}) \in f_C(R)$
  \end{quote}

  \item
  This means that
  \begin{quote}
  for all steps $\pfeilohne{a}{d} \cstep{f}{\mathcal A} a'$
  there exist answering steps $\pfeilohne{b}{d} \cstep{f}{\mathcal B_i} b_i'$
  and conditions $\mathcal C_i'$
  such that
  $(a', b_i', \mathcal C_i') \in R$
  and
  $\mathcal A \land (\conditionohne{\mathcal C}{d})_{\downarrow f}
  \models \bigvee\nolimits_{i \in I} (\mathcal C_i' \land \mathcal B_i)$;
  vice versa for steps of $\pfeilohne{b}{d}$
  \end{quote}

  \item
  Some of the borrowed contexts $f$ for which this statement holds are of the shape $f = d;f'$, therefore, we also know that
  \begin{quote}
  for all steps $\pfeilohne{a}{d} \cstep{d;f'}{\mathcal A} a'$
  there exist answering steps $\pfeilohne{b}{d} \cstep{d;f'}{\mathcal B_i} b_i'$
  and conditions $\mathcal C_i'$
  such that
  $(a', b_i', \mathcal C_i') \in R$
  and
  $\mathcal A \land (\conditionohne{\mathcal C}{d})_{\downarrow d;f'}
  \models \bigvee\nolimits_{i \in I} (\mathcal C_i' \land \mathcal B_i)$;
  vice versa for steps of $\pfeilohne{b}{d}$
  \end{quote}

  \item
  Applying \Cref{lemma-ctxtrans-composition,shift-laws} to this statement, we obtain:
  \begin{quote}
  for all steps $\pfeilohne{a}{d};d \cstep{f'}{\mathcal A} a'$
  there exist answering steps $\pfeilohne{b}{d};d \cstep{f'}{\mathcal B_i} b_i'$
  and conditions $\mathcal C_i'$
  such that
  $(a', b_i', \mathcal C_i') \in R$
  and
  $\mathcal A \land ((\conditionohne{\mathcal C}{d})_{\downarrow d})_{\downarrow f'}
  \models \bigvee\nolimits_{i \in I} (\mathcal C_i' \land \mathcal B_i)$;
  vice versa for steps of $\pfeilohne{b}{d};d$
  \end{quote}

  \item
  Applying the equalities for the variables annotated with $\pfeilohne{}{d}$ and $\conditionohne{}{d}$, we get:
  \begin{quote}
  for all steps $a \cstep{f'}{\mathcal A} a'$
  there exist answering steps $b \cstep{f'}{\mathcal B_i} b_i'$
  and conditions $\mathcal C_i'$
  such that
  $(a', b_i', \mathcal C_i') \in R$
  and
  $\mathcal A \land (\mathcal C)_{\downarrow f'}
  \models \bigvee\nolimits_{i \in I} (\mathcal C_i' \land \mathcal B_i)$;
  vice versa for steps of $b$
  \end{quote}
  This is exactly the definition of $(a,b,\mathcal C) \in f_C(R)$.
  Therefore, $u(f_C(R)) \subseteq f_C(R)$.

  Since $R\subseteq u(R)$ due to extensiveness of $u$, we can infer
  $f_C(R) \subseteq f_C(u(R))$ since $f_C$ is monotone. Combined, this
  gives us $f_C$-compatibility of $u$.
  \qedhere
  \end{itemize}
\end{proof}

\noindent
Note the stronger result ($u(f_C(R)) \subseteq f_C(R)$ instead of just $u(f_C(R)) \subseteq f_C(u(R))$)
can intuitively be explained as follows:
since $f_C$ quantifies over all context steps
and the size of the borrowed context $f$ is not bounded,
this means the successor triples are already closed under contextualization.

% \begin{cor}[$u$ is $f_C$-compatible]\label{cor-compat-fc-two}
%   For any conditional relation $R$, it holds that
%   \mbox{$u(f_C(R)) \subseteq f_C(u(R))$}.
% \end{cor}

% \begin{proof}
%   \rnewhl{
%   One can either show this by following the proof of \Cref{thm-compat-fc},
%   then after the last step, conclude from $(a', b_i', \mathcal C_i') \in R$
%   that also $(a', b_i', \mathcal C_i') \in u(R)$ since contextualization under $\id$ is always possible.
%   The resulting statement matches the definition of $(a,b,\mathcal C) \in f_C(u(R))$,
%   which shows $u(f_C(R)) \subseteq f_C(u(R))$.

%   More formally, from \Cref{thm-compat-fc} we know $u \circ f_C \subseteq f_C$.
%   $u$~is extensive, so $\id \subseteq u$,
%   and since $f_C$~is monotone, we get $f_C \subseteq f_C \circ u$.
%   This totals to $u \circ f_C \subseteq f_C \subseteq f_C \circ u$.
%   }{}
% \end{proof}

From compatibility, we obtain as a corollary that this up-to technique is useful or \emph{sound}, that is, all
elements recognized as bisimilar by the up-to technique are actually
bisimilar (see \Cref{le:up-to-closure} and~\cite{sangiorgi,ps:enhancements-coinductive}).

\begin{thm}[Characterization of CBUC]\label{thm-cbut-intu}%
A conditional relation $R$ satis\-fies Definition~\ref{def-cbut-51} % not cref because arxiv uses dvips which breaks this link
(i.e.\ it is a CBUC) if and only if
its closure under contextualization $u(R)$
is a conditional bisimulation.
\end{thm}

\begin{proof}
  $R$ satisfying \Cref{def-cbut-51} is, by \Cref{def-cbuc-post-fixpoint},
  equivalent to $R$ being a post-fixpoint of $f_C \circ u$,
  i.e., $R \subseteq f_C(u(R))$.
  Also, $R$ satisfying the definition from \Cref{thm-cbut-intu},
  i.e., $u(R)$ being a conditional bisimulation, is, by \Cref{csimc-fc},
  equivalent to $u(R) \subseteq f_C(u(R))$.

  Since $u$ is $f_C$-compatible and $u$ is a closure,
  we can instantiate \Cref{lem-pride} to obtain the desired result.
  Hence, every relation that our initial definition recognizes as a CBUC
  indeed represents a conditional bisimulation (when closed under contextualization),
  and all relations that should intuitively be a CBUC are
  recognized by \Cref{def-cbut-51} as such.
\end{proof}

\begin{rem}
%[Soundness of the definition of \Cref{thm-cbut-intu}]
%\todo{evtl. Upgrade auf Corollary?}
    From \Cref{thm-cbut-intu} we easily obtain as a
    corollary that every CBUC $R$ is contained in $\csimC$
    ($R \subseteq \csimC$), i.e.\ all elements contained in some CBUC
    are indeed conditionally bisimilar.  This follows from the
    fact that $R \subseteq u(R)$ (set $d = \id_J$) and
    $u(R) \subseteq \csimC$ (since by \Cref{thm-cbut-intu} $u(R)$
    is a conditional bisimulation).
\end{rem}

  Note that while \Cref{thm-cbut-intu} gives a more accessible
  definition of CBUCs than Definition~\ref{def-cbut-51}, the latter definition % not Cref because dvips breaks link
  is more amenable to mechanization, since $R$ might be finite,
  whereas $u(R)$ is infinite.

% TODO
% Note that we could not simply use the definition of \Cref{thm-cbut-intu} in place of
% \Cref{def-cbut-51}: While the former is easily seen to be sound, it is unwieldy
% for proofs, as it generates a full conditional bisimulation $\hat R$ as an
% intermediate step, and requires an explicit check for conditional
% bisimilarity.  This complicates verification because $\hat R$ might
% contain infinitely many elements, which was the motivation behind
% introducing up-to techniques in the first place.

\subsection{Conditional Bisimilarity Up-To Context with Representative Steps}

CBUCs allow us to represent certain infinite bisimulation relations in a finite way.
For instance, we can use a finite CBUC in \Cref{ex-cb-repr-unreliable}.
However, automated checking if two agents are conditionally bisimilar --- which can be done by incrementally extending a conditional bisimulation relation --- is still hard, even using up-to context,
since up-to context can only reduce the size of the relation itself.
However, for just a single triple, there are infinitely many context steps to be checked.

For conditional bisimulations, we introduced an alternative definition using representative steps (\Cref{def-cb-r}) and showed that it yields an equivalent notion of conditional bisimilarity (\Cref{simr-simc}).
We will show that the same approach can be used for CBUCs.

\begin{defi}[CBUC with representative steps]\label{def-cbut-r}%
A \emph{CBUC with representative steps} is a conditional relation $R$ such that the following holds:
for each triple $(a,b,\mathcal{C}) \in R$ and each representative step $a \rstep{f}{\mathcal A} a^\prime$,
there are answering steps $b \cstep{f}{\mathcal{B}_i} b_i^\prime$ and conditions $\mathcal{C}_i^\dprime$
such that for each answering step there exists $(a_i^\dprime, b_i^\dprime, \mathcal{C}_i^\dprime) \in R$
\raisebox{0pt}[11pt][4.5pt]{}% <-- increases spacing above and below this line. there are several superscripts and subscripts that are quite close to each other, this aims to reduce the visual clash.
with $a^\prime = a_i^\dprime;j_i,\ b_i^\prime = b_i^\dprime;j_i$ for some arrow $j_i$ per answering step,
and additionally $\mathcal{A} \land \mathcal{C}_{\downarrow f} \models \bigvee_{i \in I} \big( \mathcal{C}_{i \downarrow j_i}^{\prime\prime} \land \mathcal{B}_i \big)$;
vice versa for steps \raisebox{0pt}[10pt][3pt]{$b \rstep{f}{\mathcal B} b'$}.
\end{defi}

\begin{rem}\label{def-cbutr-post-fixpoint}
  A CBUC with representative steps can also be defined based on the closure under contextualization~$u$:
  A conditional relation $R$ satisfies \Cref{def-cbut-r} if and only if
  $R \subseteq f_R(u(R))$
  (i.e.\ $R$ is a post-fixpoint of $f_R \circ u$).
  Analogously to \Cref{def-cbuc-post-fixpoint}, this can be seen by expanding definitions.
\end{rem}

To show that CBUCs defined using context and representative steps are essentially equivalent,
we first relate the underlying functions by showing $f_R \subseteq f_C \circ u$.%
\footnote{%
Intuitively, $f_R$ guarantees only that representative steps are answered
by context steps and that their successors are related again,
but $f_C$ requires this for non-representative steps as well,
so generally $f_R \nsubseteq f_C$.
Therefore we contextualize using $u$ to let $f_C$ access the non-representative successors as well.}
Afterwards, we use that result to show that the two up-to techniques are equivalent.

\begin{lem}\label{lem-wrath}
  $f_R(R) \subseteq f_C(u(R))$.
\end{lem}

\begin{proof}
  Let $(a,b,\mathcal C) \in f_R(R)$ be given, which by its definition means that:

  \begin{quote}% def of f_R(R)
    for all representative steps $a \rstep{\hat f}{\mathcal A} \hat{a'}$
    (1) there are answering steps \mbox{$b \cstep{\hat f}{\hat{\mathcal B_i}} \hat{b_i'}$}
    and conditions $\hat{\mathcal C_i'}$
    (2)\kern-1pt\ such that $(\hat{a'}, \hat{b_i'}, \hat{\mathcal C_i'}) \in R$
    and (3)\kern-1pt\ \mbox{$\mathcal A \land \mathcal C_{\downarrow \hat f}
    \models \mkern-2mu\bigvee\nolimits_{i \in I} (\hat{\mathcal C_i'} \land \hat{\mathcal B_i})$};
    vice versa for representative steps of $b$
  \end{quote}

  We show that this implies $(a,b,\mathcal C) \in f_C(u(R))$.
  Consider a context step $a \cstep{f}{\mathcal A} a'$.
  This step is not necessarily a representative step.
  According to \Cref{cond-repr-zurueckfuehren}, this context step can be reduced to a representative step
  $a \rstep{\hat f}{\mathcal{R}_{\downarrow \hat c}} \hat{a'}$, where
  $\mathcal R$ is the condition of the rule used for the step,
  $c$ is the reactive context of the context step,
  and there exists $\hat g$ such that $\hat f ; \hat g = f,\ \hat c ; \hat g = c,\ \hat{a'} ; \hat g = a'$,
  with $\hat f, \hat c, \hat{a'}$ referring to the representative step.

  Since $(a,b,\mathcal C) \in f_R(R)$, we know that answering steps for our representative step exist.
  From that we can conclude the following:
  \begin{enumerate}
    \item
      $b \cstep{\hat f}{\hat{\mathcal{B}_i}} \hat{b_i'}$ implies, according to \Cref{lemma16-dbccrs}, that a step
      $b \cstep{\hat f;\hat g}{{\hat{\mathcal{B}_i}}_{\downarrow \hat g}} \hat{b_i'};\hat g$, equivalently,
      $b \cstep{f}{{\hat{\mathcal{B}_i}}_{\downarrow \hat g}} \hat{b_i'};\hat g$, is possible.
      We select $\mathcal{B}_i \defeq {\hat{\mathcal{B}_i}}_{\downarrow \hat g},\ b_i' \defeq \hat{b_i'};\hat g,\ \mathcal C_i' \defeq {\hat{\mathcal C_i'}}_{\downarrow \hat g}$ for the answering steps of $b$ using borrowed context $f$.
    \item
      $(\hat{a'}, \hat{b_i'}, \hat{\mathcal C_i'}) \in R$ implies
      $(\hat{a'};\hat g,\ \hat{b_i'};\hat g,\ \hat{\mathcal C_i'}_{\downarrow \hat g})
      = (a', b_i', \mathcal C_i') \in u(R)$.
    \item
      Using the rules of \Cref{shift-laws} we get:
      \begin{align*}
        \mathcal R_{\downarrow \hat c} \land \mathcal C_{\downarrow \hat f}
        & \models \bigvee \left( \hat{\mathcal C_i'} \land \hat{\mathcal B_i} \right)
        \\ \Rightarrow
        (\mathcal R_{\downarrow \hat c})_{\downarrow \hat g} \land {\big(\mathcal C_{\downarrow \hat{f}}\big)}_{\downarrow \hat g}
        & \models \bigvee \left( \hat{\mathcal C_i'}_{\downarrow \hat g} \land \hat{\mathcal B_i}_{\downarrow \hat g} \right)
        \\ \Leftrightarrow
        \mathcal R_{\downarrow \hat c ;\hat g} \land \mathcal C_{\downarrow \hat{f};\hat{g}}
        & \models \bigvee \left( \hat{\mathcal C_i'}_{\downarrow \hat g} \land \hat{\mathcal B_i}_{\downarrow \hat g} \right)
        \\ \Leftrightarrow
        \mathcal R_{\downarrow c} \land \mathcal C_{\downarrow f}
        & \models \bigvee \left( \mathcal C_i' \land \mathcal{B}_i \right)
      \end{align*}

      \noindent
      For context steps, $\mathcal A \models \mathcal R_{\downarrow c}$ holds,
      so we have $
        \mathcal A \land \mathcal C_{\downarrow f}
        \models \mathcal R_{\downarrow c} \land \mathcal C_{\downarrow f}
        \models \bigvee \left( \mathcal{C}_i' \land \mathcal{B}_i \right)
      $.
  \end{enumerate}

  \noindent
  Analogously, we can construct answering steps for $b \cstep{f}{\mathcal B} b'$.
  To summarize, for the given $(a,b,\mathcal C) \in f_R(R)$ we have concluded that
  \begin{quote}% def of f_C(u(R))
    for all context steps $a \cstep{f}{\mathcal A} a'$
    (1) there exist answering steps $b \cstep{f}{\mathcal B_i} b_i'$
    and conditions $\mathcal C_i'$
    (2) such that $(a', b_i', \mathcal C_i') \in u(R)$
    and (3) $\mathcal A \land \mathcal C_{\downarrow f}
    \models \bigvee\nolimits_{i \in I} (\mathcal C_i' \land \mathcal B_i)$;
    vice versa for context steps of $b$
  \end{quote}
  which is exactly the definition of $(a,b,\mathcal C) \in f_C(u(R))$.
\end{proof}

\begin{cor}\label{cor-gluttony}
  It holds that $f_C \circ u = f_R \circ u$.
\end{cor}

\begin{proof}
\begin{proofparts}
\proofPartNoNewline{$f_C \circ u \subseteq f_R \circ u$}
  %$f_C \subseteq f_R$ implies $f_C \circ u \subseteq f_R \circ u$.
  \begin{align*}
    \mkern31mu% approximately line up with the \iff and \subseteq of the align* below
    f_C \subseteq f_R \implies & f_C \circ u \subseteq f_R \circ u
  \end{align*}

\proofPartNoNewline{$f_R \circ u \subseteq f_C \circ u$}
  \begin{align*}
    \text{(\Cref{lem-wrath})} \phantom{{}\iff} & f_R \phantom{{}\circ u} \subseteq f_C \circ u \\
    \implies & f_R \circ u \subseteq f_C \circ u \circ u \\
    \text{($u$ idempotent)} \iff & f_R \circ u \subseteq f_C \circ u
    \qedhere
  \end{align*}
\end{proofparts}
\end{proof}

\begin{thm}\label{cbuts-equiv}%
  A conditional relation is a CBUC (\Cref{def-cbut-51}) if and only
  if it is a CBUC with representative steps (\Cref{def-cbut-r}).
\end{thm}

\begin{proof}
  $R$ satisfying Definitions~\ref{def-cbut-51} and~\ref{def-cbut-r} means that
  $R$ is a post-fixpoint of $f_C \circ u$ or $f_R \circ u$, respectively.
  By \Cref{cor-gluttony}, $f_C \circ u = f_R \circ u$,
  therefore post-fixpoints of one are also post-fixpoints of the other.
\end{proof}

Observe that this is a stronger result than for normal conditional bisimilarity:
for that, we know that the bisimilarities are the same ($\nu f_C = \nu f_R$) but the bisimulation functions are not ($f_C \ne f_R$).
On the other hand, using up-to techniques, the difference between the bisimulation functions themselves disappears
(i.e.\ $f_C \circ u = f_R \circ u$ instead of just $\nu(f_C \circ u) = \nu(f_R \circ u)$).
This results from \Cref{lem-wrath} and can be explained intuitively as follows:
the function $f_R$ requires that every representative step can be
answered by a context step and the resulting pair is in $R$,
while $f_C$ requires such an answer for all context steps. This means
that $f_C(R)\subseteq f_R(R)$ as explained earlier. The pairs
potentially missing in $f_C(R)$  resulted from larger-than-necessary contexts and hence $f_R$ did not require them.
Using~$u$, however, the relation is contextualized beforehand, using all (even non-representative) contexts, and hence makes these triples ``available'' to $f_C$.

Note that even though the difference between the two variants disappears,
using representative steps with CBUCs still is advantageous because it typically results in a finitely branching transition system.
This can be seen in the continuation of our example:

\begin{exa}
Consider again \Cref{ex-cb-unreliable,ex-cb-repr-unreliable}.
We have previously seen that it is possible to repeatedly borrow a
message on the left-hand node and transfer it to the right-hand node,
which leads to more and more received messages accumulating at the right-hand node.
We now show that the two types of channels are conditionally bisimilar by showing that
$R = \big\{ ( \emptyset \rightarrow \inlinechan{reliable} \leftarrow \inlinechan{},\allowbreak\ \emptyset \rightarrow \inlinechan{unreliable} \leftarrow \inlinechan{},\allowbreak\ \AIneN ) \big\}$
is a CBUC, i.e.\ it satisfies \Cref{def-cbut-r}.
We consider the same steps as in \Cref{ex-cb-repr-unreliable}:

\begin{itemize}
\item The graph $\inlinechan{reliable}$ can do a step using rule $P_R$ by borrowing a message on the left \linebreak node, with environment condition $\mathcal A = \condtrue$, and reduces to $a' = \emptyset \rightarrow \inlinechan{reliable,msgright} \leftarrow \inlinechan{}$.\linebreak
Then, $\inlinechan{unreliable}$ can answer this step using $P_U$ under $\mathcal B_i = \AIneN$ (no noise) and reacts to $b_i^\prime = \emptyset \rightarrow \inlinechan{unreliable,msgright} \leftarrow \inlinechan{}$.

Now set $j_i = \inlinechan{} \rightarrow \inlinechan{msgright} \leftarrow \inlinechan{}$, i.e.\ we consider the $m$-loop on the right node as irrelevant context.
Then, using $a_i^\dprime = \emptyset \rightarrow \inlinechan{reliable} \leftarrow \inlinechan{},\ b_i^\dprime = \emptyset \rightarrow \inlinechan{unreliable} \leftarrow \inlinechan{},\ \mathcal C_i^\dprime = \AIneN$ we have $a' = a_i^\dprime ; j_i,\ b_i^\prime = b_i^\dprime ; j_i$,
and we find that the triple without the irrelevant context $j_i$, that is $(a_i^\dprime, b_i^\dprime, \mathcal C_i^\dprime)$ (which happens to be the same as our initial triple), is contained in $R$.
As before, the implication $\mathcal A \land \mathcal C_{\downarrow f} \models \bigvee_{i \in I} \left( \mathcal{C}_i^\dprime \land \mathcal{B}_i \right)$ holds.

\item Symmetrically, $\inlinechan{unreliable}$ borrows a message on the left node and reacts to $\inlinechan{unreliable,msgright}$ under $\mathcal A = \AIneN$.
Analogously to the previous case and to \Cref{ex-cb-repr-unreliable}, $\inlinechan{reliable}$ answers this step, using $\mathcal C_i^\dprime = \AIneN$ and $j_i = \inlinechan{} \rightarrow \inlinechan{msgright} \leftarrow \inlinechan{}$.

\item Again, the remaining representative steps can be proven in an analogous way.
\end{itemize}

\noindent% we did this purely for readability reasons and NOT to gain extra space.
% without noindent, "Note that..." looks like it is still part of the list item
\Cref{fig-infinite-cb} shows a comparison of the necessary steps with and without using up-to techniques.
Note that instead of working with an infinite bisimulation, we now
have a singleton.
\end{exa}

Finally, we show compatibility of $f_R$ and summarize the theorems of this section.

\begin{cor}[$u$ is $f_R$-compatible]\label{cor-compat-fr}
  Let $R$ be a conditional relation. Then it holds that
  $u(f_R(R)) \subseteq f_R(u(R))$, i.e., $u$ is $f_R$-compatible.
\end{cor}

\begin{proof}
  By transitivity from previous results:
  \begin{align*}
    \text{(\Cref{lem-wrath})} && f_R &\subseteq \phantom{u \circ {}} f_C \circ u \\
    \implies && u \circ f_R &\subseteq u \circ f_C \circ u \\
    \text{($u \circ f_C \subseteq f_C$ by \Cref{thm-compat-fc})} && &\subseteq \mkern14mu f_C \mkern14mu \circ u \\
    \text{(by \Cref{cor-gluttony})} && &= f_R \circ u
    \qedhere
  \end{align*}
\end{proof}

\Cref{fig-upto-subseteq} summarizes the known inclusions and equalities that were proven throughout this section.

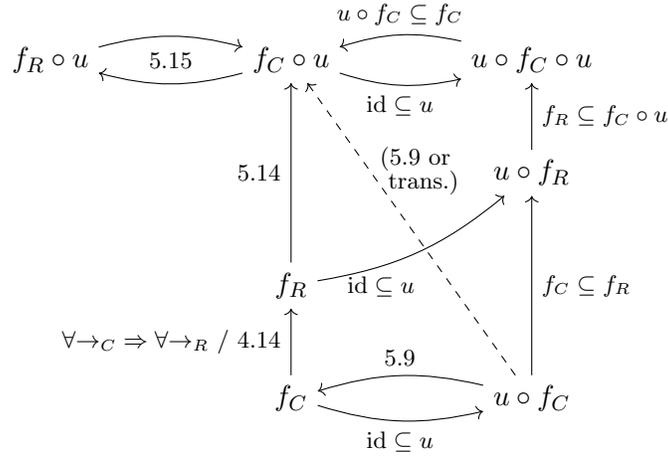
\begin{figure}[ht]
  \begin{tikzpicture}[x=3.2cm,y=1.5cm]
    \node (uc)  at (1,0) {$u \circ f_C$};
    \node (c)   at (0,0) {$f_C$};
    \node (r)   at (0,1) {$f_R$};
    \node (ur)  at (1,2) {$u \circ f_R$};
    \node (cu)  at (0,3) {$f_C \circ u$};
    \node (ru)  at (-1,3) {$f_R \circ u$};
    \node (ucu) at (1,3) {$u \circ f_C \circ u$};
    \begin{scope}[->, font=\footnotesize]
    \draw (uc) to[bend right=15] node[above] {\ref{thm-compat-fc}} (c);
    \draw (c)  to node[left] {$\forall {\to_C} \Rightarrow \forall {\to_R}$ / \ref{csimr-fr}} (r);
    \draw (r)  to node[left] {\ref{lem-wrath}} (cu);
    \draw[dashed] (uc) to node[right,align=left,pos=0.7] {(\ref{thm-compat-fc} or\\[-2pt]\ trans.)} (cu);
    \draw (uc) to node[right] {$f_C \subseteq f_R$} (ur);
    \draw (cu) to[bend left=15] (ru); % node[above] {\ref{cor-gluttony}} (ru);
    \draw (ru) to[bend left=15] (cu);
    \node at ($(ru)!0.5!(cu)$) {\ref{cor-gluttony}};
    \draw (ur) to node[right] {$f_R \subseteq f_C \circ u$} (ucu);
    \draw (ucu)to[bend right=15] node[above] {$u \circ f_C \subseteq f_C$} (cu);
    \draw (cu) to[bend right=15] node[below] {$\id \subseteq u$} (ucu);
    \draw (r) to[bend right=15] node[below,pos=0.3] {$\id \subseteq u$} (ur);
    \draw (c) to[bend right=15] node[below] {$\id \subseteq u$} (uc);
    \end{scope}
  \end{tikzpicture}
  \caption{Known relations ($X \to Y$ indicates $X \subseteq Y$)}%
  \label{fig-upto-subseteq}
\end{figure}

Note that using the results of this section,
it is possible to provide alternative proofs of various theorems of \Cref{sec:condbisim},
in particular:

\begin{itemize}
  \item \Cref{cond-strengthening} ($\csimC$ is closed under condition strengthening):
  We can define the function
  $str(R) \defeq \{ (a,b,\mathcal C) \mid (a,b,\mathcal C') \in R,\ \mathcal C \models \mathcal C' \}$
  and show that $str$ is $f_C$-compatible by showing
  $(a,b,\mathcal C) \in str(f_C(R)) \implies (a,b,C) \in f_C(str(R))$.
  This can be done by expanding the definitions of $str, f_C$ in $(a,b,\mathcal C) \in str(f_C(R))$
  and rewriting it to match the definition of $(a,b,\mathcal C) \in f_C(str(R))$.
  Then, we can apply \Cref{compat-preserves-fp}.

%  \todohl{vielleicht geht auch noch ``$\csimC$ ist transitiv'', unter Verwendung der $f_C$-Compat von $str$ (strengthening Funktion, s.o.)?
%  Aber dazu müsste man $R_1R_2$ (oder zumindestens $\csimC\csimC$)
%  irgendwie schön als Funktion ausdrücken können. B: Du könntest
%  transitive Hülle machen. Aber ich denke, wir lassen das erstmal so.}{}

  \item \Cref{cb-congruence} part 4 ($\csimC$ is closed under contextualization):
  Restated using $u,f_C$ (note that $\csimC = \nu f_C$),
  we have to show that $u(\nu f_C) \subseteq \nu f_C$.
  This follows immediately from $f_C$-compatibility and \Cref{compat-preserves-fp}.

  \item \Cref{lemma-cbr-cuc} ($\csimR$ is closed under contextualization):
  As for $\csimC$, but use $f_R$ instead of $f_C$.

  \item \Cref{simr-simc} ($\csimC = \csimR$):
  We show that $\nu f_C = \nu f_R$.
  For this we use the fact that $f_R \circ u = f_C \circ u$ (\Cref{cor-gluttony})
  and therefore also $\nu(f_R \circ u) = \nu(f_C \circ u)$.
  $u$ is $f_C$-compatible (\Cref{thm-compat-fc}),
  so by \Cref{le:up-to-closure}, \mbox{$\nu(f_C \circ u) = \nu f_C$}.
  Similarly, $u$ is $f_R$-compatible (\Cref{cor-compat-fr})
  and therefore $\nu(f_R \circ u) = \nu f_R$.
  Combining these results, we have
  $\csimC = \nu f_C = \nu(f_C \circ u) = \nu(f_R \circ u) = \nu f_R = \csimR$.
\end{itemize}

\section{Comparison and An Alternative Characterization}%
\label{sec:alternative-char}

\subsection{An Equivalent Characterization Based on Environment Steps}

  We will now give an alternative characterization of
  conditional bisimilarity, in order to justify
  Definitions~\ref{def-cb-c} and~\ref{def-cb-r}. This alternative definition is
  more elegant since it characterizes $\csimC$ as the largest
  conditional congruence that is a conditional environment bisimulation.
  On the other hand, this definition is (like conditional bisimilarity, as described in \Cref{ex-cb-repr-unreliable}) not directly suitable for mechanization,
  since the underlying transition system is not finitely branching.

In~\cite{DBC-CRS}, environment steps, which capture the idea that a
reaction is possible under some \emph{passive} context $d$, have been
defined to obtain a more natural characterization of saturated
bisimilarity.  Unlike the borrowed context $f$, the passive context
$d$ does not participate in the reaction itself, but we refer to it to
ensure that the application condition of the rule holds.

\begin{defi}[Environment step~\cite{DBC-CRS}]
Let $\mathcal S$ be a conditional reactive system and let $a \colon
\obnull \to K,\ a' \colon \obnull \to K,\ d\colon K\to J$ be arrows.
We write $a \envstep{d} a'$ whenever there exists a rule $(\ell,r,\mathcal R) \in \mathcal S$ and an arrow $c$ such that $a = \ell;c,\ a' = r;c$ and $c;d \models \mathcal R$.
\end{defi}

Environment steps and context steps are related: they can be
transformed into each other. Furthermore saturated bisimilarity is the
coarsest bisimulation relation over environment steps that is also a
congruence~\cite{DBC-CRS}. We now give a characterization of
conditional bisimilarity based on environment steps:

\begin{defi}[Conditional environment congruence]%
  \label{def-cond-env-cong}
A conditional relation $R$ is a \emph{conditional environment bisimulation} if
whenever $(a,b,\mathcal C) \in R$ and $a \envstep{d} a'$ for some $d \models \mathcal C$,
then $b \envstep{d} b'$ and $(a', b', \mathcal C') \in R$ for some condition $\mathcal C'$ such that $d \models \mathcal C'$;
vice versa for $b \envstep{d} b'$.
We denote by $\csimenv$ the largest conditional environment
bisimulation that is also a conditional congruence and call it
\emph{conditional environment congruence}.
\end{defi}

For the proof of \Cref{cec-cb}, we need the following lemma:

\begin{lemC}[{\cite[Lemma 22]{DBC-CRS}}]\label{lemma-22}%
Given a context step $a \cstep{f}{\mathcal A} a'$ and a passive context $d$ such that $d \models \mathcal A$, we have an environment step $a;f \envstep{d} a'$.
Conversely, given an environment step $a;f \envstep{d} a'$, there exists a condition $\mathcal A$ such that $d \models \mathcal A$ and we have a context step $a \cstep{f}{\mathcal A} a'$.
\end{lemC}

\begin{thm}\label{cec-cb}%
Conditional bisimilarity and conditional environment congruence \coincide, that is, $\csimC = \csimenv$.
\end{thm}

\begin{proof}
In both parts we show only how steps of $a$ can be answered by $b$, the other direction can be shown analogously.
\begin{proofparts}
\proofPart{$\csimC \subseteq \csimenv$}
  We show that $\csimC$ is a conditional environment bisimulation.
  Together with the fact that $\csimC$ is a conditional congruence (\Cref{cb-congruence}),
  we obtain the result that $\csimC$ is contained in conditional environment congruence.

  Let $(a,b,\mathcal C) \in \csimC$ and $a \envstep{d} a'$ for some $d \models \mathcal C$.
  We rewrite $a \envstep{d} a'$ as $a;\id \envstep{d} a'$ and, using \Cref{lemma-22}, obtain a context step $a \cstep{\id}{\mathcal A} a'$ for some condition $\mathcal A$ such that $d \models \mathcal A$.

  Since $(a,b,\mathcal C) \in \csimC$, there exist answering steps $b \cstep{\id}{\mathcal B_i} b_i^\prime$ such that $(a',b_i^\prime,\mathcal C_i^\prime) \in \csimC$ for some $\mathcal C_i^\prime$ and $\mathcal A \land \mathcal C_{\downarrow\id} \models \bigvee_{i \in I}\left( \mathcal B_i \land \mathcal C_i^\prime \right)$.
  Since $d$ satisfies both $\mathcal A$ and $\mathcal C \equiv \mathcal C_{\downarrow\id}$,
  there exists an index $i$ such that $d \models \mathcal B_i \land \mathcal C_i^\prime$ and $(a', b_i^\prime, \mathcal C_i^\prime) \in \csimC$.
  This directly gives us the answering step required by conditional environment bisimilarity:
  Since $d \models \mathcal B_i$, using \Cref{lemma-22} we rewrite the corresponding context step $b \cstep{\id}{\mathcal B_i} b_i^\prime$ to an environment step $b;\id = b \envstep{d} b_i^\prime$.
  Setting $b' \defeq b_i^\prime,\ \mathcal C' \defeq \mathcal C_i^\prime$, we obtain the required triple $(a', b', \mathcal C') \in \csimC$.

\proofPart{$\csimenv \subseteq \csimC$}
  We show that $\csimenv$ is a conditional bisimulation.
  Let $(a,b,\mathcal C) \in \csimenv$ and $a \cstep{f}{\mathcal A} a'$.

  Let $d$ be some context. If $d \notmodels \mathcal A \land \mathcal C_{\downarrow f}$, we can easily satisfy \Cref{def-cb-c} by letting $b$ answer with an empty set of answering steps.
  We therefore assume that $d \models \mathcal A \land \mathcal C_{\downarrow f}$.

  Since $d \models \mathcal A$, using \Cref{lemma-22} we rewrite $a \cstep{f}{\mathcal A} a'$ to $a;f \envstep{d} a'$.
  As $\csimenv$ is, by definition, a conditional congruence, $(a,b,\mathcal C) \in \csimenv$ implies $(a;f, b;f, \mathcal C_{\downarrow f}) \in \csimenv$.

  Since $(a;f, b;f, \mathcal C_{\downarrow f}) \in \csimenv$, $a;f \envstep{d} a'$ and $d \models \mathcal C_{\downarrow f}$,
  there exists an answering step $b;f \envstep{d} b_d^\prime$ and for some condition $\mathcal C_d^\prime$ such that $d \models \mathcal C_d^\prime$ we have $(a', b_d^\prime, \mathcal C_d^\prime) \in \csimenv$.

  By \Cref{lemma-22}, $b;f \envstep{d} b_d^\prime$ implies $b \cstep{f}{\mathcal B_d} b_d^\prime$ for some $\mathcal B_d$ such that $d \models \mathcal B_d$.

  Thus, whenever $d \models \mathcal A \land \mathcal C_{\downarrow f}$,
  there exists an answering step $b \cstep{f}{\mathcal B_d} b_d^\prime$
  such that $d \models \mathcal B_d \land \mathcal C_d^\prime$
  and $(a', b_d^\prime, \mathcal C_d^\prime) \in \csimenv$.
  which concludes the proof of $\csimenv$ being a conditional
  bisimulation.
  \qedhere
\end{proofparts}
\end{proof}

\noindent
Note that in the second part of the proof, we use the fact that $b$ can reply
with an infinite set of answering steps, since the infinitely many answering
steps $b;f \envstep{d} b_d'$ might give rise to infinitely many different
$\mathcal B_d$ and accompanying $\mathcal C_d'$.

It is an open question if the proof is also possible with finitely many
answering steps. In~\cite[Theorem 23]{DBC-CRS}, a similar comparison of
saturated bisimilarity and environment congruence was done, although for
binary relations which did not include conditions in the relation itself.
In that proof, the finiteness assumption (\textsc{Fin}) was used to obtain
a finite set of answering steps, which is however not possible in the
presence of conditions.

\subsection{Comparison to Other Equivalences}

We conclude this section by considering
the binary relation
\mbox{$\csimtrue
\defeq \{ (a,b) \mid (a,b,\condtrue) \in \csimC \}$},
derived from
conditional bisimilarity, which is ternary. Intuitively it contains
pairs $(a,b)$, where $a,b$ are system states that behave equivalently
in every possible context.
% Several behavioural equivalences exist that check for identical
% behaviour not just for contexts satisfying a given condition, but for
% all contexts. Based on conditional bisimilarity, we define the relation
% $\csimtrue \defeq \{ (a,b) \mid (a,b,\condtrue) \in \csimC \}$.
We investigate how $\csimtrue$ compares to other behavioural
equivalences that also check for identical behaviour in all contexts.
First, we consider saturated bisimilarity ($\simC$), which has been
characterized in~\cite{DBC-CRS} as the coarsest relation which is a
congruence as well as a bisimilarity:

\begin{thm}\label{satz-truesb}%
Saturated bisimilarity implies $\condtrue$-conditional bisimilarity ($\simC \subseteq \csimtrue$).
However, $\condtrue$-conditional bisimilarity does \emph{not} imply saturated bisimilarity ($\csimtrue \nsubseteq \simC$).
\end{thm}

\begin{proof}
\begin{proofparts}
\proofPart{$\simC \subseteq \csimtrue$}
  Let $R$ be a saturated bisimulation relation. Then we define
  $R'=R\times\{\condtrue\}$ and show that $R'$ is a conditional
  bisimulation relation. For that purpose, let some $(a,b,\condtrue)\in
  R'$ be given, i.e.\ $(a,b)\in R$.  Now assume a transition $a
  \cstep{f}{\mathcal A} a^\prime$, then by the fact that $(a,b)\in R$ we
  know that there exist some answering steps $b \cstep{f}{\mathcal{B}_i}
  b_i^\prime$, $i\in I$, such that $(a',b_i')\in R$ for all $i\in I$. By
  definition of $R'$ it follows that, for all $i\in I$,
  $(a',b_i',\condtrue)\in R'$. So it remains to show that $\mathcal
  A\wedge\condtrue_{\downarrow f}\models\bigvee_{i\in I}\left( \condtrue
    \mathop\land \mathcal{B}_i \right)$. We can simplify this to $\mathcal
  A\models\bigvee_{i\in I} \mathcal{B}_i $, which holds because $R$ is a saturated bisimulation.

  Steps $b \cstep{f}{\mathcal B} b'$ can be answered analogously.

\proofPart{$\csimtrue \nsubseteq \simC$}
  Consider the following reactive system $\{R_A, R_{B1}, T_{B1}, R_{B2}, T_{B2}\}$:
  \[\bgroup%
  \renewcommand{\arraystretch}{1.5}%
  \setlength{\arraycolsep}{0.3pt}%
  \begin{array}{rccclrcccl} % ...=... | L | ... | R | ... | ...=... | L | ... | R | ...
    R_A    = \big( \emptyset \rightarrow & \inlinenlg{a}  & \leftarrow \emptyset,\  \emptyset \rightarrow & \inlinenlg{ea} & \leftarrow \emptyset, \condtrue_\emptyset \big) \\
    R_{B1} = \big( \emptyset \rightarrow & \inlinenlg{b}  & \leftarrow \emptyset,\  \emptyset \rightarrow & \inlinenlg{b1} & \leftarrow \emptyset, \mathcal A_C \big) &
    T_{B1} = \big( \emptyset \rightarrow & \inlinenlg{b1} & \leftarrow \emptyset,\  \emptyset \rightarrow & \inlinenlg{e1} & \leftarrow \emptyset, \neg\mathcal A_C \big) \\
    R_{B2} = \big( \emptyset \rightarrow & \inlinenlg{b}  & \leftarrow \emptyset,\  \emptyset \rightarrow & \inlinenlg{b2} & \leftarrow \emptyset, \neg\mathcal A_C \big) &
    T_{B2} = \big( \emptyset \rightarrow & \inlinenlg{b2} & \leftarrow \emptyset,\  \emptyset \rightarrow & \inlinenlg{e2} & \leftarrow \emptyset, \mathcal A_C \big)
  \end{array}\egroup\]
  where $
    \mathcal A_C = \big(
      \emptyset, \forall, \big\{ (
        \emptyset \rightarrow \inlinenlg{c} \leftarrow \inlinenlg{c},\  \condtrue_C
      ) \big\}
    \big)
    $.
  An $a$-loop can be replaced with a graph which allows no further steps.
  A $b$-loop can, in case the environment contains a $c$-loop ($\mathcal A_C$), transition to a $b_1$-loop,
  from which another transition is possible if \emph{no} $c$-loop is present
  (as this contradicts the condition of the first step, this transition can never actually be executed).
  Similarly, if \emph{no} $c$-loop is present, a transition to a $b_2$-loop is possible and subsequently another transition is possible \emph{if} there is a $c$-loop.

  It is easy to see that no matter which context $\inlinenlg{a}$ and $\inlinenlg{b}$ are placed into, both admit at most one transition.
  Therefore, $(\inlinenlg{a},\inlinenlg{b}) \in \csimtrue$, as witnessed by the conditional bisimulation relation
  $R = \{
  (\inlinenlg{a},\inlinenlg{b},\condtrue),\ %
  (\inlinenlg{ea},\inlinenlg{b1},\mathcal A_C),\ %
  (\inlinenlg{ea},\inlinenlg{b2},\neg\mathcal A_C)
  \}$.
  For saturated bisimilarity however, the initial step of $\inlinenlg{a}$ to $\inlinenlg{ea}$ can be answered by $\inlinenlg{b}$ with two steps as for conditional bisimilarity, and it would be required that $(\inlinenlg{ea},\inlinenlg{b1}),(\inlinenlg{ea},\inlinenlg{b2}) \in \simC$.
  But then, $\inlinenlg{b1}$ can do a step (under $\neg\mathcal A_C$ as indicated by rule $T_{B1}$) which $\inlinenlg{ea}$ cannot answer.
  \qedhere
\end{proofparts}
\end{proof}

\noindent
For saturated bisimilarity,
if a step of $a$ is answered by $b$ with multiple steps, all
$b_i^\prime$ reached in this way must be saturated bisimilar to $a'$
(that is, show the same behaviour even if the environment is later changed to one which did not allow the given $b_i^\prime$ to be reached).
In fact, it was an explicit goal in the design of saturated bisimilarity to account for external modification of the environment.
% todo: hier vielleicht die Erklärung ausbauen (der Teil mit der externen Modifikation)

On the other hand, for conditional bisimilarity, each $b_i^\prime$ is
only required to be conditionally bisimilar to $a'$ under the
condition which allowed this particular answering step --- that is,
after a step, the environment is fixed (or, depending on the system,
can only assume a subset of all possible environments,
cf.\ \Cref{def-cond-env-cong,cec-cb}).

Next, we compare $\csimtrue$ to $\id$-congruence,
the coarsest congruence contained in bisimilarity over the reaction relation $\leadsto$.
It simply relates two agents whenever they are bisimilar in all contexts,
i.e.\ $\simidc \defeq \{ (a,b) \mid \text{for all contexts $d$, }\ a;d,\
b;d \text{ are bisimilar wrt.} \leadsto \}$.

\begin{thm}\label{cb-vs-id}%
  It holds that $\condtrue$-conditional bisimilarity implies
  $\id$-congruence ($\csimtrue \subseteq \simidc$).  However,
  $\id$-congruence does \emph{not} imply $\condtrue$-conditional
  bisimilarity ($\simidc \nsubseteq \csimtrue$).
\end{thm}

\begin{proof}
\begin{proofparts}
\proofPart{$\csimtrue \subseteq \simidc$}
  Given $(a,b) \in \csimtrue$, equivalently $(a,b,\condtrue) \in \csimC$,
  by \Cref{satz-brauchbar} we know that for all contexts $d$ such that $d \models \condtrue$ (i.e.\ all contexts),
  $(a;d, b;d)$ is contained in a bisimulation relation over $\leadsto$.
  This, however, is the exact requirement for $(a,b) \in \simidc$.

\proofPart{$\simidc \nsubseteq \csimtrue$}
  Consider the example presented in \Cref{converse-brauchbar}:
  The graphs $A$ and $B$ are bisimilar under all contexts
  (i.e., \mbox{$(a,b) \in \simidc$}),
  however, they are not conditionally bisimilar under~$\condtrue$
  (i.e., \mbox{$(a,b) \notin \csimtrue$}).
  \qedhere
\end{proofparts}
\end{proof}

  \noindent
  Intuitively, $\condtrue$-conditional bisimilarity allows to
  observe whether some item is consumed and recreated (by including it
  in both sides of a rule) or whether it is simply required (using an
  existential rule condition, cf.~\Cref{cb-vs-id}).  On the other
  hand, $\id$-congruence does not recognize this and simply checks
  whether reactions are possible in the same set of contexts.

  Hence we have $\simC \subsetneq \csimtrue \subsetneq \simidc$, which
  implies that checking for identical behaviour in all contexts using
  conditional bisimilarity gives rise to a new kind of behavioural
  equivalence, which does not allow arbitrary changes to the
  environment (as $\simC$ does), yet allows distinguishing borrowed
  and passive context (which $\simidc$ does not).

\section{Conclusion, Related and Future Work}%
\label{sec:conclusion}

The conditions that we studied in this paper are also known under the
name of nested conditions or graph conditions and were introduced in~\cite{r:representing-fol}, where their equivalence to first-order
logic was shown. They were studied more extensively in~\cite{PennemannNAC,p:development-correct-gts} and generalized~to
reactive systems in~\cite{CRS}. In fact, the related notion of Q-trees
was introduced earlier in~\cite{fs:categories-allegories}.

As stated earlier, there are some scattered approaches to
notions of behavioural equivalence that can be compared to conditional
bisimilarity. The concept of behaviour depending~on a context is also
present in Larsen's PhD thesis~\cite{l:context-dependent-bisim}. There,
the idea is to embed an LTS into an environment, which is modelled as
an action transducer, an LTS that consumes transitions of the system
under investigation --- similar to CCS synchronization. Larsen
then defines environment-parameterized bisimulation by considering
only those transitions that~are consumed in a certain environment.
In~\cite{hl:symbolic-bisimulations}, Hennessy and Lin describe
symbolic bisimulations in the setting of value-passing processes,
where Boolean expressions restrict the interpretations for which one
shows bisimilarity. Instead in~\cite{bbb:bisim-unification}, Baldan,
Bracciali and Bruni propose bisimilarity on open systems, specified by
terms with a hole or place-holder. Instead of imposing conditions on
the environment, they restrict the components that are filling the~%
holes.

In~\cite{f:bisimulations-boolean-vectors}, Fitting studies a matrix
view of unlabelled transition systems, annotated by Boolean
conditions. In~\cite{BKKS17} we have shown that such systems can
alternatively be viewed as conditional transition systems, where
activation of transitions depends on conditions of the environment and
one can state the bisimilarity of two states provided that the
environment meets certain requirements. This view is closely tied to
featured transition systems, which have been studied extensively in
the software engineering literature. The idea here is to specify
system behaviour dependent on the features that are present in the
product (see for instance~\cite{DBLP:conf/icse/CordyCPSHL12} for
simulations on featured transition systems).

Our contribution in this paper is to consider conditional bisimilarity
based on contextualization in a rule-based setting. That is, system
behaviour is specified by generic rewriting rules, system states can
be composed with a context specifying the environment and we impose
restrictions on those contexts. By viewing both system states and
contexts as arrows of a category, we can work in the framework of
reactive systems \`a la Leifer and Milner and define a general theory
of conditional bisimilarity.  While in~\cite{DBC-CRS} conditions were
only used to restrict applicability of the rules and bisimilarity was
checked for all contexts, we here additionally use conditions to
establish behavioural equivalence only in specific contexts.

As future work we want to take a closer look at the logic that we used
to specify conditions. Conditional bisimilarity is defined in a way
that is largely independent of the kind of logic, provided that the
logic supports Boolean operators and shift. It is unclear and worth
exploring whether the logic considered by us is expressive enough to
characterize all contexts that ensure bisimilarity of two given
arrows. This also affects the question whether or not infinitely many
answering steps are required in \Cref{cec-cb}.

Furthermore, it is an open question whether there is an alternative
characterization of the $\id$-congruence of \Cref{cb-vs-id} that is
amenable to mechanization.

We have already implemented label derivation and bisimulation checking
in the borrowed context approach, see for instance~\cite{n:automatisch-bisim-gts},
and successfully applied it to a system with message-passing rules similar to
the ones given in \Cref{ex-cb-repr-unreliable}, however without conditions in
either the rules or the bisimulation relation.
Our aim is to also obtain an efficient implementation for the
scenario described in this paper. Note that our conditions subsume
first-order logic~\cite{CRS} and hence in order to come to terms with
the undecidability of implication we have to resort to simpler
conditions or use approximative methods.

Another natural question is whether our results can be stated in a
coalgebraic setting, since coalgebra provides a generic framework for
behavioural equivalences. We have already studied a much simplified
coalgebraic version of conditional systems (without considering
contextualization) in~\cite{abhkms:coalgebra-min-det}, using
coalgebras living in Kleisli categories. Reactive systems can also be
viewed as coalgebras (see~\cite{b:abstract-semantics-observable-contexts}).  However, a
combination of these features has not yet been considered as far as we
know.

Another direction for future research are further optimizations in
terms of the up-to context technique. Note, that even bisimulations
up-to context can still be infinite in size, which is somehow
unavoidable due to undecidability issues, so further optimizations
should be investigated. Furthermore we plan to integrate this method
with other kinds of up-to techniques such as up-to bisimilarity, which
should be easy due to the integration into the lattice-theoretical
framework.

\bigskip

\noindent\emph{Acknowledgements:} We would like to thank the
anonymous reviewers of the conference version of the paper for many useful
hints, in particular for suggesting to integrate our contribution into
the lattice-theoretical view of up-to functions.

%%%%%%%%%%%%%%%%%%%%%%%%%%%%%%%%%%%%%%%%%%%%%%%%%%%%%%%%%%%%%%%%%%%%%%%

\bibliographystyle{alpha} % this is the exact style that should have been set already, but if we do not specify it explicitly, the document will not build
% TODO: The author guide uses alpha and the actual default is "alpha" too, but the majority of recent lmcs publications use plainurl anyway.
\bibliography{special-issue}

\newcommand{\etalchar}[1]{$^{#1}$}
\begin{thebibliography}{CMR{\etalchar{+}}97}

\bibitem[ABH{\etalchar{+}}12]{abhkms:coalgebra-min-det}
Ji\v{r}{\'\i} Ad\'amek, Filippo Bonchi, Mathias H\"{u}lsbusch, Barbara
  K\"{o}nig, Stefan Milius, and Alexandra Silva.
\newblock A coalgebraic perspective on minimization and determinization.
\newblock In {\em Proc. of FOSSACS '12}, pages 58--73. Springer, 2012.
\newblock {LNCS/ARCoSS} 7213.

\bibitem[BBB02]{bbb:bisim-unification}
Paolo Baldan, Andrea Bracciali, and Roberto Bruni.
\newblock Bisimulation by unification.
\newblock In {\em Proc. of AMAST '02}, pages 254--270. Springer, 2002.
\newblock {LNCS} 2422.

\bibitem[BCHK11]{CRS}
H.J.~Sander Bruggink, Rapha{\"e}l Cauderlier, Mathias H{\"u}lsbusch, and
  Barbara K{\"o}nig.
\newblock Conditional reactive systems.
\newblock In {\em Proc. of FSTTCS~'11}, volume~13 of {\em {LIPIcs}}. Schloss
  Dagstuhl -- Leibniz Center for Informatics, 2011.

\bibitem[BKKS17]{BKKS17}
Harsh Beohar, Barbara K\"onig, Sebastian K\"upper, and Alexandra Silva.
\newblock Conditional transition systems with upgrades.
\newblock In {\em Proc. of TASE '17 (Theoretical Aspects of Software
  Engineering)}. IEEE Xplore, 2017.

\bibitem[BKM06]{bkm:saturated}
Filippo Bonchi, Barbara K\"onig, and Ugo Montanari.
\newblock Saturated semantics for reactive systems.
\newblock In {\em Proc. of LICS '06}, pages 69--80. IEEE, 2006.

\bibitem[Bon08]{b:abstract-semantics-observable-contexts}
Filippo Bonchi.
\newblock {\em Abstract Semantics by Observable Contexts}.
\newblock PhD thesis, Universit\`a degli Studi di Pisa, Dipartimento di
  Informatica, May 2008.

\bibitem[BW99]{BW:nCT}
Michael Barr and Charles Wells.
\newblock {\em Category Theory for Computing Science}.
\newblock Les Publications CMR, 1999.

\bibitem[CCP{\etalchar{+}}12]{DBLP:conf/icse/CordyCPSHL12}
Maxime Cordy, Andreas Classen, Gilles Perrouin, Pierre-Yves Schobbens, Patrick
  Heymans, and Axel Legay.
\newblock Simulation-based abstractions for software product-line model
  checking.
\newblock In {\em Proc. of ICSE '12}, pages 672--682. IEEE, 2012.

\bibitem[CMR{\etalchar{+}}97]{cmrehl:algebraic-approaches}
Andrea Corradini, Ugo Montanari, Francesca Rossi, Hartmut Ehrig, Reiko Heckel,
  and Michael L\"{o}we.
\newblock Algebraic approaches to graph transformation---part~{I}: Basic
  concepts and double pushout approach.
\newblock In G.~Rozenberg, editor, {\em Handbook of Graph Grammars and
  Computing by Graph Transformation, Vol.~1: Foundations}, chapter~3. World
  Scientific, 1997.

\bibitem[EK04]{EK04}
Hartmut Ehrig and Barbara K{\"o}nig.
\newblock Deriving bisimulation congruences in the {DPO} approach to graph
  rewriting.
\newblock In {\em Proc. of FOSSACS '04}, pages 151--166. Springer, 2004.
\newblock {LNCS} 2987.

\bibitem[EPS73]{DPO}
Hartmut Ehrig, Michael Pfender, and Hans~Jürgen Schneider.
\newblock Graph-grammars: An algebraic approach.
\newblock In {\em 14th Annual Symposium on Switching and Automata Theory (SWAT
  1973)}, pages 167--180, Oct 1973.

\bibitem[Fit02]{f:bisimulations-boolean-vectors}
Melvin Fitting.
\newblock Bisimulations and boolean vectors.
\newblock In {\em Advances in Modal Logic}, volume~4, pages 1--29. World
  Scientific Publishing, 2002.

\bibitem[FS90]{fs:categories-allegories}
Peter~J. Freyd and Andre Scedrov.
\newblock {\em Categories, Allegories}.
\newblock North-Holland, 1990.

\bibitem[HHT96]{NegativeAC}
Annegret Habel, Reiko Heckel, and Gabriele Taentzer.
\newblock Graph grammars with negative application conditions.
\newblock {\em Fundamenta Informaticae}, 26(3,4):287--313, December 1996.

\bibitem[HK12]{DBC-CRS}
Mathias H\"{u}lsbusch and Barbara K\"{o}nig.
\newblock Deriving bisimulation congruences for conditional reactive systems.
\newblock In {\em Proc. of FOSSACS '12}, pages 361--375. Springer, 2012.
\newblock {LNCS/ARCoSS} 7213.

\bibitem[HL95]{hl:symbolic-bisimulations}
Matthew Hennessy and Huimin Lin.
\newblock Symbolic bisimulations.
\newblock {\em Theoretical Computer Science}, 138(2):353--389, 1995.

\bibitem[HMP01]{RevisitedDPO}
Annegret Habel, Jürgen Müller, and Detlef Plump.
\newblock Double-pushout graph transformation revisited.
\newblock {\em Mathematical Structures in Computer Science}, 11(5):637--688,
  October 2001.

\bibitem[HP09]{PennemannNAC}
Annegret Habel and Karl-Heinz Pennemann.
\newblock Correctness of high-level transformation systems relative to nested
  conditions.
\newblock {\em Mathematical Structures in Computer Science}, 19(2):245--296,
  2009.

\bibitem[JM03]{jm:bigraphs}
Ole~H{\o}gh Jensen and Robin Milner.
\newblock Bigraphs and transitions.
\newblock In {\em Proc. of POPL 2003}, pages 38--49. ACM, 2003.

\bibitem[KSS05]{kss:labels-from-reductions}
Bartek Klin, Vladimiro Sassone, and Pawe{\l} Soboci\'{n}ski.
\newblock Labels from reductions: towards a general theory.
\newblock In {\em Proc. of CALCO '05}, pages 30--50. Springer, 2005.
\newblock LNCS 3629.

\bibitem[Lar86]{l:context-dependent-bisim}
Kim~Guldstrand Larsen.
\newblock {\em Context-Dependent Bisimulation between Processes}.
\newblock PhD thesis, University of Edinburgh, 1986.

\bibitem[LM00]{LM00}
James~J. Leifer and Robin Milner.
\newblock Deriving bisimulation congruences for reactive systems.
\newblock In {\em CONCUR 2000 --- Concurrency Theory: 11th International
  Conference University Park, PA, USA, August 22--25, 2000 Proceedings}, pages
  243--258. Springer Berlin Heidelberg, 2000.

\bibitem[LS05]{ls:adhesive-journal}
Stephen Lack and Pawe{\l} Soboci\'{n}ski.
\newblock Adhesive and quasiadhesive categories.
\newblock {\em RAIRO -- Theoretical Informatics and Applications},
  39(3):511--545, 2005.

\bibitem[Nol12]{n:automatisch-bisim-gts}
Dennis Nolte.
\newblock {Automatischer Nachweis von Bisimulations{\"a}quivalenzen bei
  Graphtransformationssystemen}.
\newblock Master's thesis, Universit\"at Duisburg-Essen, November 2012.

\bibitem[Pen09]{p:development-correct-gts}
Karl-Heinz Pennemann.
\newblock {\em Development of Correct Graph Transformation Systems}.
\newblock PhD thesis, Universit\"at Oldenburg, May 2009.

\bibitem[Pou07]{p:complete-lattices-up-to}
Damien Pous.
\newblock Complete lattices and up-to techniques.
\newblock In {\em Proc. of APLAS '07}, pages 351--366. Springer, 2007.
\newblock {LNCS} 4807.

\bibitem[PS11a]{pous_sangiorgi_2011}
Damien Pous and Davide Sangiorgi.
\newblock Enhancements of the bisimulation proof method.
\newblock In {\em Advanced Topics in Bisimulation and Coinduction}, Cambridge
  Tracts in Theoretical Computer Science, pages 233--289. Cambridge University
  Press, 2011.

\bibitem[PS11b]{ps:enhancements-coinductive}
Damien Pous and Davide Sangiorgi.
\newblock Enhancements of the coinductive proof method.
\newblock In Davide Sangiorgi and Jan Rutten, editors, {\em Advanced Topics in
  Bisimulation and Coinduction}. Cambridge University Press, 2011.

\bibitem[PS19]{ps:coinduction-enhancements-historical}
Damien Pous and Davide Sangiorgi.
\newblock Bisimulation and coinduction enhancements: A historical perspective.
\newblock {\em Formal Aspects of Computing}, 31(6):733--749, 2019.

\bibitem[Ren04]{r:representing-fol}
Arend Rensink.
\newblock Representing first-order logic using graphs.
\newblock In {\em Proc. of ICGT '04}, pages 319--335. Springer, 2004.
\newblock {LNCS} 3256.

\bibitem[San98]{sangiorgi}
Davide Sangiorgi.
\newblock On the bisimulation proof method.
\newblock {\em Mathematical Structures in Computer Science}, 8(5):447--479,
  1998.

\bibitem[Sob04]{Sobocinski}
Pawe{\l} Soboci{\'{n}}ski.
\newblock {\em Deriving process congruences from reaction rules}.
\newblock PhD thesis, University of Aarhus, 2004.

\bibitem[SS05]{ss:reactive-cospans}
Vladimiro Sassone and Pawe{\l} Soboci\'{n}ski.
\newblock Reactive systems over cospans.
\newblock In {\em Proc. of LICS '05}, pages 311--320. IEEE, 2005.

\bibitem[Tar55]{t:lattice-fixed-point}
Alfred Tarski.
\newblock A lattice-theoretical fixpoint theorem and its applications.
\newblock {\em Pacific Journal of Mathematics}, 5:285--309, 1955.

\end{thebibliography}

\appendix

\section{Supplementary Material on Adhesive Categories}%
\label{sec:adhesive}

We recall here the definition of adhesive
categories~\cite{ls:adhesive-journal}.  We do not provide any
introduction to basic categorical constructions such as products,
pullbacks and pushouts, instead referring the reader to Sections~5
and~9 of~\cite{BW:nCT}.

\begin{defi}[adhesive categories]%
  \label{de:adhesive}
  A category is called  %\cat{C} be a category. It is called
  \emph{adhesive} if
  \begin{itemize}
  \item it has pushouts along monos;
  \item it has pullbacks;
  \item pushouts along monos are \emph{Van Kampen} ({\sc vk}) squares.
  \end{itemize}
  \noindent Referring to \Cref{vks}, a {\sc VK} square is a pushout
  such as $(i)$, such that for each commuting cube as in~$(ii)$ having
  $(i)$ as bottom face and the back faces of which are pullbacks, the
  front faces are pullbacks if and only if the top face is a pushout.

  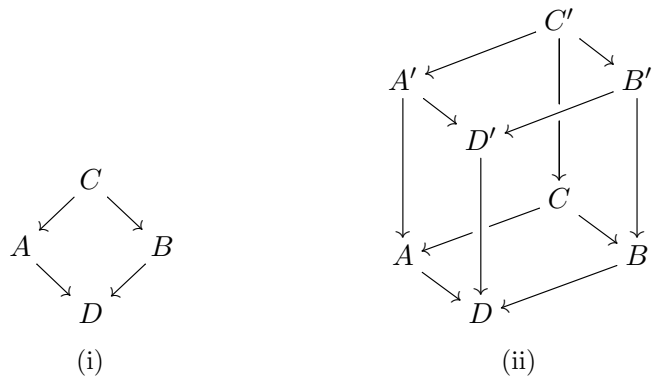
\begin{figure}[ht]
    \renewcommand\thesubfigure{\roman{subfigure}}
    \begin{subfigure}[b]{0.3\textwidth}
      %\def\objectstyle{\scriptstyle}\def\labelstyle{\scriptstyle}
      %\xymatrix@C=.4cm@R=.4cm{
      %                & \\
      %                & C \ar[dl]_m \ar[rd]^f \\
      %A \ar[rd]_g     &
      %                & B \ar[ld]^n \\
      %                & D
      %} & \hspace{2cm}
      %\def\objectstyle{\scriptstyle}\def\labelstyle{\scriptstyle}
      %\xymatrix@C=.4cm@R=.3cm{
      %                & & C' \ar[dll]_{m'} \ar[rd]^{f'} \ar[ddd] ^>>>>>>c
      %                                                  |!{[ddl];[dr]}{\hole} \\
      %A' \ar[rd]_{g'} \ar[ddd]_a
      %                & & & B' \ar[dll]_>>>>>>{n'} \ar[ddd]^b \\
      %                & D' \ar[ddd]_<<<<<<d &  \\
      %                & & C  \ar[dll]^<<<<<<{m} |!{[ul];[dl]}{\hole} \ar[rd]^{f}
      %                       \ar@{-}[uuu] |!{[uur];[ul]}{\hole}  \\
      %A  \ar[rd]_{g}  & & & B  \ar[dll]^{n} \\
      %                & D
      %}
      %\def\objectstyle{\scriptstyle}\def\labelstyle{\scriptstyle}
      \centerline{\xymatrix@C=.4cm@R=.4cm{
                      & \\
                      & C \ar[dl] \ar[rd] \\
      A \ar[rd]     &
                      & B \ar[ld] \\
                      & D
      }}
      \caption{}
    \end{subfigure}\hspace{1cm}
    \begin{subfigure}[b]{0.3\textwidth}
      \centerline{\xymatrix@C=.4cm@R=.27cm{
                      & & C' \ar[dll] \ar[rd] \ar[ddd]
                                                        |!{[ddl];[dr]}{\hole} \\
      A' \ar[rd] \ar[ddd]
                      & & & B' \ar[dll] \ar[ddd] \\
                      & D' \ar[ddd] &  \\
                      & & C  \ar[dll] |!{[ul];[dl]}{\hole} \ar[rd]
                             \ar@{-}[uuu] |!{[uur];[ul]}{\hole}  \\
      A  \ar[rd]  & & & B  \ar[dll] \\
                      & D
      }}
      \caption{}
    \end{subfigure}
    \caption{A pushout square $(i)$, and a commutative cube $(ii)$.}
    \protect\label{vks}
  \end{figure}
\end{defi}

The motivation for using adhesive categories is that they are a
suitable categorical framework for reasoning one rewriting of abstract
objects, in the spirit of graph rewriting.

\end{document}